%% file: main.tex
\documentclass[11pt]{article}
\usepackage[margin=2.5cm]{geometry}
\usepackage{graphicx,amsmath,amsfonts,amssymb,amsthm,bbold,dsfont,authblk,cite,xcolor,makecell,enumerate}
\usepackage{mathtools}
\usepackage{booktabs,array}
\usepackage[normalem]{ulem}

\usepackage{thmtools}
\declaretheorem[name=Theorem, numberwithin=section]{theorem}

\declaretheorem[name=Lemma, sibling=theorem]{lemma}
\declaretheorem[name=Definition, sibling=theorem]{definition}
\declaretheorem[name=Corollary, sibling=theorem]{corollary}

\declaretheorem[name=Remark, sibling=theorem]{remark}
\declaretheorem[name=Claim, sibling=theorem]{claim}
\usepackage[breakable]{tcolorbox}
\newtcolorbox[use counter*=theorem]{algorithm}[2][]{
    breakable,
    colback=black!1,
    colframe=black!55,
    colbacktitle=black!8,
    coltitle=black,
    fonttitle=\bfseries,
    title={Algorithm~\thetcbcounter\quad #2},
    title after break={Algorithm~\thetcbcounter\quad #2\ (continued)},
    arc=0pt,
    boxrule=0.5pt,
    boxsep=0pt,
    left=9pt,right=9pt,top=7pt,bottom=7pt,
    toptitle=5pt,bottomtitle=5pt,
    before skip=12pt,after skip=12pt,
    lines before break=4,
    #1
}

\usepackage[colorlinks]{hyperref}
\usepackage{braket}
\usepackage[mathscr]{euscript}
\usepackage{tikz}
\usepackage{pgfplots}
\pgfplotsset{compat=1.18}
\usepgfplotslibrary{groupplots}
\usepackage{pgfmath}
\usepackage{quantikz}
\input{figures/figure-style}
\usepackage{bm}
\usepackage[labelfont=bf]{caption}
\usepackage{xfrac}
\definecolor{NiceOrange}{rgb}{.9,0.55,0}
\definecolor{NicePurple}{rgb}{0.3,0.1,1}
\definecolor{DarkBlue}{rgb}{0.1,0.1,0.5}
\hypersetup{
    unicode=false,
    pdftoolbar=true,
    pdfmenubar=true,
    pdffitwindow=false,
    pdfnewwindow=true,
    colorlinks=true,
    linkcolor=DarkBlue,
    citecolor=NicePurple,
    filecolor=NiceOrange,
    urlcolor=DarkBlue,
    pdftitle={Quantum Algorithms for OPI Variants Beyond Locality and Classical Decodability},
    pdfauthor={Seyoon Ragavan and Noah Shutty},
}
\usepackage{listings}
\newcommand{\poly}{\mathsf{poly}}

\newcommand{\QFT}{\mathrm{QFT}}

\newcommand{\FF}{\mathbb{F}}
\newcommand{\F}{\mathbb{F}}
\newcommand{\ZZ}{\mathbb{Z}}
\newcommand{\EE}{\mathbb{E}}

\newcommand{\id}{\mathbb{I}}

\newcommand{\negl}{\mathsf{negl}}

\newcommand{\Tr}{\mathrm{Tr}}
\newcommand{\Var}{\mathrm{Var}}

\DeclareMathOperator{\im}{im}
\DeclareMathOperator{\spann}{span}
\DeclareMathOperator{\diag}{diag}
\DeclareMathOperator{\rank}{rank}
\def\one{\mathbf{1}}
\def\rmcsp{\mathsf{RMSupp}}
\def\rmsyn{\mathsf{RMSyndrome}}
\def\supp{\mathrm{supp}}
\def\calD{\mathcal{D}}

\def\cnt{\mathsf{Count}}
\def\cntexp{\mathsf{CntExp}}
\def\nuba{\nu, \{B_r\}, \{\alpha_r\}}

\def\class{\mathsf{class}}
\def\calK{\mathcal{K}}

\def\dec{\mathsf{dec}}
\def\prep{\mathsf{prep}}
\def\sparse{\mathsf{sparse}}
\def\wt{\mathsf{wt}}
\def\RR{\mathbb{R}}

\def\EFW{\mathsf{EFW}}
\def\src{\mathsf{src}}
\def\tar{\mathsf{tar}}
\def\hv{\mathbf{h}}
\def\dist{\mathsf{dist}}

\def\calB{\mathcal{B}}
\def\calT{\mathcal{T}}
\def\GRS{\mathrm{GRS}}
\def\RS{\mathrm{RS}}
\def\RM{\mathrm{RM}}
\def\FRS{\mathrm{FRS}}
\def\QDP{\mathsf{QDP}}
\def\CIP{\mathsf{CIP}}
\def\err{\mathsf{err}}

\newcommand{\norm}[1]{\left\| {#1} \right\|}

\def\cv{\mathbf{c}}
\def\bv{\mathbf{b}}
\def\gv{\mathbf{g}}
\def\dv{\mathbf{d}}
\def\ev{\mathbf{e}}
\def\xv{\mathbf{x}}
\def\zv{\mathbf{z}}
\def\av{\mathbf{a}}
\def\yv{\mathbf{y}}
\def\vv{\mathbf{v}}
\def\sv{\mathbf{s}}
\def\wv{\mathbf{w}}
\def\tv{\mathbf{t}}
\def\uv{\mathbf{u}}

\def\matA{\mathbf{A}}
\def\matB{\mathbf{B}}

\allowdisplaybreaks

\def\tpoiss{\mathsf{TiltPoiss}}

\newif\ifnotes
\notesfalse

\title{Quantum Algorithms for OPI Variants Beyond Locality and Classical Decodability}
\author[1,2]{Seyoon Ragavan\footnote{\url{sragavan@mit.edu}}}
\author[1]{Noah Shutty\footnote{\url{noahshutty@gmail.com}}}
\affil[1]{\small{\it{Google Quantum AI, Venice, CA 90291}}}
\affil[2]{\small{\it{Computer Science and Artificial Intelligence Laboratory, MIT, Cambridge, MA 02139}}}
\date{\today}

\begin{document}

\maketitle

\begin{abstract}
    Regev’s reduction is a quantum algorithmic framework for finding codewords satisfying nonlinear constraints by decoding the dual code. To date, applications that have not been dequantized have relied on efficient classical decoders and coordinate-wise constraints specifying a set of allowed values for each coordinate. We push past these two restrictions via two separate contributions.

    Our first contribution uses a quantum decoder to find solutions $\mathbf{y}\in(\mathbb{F}_q\setminus\{0\})^m$ to $\mathbf{B}\mathbf{y}=0$, where $\mathbf{B}\in\mathbb{F}_q^{n\times m}$. For fixed prime $q>2$, Chen, Liu, and Zhandry (Eurocrypt 2022) solve this problem for random matrices with $m=\Omega(n^2)$, a regime now covered by classical algorithms. We adapt their template to codes (spanned by the rows of $\mathbf{B}$) with a ``two-fold multiplication property'': the coordinate-wise products of pairs of codewords span a space of dimension much smaller than $m$. Under suitable distance conditions, this algebraic structure of $\mathbf{B}$ allows us to solve instances with $m\leq n^{2-\Omega(1)}$, beyond the established guarantees of the aforementioned classical algorithms. We also give an efficient classical algorithm under a stronger three-fold multiplication property, leaving regimes over prime fields where only the two-fold property holds as candidates for quantum advantage. One candidate instantiation takes the row span of $\mathbf{B}$ to be a dimension-$n$ Reed–Muller code punctured at $m$ random evaluation points. Equivalently, our algorithm finds a vector from the dual Reed--Muller code that is nonzero at the $m$ selected points and zero everywhere else, which can be thought of as solving a multivariate variant of optimal polynomial intersection (OPI). Because it uses a quantum decoder, our Reed--Muller result reaches regimes where no efficient classical decoder for these codes is known.

    Our second contribution retains classical decoding but allows global constraints on symbol frequencies. We study “histogram-local” constraints, which specify the allowed numbers of occurrences of each symbol; for example, a ``max-load'' constraint permits each symbol to appear at most $T$ times, for a constant $T\geq2$. For broad families of these constraints, we show that a uniformly random satisfying vector remains satisfying with constant probability after resampling a uniformly random coordinate. This stability provides sufficient Fourier mass at decodable Hamming weights, yielding efficient quantum algorithms for variants of OPI that combine coordinate-wise and histogram-local constraints. Adapting the methods of Yamakawa--Zhandry (FOCS 2022), we prove a quantum–classical separation for a variant of these problems relative to a random oracle.
\end{abstract}

\clearpage
\tableofcontents

\clearpage
\section{Introduction}\label{sec:introduction}

A \emph{code intersection problem} asks the algorithm to find a vector that belongs to both a linear code $C\subseteq\mathbb{F}_q^m$ (the solution space of a system of linear equations) and a set $A\subseteq\mathbb{F}_q^m$ specified by nonlinear constraints. We assume that membership in $A$ is efficiently verifiable, and are particularly interested in settings where vectors in $A$ are easy to construct. Finding a vector in the intersection can nevertheless be hard: the problem is $\mathsf{NP}$-hard even when $q=3$ and $A=A_1\times\cdots\times A_m$ is a Cartesian product. We ask when quantum algorithms can exploit the structure of the code and constraints to solve this problem more efficiently than classical algorithms. A motivating example is \emph{optimal polynomial intersection} (OPI), which---in the full-satisfaction setting that we consider---asks for a low-degree polynomial whose values lie in prescribed sets at all evaluation points.

Regev's reduction~\cite{Regev05}, also known as decoded quantum interferometry (DQI)~\cite{jordan2024optimization}, provides an exciting paradigm for obtaining such quantum algorithms. Given an efficiently preparable quantum state $\ket{\psi}$ supported on $A$, it reduces code intersection to a \emph{quantum decoding problem} for $C^\perp$: recover a uniformly random codeword $\mathbf{c}\in C^\perp$ from the state $X^{\mathbf{c}}\QFT^\dag\ket{\psi}$, where $X^{\mathbf{c}}$ adds $\mathbf{c}$ in the computational basis~\cite{chailloux2023quantum}. Yamakawa--Zhandry used this framework to obtain a provable quantum speedup for code intersection relative to a random oracle~\cite{yamakawa2024verifiable}. However, applications that have not been dequantized have so far shared two restrictions:

\begin{itemize}
    \item \textbf{Classical decodability:} the quantum decoding problem is solved by measuring the state in the standard basis to obtain a noisy codeword $\cv+\ev$ (for $\cv \in C^\perp$) and applying an efficient classical decoder for $C^\perp$ to the resulting noisy codeword. In our view, this reliance on classical decoders is the core reason why many successful applications of Regev's reduction work with algebraic codes based on low-degree polynomials.

    It is surprising and unnatural to us that the potential of an inherently quantum decoder has not been realized in the form of stronger algorithms. Earlier attempts to obtain speedups this way~\cite{clz21,shutty2026optimizationusinglocallyquantumdecoders} have been matched by classical algorithms~\cite{ivanyos2012nil2,ivanyosSantha2015diagonal,imran2024zerosum,kothari2025sis,shutty2026optimizationusinglocallyquantumdecoders}.

    \item \textbf{Coordinate-wise nonlinear constraints:} the nonlinear conditions prescribe an allowed set of values for each coordinate (we sometimes refer to such nonlinear constraints as ``local''). In this work, we focus on the setting where every condition is required to hold. Other works~\cite{jordan2024optimization,jo2026efficient,horinaga2026worstcase} have considered the setting of satisfying a large fraction of constraints, and Horinaga and Yamakawa~\cite{horinaga2026worstcase} also give worst-case algorithms for full satisfaction, but these still use coordinate-wise predicates and the underlying algorithms still work with a state $\ket{\psi}$ that is a product state or a Dicke-type state.
    When $\ket{\psi}$ is a product state, so is $\QFT^\dag\ket{\psi}$, so the errors in the decoding problem are independent across coordinates, which makes decoding guarantees easier to establish.
\end{itemize}

We give two applications of Regev's reduction that overcome these restrictions separately. Our first finds codewords with no zero coordinates using a quantum decoder that exploits a ``multiplication property'' of $C^\perp$. This yields new candidates for quantum advantage without relying on classical decodability. Our second retains classical decoding but permits global constraints on symbol frequencies, extending the methods of Yamakawa--Zhandry to obtain a provable quantum--classical separation relative to a random oracle.

\subsection{Contribution 1: Beyond Classical Decodability}

Our first contribution concerns a simple code intersection problem: given $\mathbf{B}\in\mathbb{F}_q^{n\times m}$, find
\[
    \mathbf{y}\in(\mathbb{F}_q\setminus\{0\})^m
    \qquad\text{such that}\qquad
    \mathbf{B}\mathbf{y}=0.
\]
Here $C=\ker\mathbf{B}$, and $C^\perp$ is the row span of $\mathbf{B}$. For fixed prime $q>2$, Chen, Liu, and Zhandry~\cite{clz21} gave a polynomial-time quantum algorithm for random $\matB$ with $m=\Omega(n^2)$. Although classical algorithms~\cite{ivanyos2012nil2,ivanyosSantha2015diagonal,imran2024zerosum,kothari2025sis} now cover this regime, we show that the quantum decoding approach of~\cite{clz21} can exploit additional structure to solve instances with polynomially fewer columns.

The structure we use is a ``two-fold multiplication property'': the coordinate-wise products of pairs of codewords in $C^\perp$ span a space of dimension much smaller than $m$. Together with suitable distance conditions, this property enables efficient quantum decoding without assuming an efficient classical decoder for $C^\perp$. Interestingly, we also show that a ``three-fold multiplication property'' enables efficient classical algorithms:

\begin{theorem}[Informal; see Theorem~\ref{thm:qalgomain}
    and Lemma~\ref{lem:classical-cube}]
\label{thm:informal-multiplication}
Fix an odd prime $q$ (the formal results, Theorem~\ref{thm:qalgomain} and Lemma~\ref{lem:classical-cube}, hold for every odd prime power). Let $\mathbf{B}\in\mathbb{F}_q^{n\times m}$ have linearly independent rows, and let $C=\ker\mathbf{B}$. Write $s_2$ and $s_3$ for the dimensions of the spaces spanned by all coordinate-wise products of two and three codewords in $C^\perp$, respectively, and $\dist(C^\perp)$ for its minimum (unnormalized) Hamming distance.

There are algorithms running in time $\poly(m)$ that find
$\mathbf{y}\in(\mathbb{F}_q\setminus\{0\})^m$ with
$\mathbf{B}\mathbf{y}=0$ under the following conditions:
\begin{itemize}
    \item \textbf{Quantum algorithm:}
    \[
        \dist(C^\perp)\log_q(q-1)\geq 2n+s_2+2.
    \]
    \item \textbf{Classical algorithm:}
    \[
        \dist(C^\perp)\geq n+s_3.
    \]
\end{itemize}
Each algorithm succeeds with probability at least $1/2$ over its internal randomness, for every matrix satisfying its respective condition.
\end{theorem}

\paragraph{A Reed--Muller instantiation.}
We instantiate this approach by taking the row span of $\matB$ to be a dimension-$n$ Reed--Muller code punctured at $m$ random evaluation points. These codes evaluate low-degree multivariate polynomials, so their coordinate-wise products are evaluations of polynomials of at most twice the original degree. For suitable parameters, this yields polynomial-time quantum algorithms in regimes with $m\leq n^{2-\Omega(1)}$. 
On the other hand, the classical guarantee of Theorem~\ref{thm:informal-multiplication} requires a space of polynomials of \emph{three times} the original degree to have dimension much less than $m$. This difference leaves intermediate parameter regimes over prime fields as candidates for quantum advantage; over extension fields, a simpler classical algorithm applies (Section~\ref{sec:classical}).

The Reed--Muller instantiation has a particularly appealing equivalent formulation. Extending $\mathbf{y}$ by zero outside the sampled evaluation points gives a word in the dual of the \emph{unpunctured} Reed--Muller code (which is itself a Reed--Muller code). Our task is therefore to find a dual codeword whose support is exactly a prescribed random set of $m$ points: it must be nonzero at every selected point and zero everywhere else. 
This is a variant of the ``multivariate OPI'' problem considered in~\cite{jordan2024optimization}, whose algorithm requires an efficient classical decoder for $C^\perp$; as we explain next, our algorithm works in regimes where no such decoder is known. 
We depict these formulations in Figure~\ref{fig:rm-support-size}.

\input{figures/rm-support-size}

Our Reed--Muller instantiation makes the departure from classical decodability concrete. 
Solving the quantum decoding problem by measuring in the standard basis and then decoding classically would require an efficient classical decoder for $C^\perp$ against independent symbol errors (Equation~\eqref{eq:qdpoverview}). In our instantiation, $q$ is a constant and the code consists of polynomials of degree $\Theta(v)$ in $v$ variables, so a classical decoder would have to recover such a polynomial from its values at the sampled points, each corrupted with probability $1/q=\Theta(1)$. 
We are not aware of an efficient classical decoder for this task in this regime. Even for unpunctured binary Reed--Muller codes, extending efficient decoding to degree $\Theta(v)$ at a constant random-error rate is a recognized open problem, despite substantial progress in the sublinear-degree regime~\cite{saptharishi2017decoding,abbe2020rm}. Our quantum decoder uses the two-fold multiplication property to recover the codeword directly from the quantum state, bypassing this classical decoding challenge and allowing Regev's reduction to proceed.

\subsection{Contribution 2: Beyond Product Constraints}\label{sec:introcontribution2}
Our second contribution permits constraints on symbol frequencies after applying an arbitrary permutation of the alphabet in each coordinate. Fix permutations $\pi_1,\ldots,\pi_m$ of $\mathbb{F}_q$ and, for each symbol $a\in\mathbb{F}_q$, a set $\calB_a\subseteq\{0,\ldots,m\}$ of allowed occurrence counts. A word $\mathbf{x}$ belongs to our \emph{histogram-local} target set $A$ if its relabeling $\mathbf{y}=(\pi_1(x_1),\ldots,\pi_m(x_m))$ satisfies
\[
    c_a(\mathbf{y})
    :=\bigl|\{i\in[m]:y_i=a\}\bigr|
    \in \calB_a
    \qquad\text{for every }a\in\mathbb{F}_q.
\]
We call this target set ``histogram-local'' because it constrains each entry $c_a(\yv)$ of the histogram of $\yv$ separately.
We seek a word $\mathbf{x}\in C\cap A$; the code constraint applies before relabeling by the $\pi_i$'s, and the histogram constraints afterward. 

For example, this captures the full-satisfaction version of OPI with equally sized allowed sets $A_1,\ldots,A_m$: take $C$ to be a Reed--Solomon code, choose $\pi_i$ such that $\pi_i(A_i)=A_1$, and forbid symbols $a \notin A_1$ by setting their set of allowed counts $\calB_a$ to $\{0\}$, leaving $\calB_a = [0, m]$ for $a \in A_1$. The resulting conditions are exactly $x_i\in A_i$ for every coordinate. 
More general choices of $\calB_a$ take us beyond OPI by imposing global restrictions on the relabeled word. For example, we can include a \emph{max-load} constraint by requiring that after relabeling, no symbol should appear more than $3$ times.

For a concrete example with code length $m=q-1$, partition the alphabet into three classes in proportions $1/2$, $1/3$, and $1/6$, up to rounding, and impose the following requirements:
\[
\begin{array}{c|ccc}
    \text{Fraction of the alphabet}
        & 1/2 & 1/3 & 1/6 \\ \hline
    \text{Allowed occurrence counts for each symbol}
        & \{0\} & \{1,2\} & \{3,4\}
\end{array}
\]
The first class gives the familiar OPI condition that half the symbols are forbidden in each coordinate. The other two classes prescribe how often each of the remaining symbols must occur. 
We illustrate this example in Figure~\ref{fig:histogram-csp}.
For this example, we show the following theorem. Its random-oracle part adapts the template of Yamakawa--Zhandry~\cite{yamakawa2024verifiable}.

\input{figures/histogram-csp}

\begin{theorem}[Informal; see Theorem~\ref{thm:plainrs}
    and Corollaries~\ref{cor:heterogeneous-separation} and~\ref{cor:heterogeneous-plain}]
For suitable families of Reed--Solomon codes $C$ of asymptotic rate $7/8$, the following hold. All probabilities are averaged over the random input permutations or oracle functions and the respective algorithm's internal randomness.
\begin{itemize}
    \item \textbf{Plain model.}
    Over prime fields, with independent uniformly random coordinate permutations $\pi_i$ given by their complete truth tables, a quantum algorithm running in time $\poly(q)$ finds a codeword whose relabeling satisfies all three classes of constraints with probability at least $1/\poly(q)$. This uses ordinary Reed--Solomon codes without folding.

    \item \textbf{Random-oracle model.}
    In a variant where independent random functions are applied to blocks of the codeword before computing the histogram, with a security parameter $\lambda$ and codes of length $\Theta(\lambda^2)$, a quantum algorithm running in time $\poly(\lambda)$ satisfies all three classes of constraints with constant probability. Every classical algorithm making $\poly(\lambda)$ oracle queries has success probability $2^{-\Omega(\lambda)}$, even with unbounded computation and even if any two of the three constraint classes are removed.
\end{itemize}
\end{theorem}

Our result shows that we can add additional constraints to OPI to compound the classical hardness, while retaining the quantum easiness. Indeed, we may be living in a (currently hypothetical) world where there is an efficient classical algorithm for OPI that still fails to satisfy the second and third classes of nonlinear constraints. 
Our oracle separation provides evidence that these extra classes of nonlinear constraints present obstacles in their own right.

\section{Technical Overview}

Before we dive into our contributions, we take a moment to recall the quantum decoding problem ($\QDP$) and what guarantee we need for a $\QDP$ algorithm to be usable in Regev's reduction. Our work and the remainder of this technical overview will then treat Regev's reduction as a black box.
Consider a state $\ket{\psi_\err} = \sum_{\ev \in \FF_q^m} \alpha(\ev) \ket{\ev}$. Then $\QDP_{C^\perp, \ket{\psi_\err}}$ samples a uniform codeword $\cv \gets C^\perp$ and provides the algorithm the state
$$X^\cv \ket{\psi_\err} = \sum_{\ev \in \FF_q^m} \alpha(\ev) \ket{\cv+\ev}$$
as input. The algorithm succeeds if and only if it correctly outputs $\cv$ in its entirety. 
Its \emph{success probability} $P_\dec$ is averaged over any internal randomness and the randomness of $\cv$.

Given an efficiently preparable state $\ket{\psi}$ supported on $A$, Regev's reduction~\cite{Regev05} reduces the task of finding a word in $C \cap A$ to solving $\QDP_{C^\perp, \QFT^\dag \ket{\psi}}$. 
For worst-case sets $A$, the reduction (Theorem~\ref{thm:worstcasecip}) succeeds only with probability about $1-2|C^\perp|(1-P_\dec)$, because the amplitudes of decoding errors could add constructively~\cite[Sections~6.1 and~6.3.1]{chailloux2023quantum}; so even $P_\dec$ very close to 1 need not suffice. 

However, recent average-case versions of the reduction significantly relax this requirement~\cite{chailloux2024quantum,chaillouxhermouet,chailloux2025opixsoftdecoders,shutty2026optimizationusinglocallyquantumdecoders}. If the target set $A$ is drawn from a translation-invariant distribution---for example, a fixed set $A_0$ translated by a uniformly random vector $\Delta \gets \FF_q^m$---then averaging over translations controls the coherent accumulation of decoding errors. We restate this average-case Regev reduction in Theorem~\ref{thm:chailloux}; informally, we can solve the code intersection problem with probability $\geq P_\dec$ (where ``probability'' here is taken over the randomness of $A$ as well). With this relaxation, inverse-polynomial $P_\dec$ therefore suffices. Going forward, we will therefore focus entirely on algorithms for $\QDP$ and we will be content with any inverse-polynomial success probability.\footnote{A careful reader might notice that our multiplication-property/Reed--Muller result does not rely on any translation-invariance of the underlying distribution. This is because in this case, our quantum decoder's failure probability is exponentially small and we can actually prove that the decoding failures can be tolerated, via the worst-case Regev reduction which we restate in Theorem~\ref{thm:worstcasecip}. This distinction is not crucial to our main ideas so we ignore it in this technical overview.}

\subsection{Algorithms Exploiting Multiplication Properties}
\label{sec:overview-quantum-decoder}

We explain how the two-fold and three-fold multiplication properties give the quantum and classical algorithms of Theorem~\ref{thm:informal-multiplication}, respectively, and compare their guarantees for Reed--Muller codes.

Recall that in this setting we have a matrix $\matB \in \FF_q^{n \times m}$, the code $C$ is $\ker \matB$, and we want to find a vector in $C \cap A$ where $A = (\FF_q \setminus \{0\})^m$. 
Following Chen--Liu--Zhandry~\cite{clz21}, we do this using Regev's reduction with the uniform superposition $\ket{A}$ over $A$. Hence our $\QDP$ input for a hidden $\cv \gets C^\perp$ has the form
\begin{equation}\label{eq:qdpoverview}
    X^\cv\QFT^\dag\ket{A}
    =\bigotimes_{i=1}^m\ket{\phi_{c_i}},
    \qquad
    \ket{\phi_c}
    =\sqrt{\frac{q-1}{q}}\ket{c}
      -\frac{1}{\sqrt{q(q-1)}}\sum_{a\neq c}\ket{a}.
\end{equation}
Measuring each register in the standard basis would give independent symbol errors with probability $1/q$.
Instead, we measure each register in a way that yields a two-element list always containing $c_i$, and then recover the codeword from these $m$ lists by linear algebra.

\paragraph{An unambiguous measurement that returns a pair.}
To execute the first step and obtain a pair containing $c_i$ for each $i \in [m]$, we use the filtered or unambiguous measurements of prior work~\cite{clz21,chailloux2023quantum,buzet2025finegrained,shutty2026optimizationusinglocallyquantumdecoders,dalai2026zeroerror}, whose output list always contains the correct value. 
These measurements work out of the box, but we give a simple self-contained formulation of the measurement below for completeness (Lemma~\ref{lem:usd}). We simply use the isometry
\[
    \ket{\phi_c}\longmapsto
    \frac{1}{\sqrt{2(q-1)}}
    \sum_{a\neq c}\bigl(\ket{c,a}-\ket{a,c}\bigr).
\]
This map exists because, on both sides of the display, the inner product of the states indexed by $c$ and $c'$ is $1$ when $c=c'$ and $-1/(q-1)$ otherwise; two families of states with the same inner products are related by an isometry. 
Measuring the two output registers and forgetting their order gives $\{c,c'\}$, where $c'$ is uniform on $\FF_q\setminus\{c\}$.

Applied separately to the $m$ registers, this produces pairs $\{c_i,c_i'\}$ whose false entries $c_i'$ are independent, conditional on $\cv$.
The measurement does not identify which entry is correct.
Its benefit is that every coordinate now gives an exact constraint on the hidden codeword.

\paragraph{Recovering the codeword using Arora--Ge relinearization~\cite{arorage}.}
Label the observed pairs $\{a_i,a_i'\}$, in any order. 
Then the hidden codeword $\cv \in C^\perp$ necessarily satisfies:
\begin{equation}\label{eq:agproduct}
    (c_i-a_i)(c_i-a_i')=0
    \qquad\text{for every }i\in[m].
\end{equation}
Recall that we can write $\cv = \matB^\top \xv$ for some $\xv \in \FF_q^n$. Then the observation by~\cite{clz21} is that we have $m$ quadratic equations in $n$ variables (the coordinates of $\xv$). When $m \geq \Omega(n^2)$, these can be solved by creating auxiliary variables for each quadratic monomial $x_ix_j$ in the $n$ variables and solving the resulting linear system. This is the idea underlying~\cite{arorage} that was then utilized by~\cite{clz21}. 
Unfortunately, the same $m \geq \Omega(n^2)$ requirement enables matching classical algorithms~\cite{ivanyos2012nil2,ivanyosSantha2015diagonal,imran2024zerosum,kothari2025sis}.

What we have discussed up to now is the algorithm of~\cite{clz21}; next, we will see how we can handle $m \leq n^{2-\Omega(1)}$ for suitable $\matB$.

\paragraph{Our idea: reducing auxiliary variables using a two-fold multiplication property.} 
The high-level idea is that if coordinate-wise products of pairs of codewords from $C^\perp$ fall in a low-dimensional space, we need far fewer than $\Theta(n^2)$ new unknowns for the quadratic terms in Equation~\eqref{eq:agproduct}. To see this, write the $m$ equations as one vector equation, where $\odot$ denotes coordinate-wise multiplication:
\begin{equation}
    (\cv - \av) \odot (\cv-\av') = 0 \Leftrightarrow \cv \odot \cv - (\av+\av') \odot \cv + \av \odot \av' = 0.
\end{equation}
Now, let
\[
    (C^\perp)^{*2}=\spann\{\uv\odot\vv:\uv,\vv\in C^\perp\},
    \qquad s_2=\dim (C^\perp)^{*2},
\]
Since $\cv\odot\cv\in(C^\perp)^{*2}$, we replace $\cv\odot\cv$ by an unknown $\wv\in(C^\perp)^{*2}$ and $\cv$ by an unknown $\vv\in C^\perp$, and solve the linear system
\[
    \wv-(\av+\av')\odot\vv+\av\odot\av'=\mathbf0,
    \qquad \vv\in C^\perp,\quad \wv\in (C^\perp)^{*2}.
\]
After choosing bases for $C^\perp$ and $(C^\perp)^{*2}$ (the latter of which can be computed from the products of pairs of rows of $\matB$), there are only $n+s_2$ unknown coefficients, which could potentially be much smaller than $m$ even if $m = n^{2-\Omega(1)}$. The ``two-fold multiplication property'' asks that $s_2$ be small. The decoder fails only if the system has a second solution. Such a solution differs from $(\cv,\cv\odot\cv)$ by some nonzero $\Delta_\vv\in C^\perp$ and some $\Delta_\wv\in(C^\perp)^{*2}$ with $\Delta_\wv=(\av+\av')\odot\Delta_\vv$. At each of the at least $\dist(C^\perp)$ coordinates where $\Delta_\vv$ is nonzero, this fixes $a_i+a_i'=c_i+c_i'$, which is uniform on $q-1$ values, independently across coordinates. A union bound over the at most $q^{n+s_2}$ pairs $(\Delta_\vv,\Delta_\wv)$ gives (Lemma~\ref{lem:arorage})
\[
    \Pr[\text{decoder fails}]
    \leq
    \frac{q^{n+s_2}}{(q-1)^{\dist(C^\perp)}},
\]
which essentially yields the quantum side of Theorem~\ref{thm:informal-multiplication}. (The threshold there is $2n+s_2+2$ instead of $n+s_2$ in order to absorb decoding failures, but we will not worry about this detail here.)

The above bound turns out to be loose. For a fixed $\Delta_\vv$, the vector $\Delta_\wv$ must vanish wherever $\Delta_\vv$ does, so the union bound needs only the subspace of $(C^\perp)^{*2}$ vanishing on the zeros of $\Delta_\vv$, and the probability is at most $(q-1)^{-\wt(\Delta_\vv)}$. Summing over $\Delta_\vv$ gives the tighter bound of Theorem~\ref{thm:qalgomain}, which improves our Reed--Muller results.

\paragraph{A classical algorithm using a three-fold multiplication property.}
The classical algorithm uses multiplication structure to construct a vector with no zero coordinates whose syndrome depends linearly on a set of unknowns. 
To accomplish this, there are two main ideas:
\begin{itemize}
    \item Restrict $\yv$ to a degree-2 ansatz that is nonzero in every coordinate \emph{by design}. Each coordinate of the ansatz is a nonzero multiple of $x^2-\xi$, where $x\in\FF_q$ varies with the unknowns and $\xi$ is a nonsquare, so it never vanishes.

    \item This then gives a system of seemingly degree-2 equations for which we need to find a solution. However, we will carefully set up the ansatz so that the degree-2 terms vanish and we are left with a linear system of equations that can easily be solved. Setting up the ansatz in this way uses a vector $\hv$ orthogonal to $(C^\perp)^{*3}$, and the linear system is solvable when $\hv$ has large enough support; this holds for a random such $\hv$ when $s_3$ is ``small.''
\end{itemize}

Related reductions from quadratic to linear equations occur in the Unbalanced Oil and Vinegar signature scheme~\cite{kipnis1999unbalanced}, collision attacks on quadratic hash functions by Ding and Yang~\cite{ding2008multivariates} and Applebaum et al.~\cite{applebaum2017lowcomplexity}, and Zhandry's construction of affine spaces of colliding inputs for quantum lightning~\cite{zhandry2019lightning}.

In more detail (Algorithm~\ref{alg:classical-cube} and Lemma~\ref{lem:classical-cube}), assume that $q$ is odd and that $\matB$ has linearly independent rows.
Let $U=(C^\perp)^{*3}$ be the space spanned by all coordinate-wise products of three codewords in $C^\perp$, with $s_3=\dim U$.
We compute $U$ from triple products of rows of $\matB$ and sample $\hv$ uniformly from $U^\perp$.
Write $T=\supp(\hv)$ and split $\matB$ into the columns $\matB'$ on $T$ and $\matB''$ outside $T$.
We set the desired output to $1$ outside $T$, where the ansatz below (which has the factor $\hv$) would vanish, so it remains to find a vector $\yv'$ with no zero coordinates satisfying
$$\matB'\yv'=-\matB''\one.$$
If $\rank\matB'<n$, we resample $\hv$. The requirement that $\dist(C^\perp) \geq n+s_3$ is in order to ensure that $\rank \matB' \geq n$ with good probability; once we successfully sample $\hv$ such that $\rank \matB' \geq n$, the remainder of the algorithm is deterministic and does not rely on this condition.

On $T$, the vector $\hv'=\hv|_T$ has no zero entries.
Choose a nonsquare $\xi\in\FF_q^\times$ and consider the quadratic ansatz
\[
    \yv'(\xv)=\hv'\odot
    \left[\left((\matB')^\top\xv+\tv\right)^{\odot 2}-\xi\one\right],
    \qquad \xv\in\FF_q^n,
\]
where $\tv\in\FF_q^T$ is a shift to be chosen. It is clear that this ansatz ensures that $\yv'$ is nonzero in every coordinate. 
Moreover, $\matB' \yv'(\xv)$ appears to have a degree-2 dependence on $\xv$ but on closer inspection it is actually linear! This is because the degree-2 part of each entry of $\matB'\yv'(\xv)$ is the inner product of $\hv$, which is zero outside $T$, with a product of three codewords of $C^\perp$: a row of $\matB$ and two copies of $\matB^\top\xv$. Since $\hv\in U^\perp$, this inner product is zero.
Consequently,
\[
    \matB'\yv'(\xv)
    =2\matB'\diag(\hv'\odot\tv)(\matB')^\top\xv
      +\matB'\yv'(\mathbf0).
\]
Thus we can enforce the desired syndrome by solving a linear system, provided its coefficient matrix is invertible.

We can select the shift $\tv$ to ensure this.
Choose $n$ linearly independent columns of $\matB'$, forming the submatrix $\matB_I$, with index set $I\subseteq T$, and set $t_i=h_i^{-1}$ for $i\in I$ and $t_i=0$ for $i\in T\setminus I$.
Then
\[
    \matB'\diag(\hv'\odot\tv)(\matB')^\top
    =\matB_I\matB_I^\top
\]
is invertible, and therefore we can solve the linear system, establishing the classical part of Theorem~\ref{thm:informal-multiplication}.
Thus the classical algorithm uses three-fold orthogonality to cancel quadratic terms, while the quantum decoder uses the two-fold property to reduce the number of auxiliary variables.

\paragraph{Application to randomly punctured Reed--Muller codes.}
The Reed--Muller settings that we consider can be specified by a prime field size $q > 2$ and some constant parameters $\rho \in (0, q-1)$ and $c > 1$. We then let the number of variables $v$ grow.
Take the rows of $\matB$ to be the evaluations of all reduced monomials of total degree at most $\lfloor\rho v\rfloor$ on a uniformly random $m$-element subset $S\subseteq\FF_q^v$.
Here \emph{reduced} means that each individual degree is at most $q-1$. The numbers of rows and sampled points are
\[
    n=\cnt_q(v,\lfloor\rho v\rfloor),
    \qquad m=\min(\lfloor n^c\rfloor,q^v),\text{ where}
\]
\[
        \cnt_q(v,r)
        :=\#\{\boldsymbol{\alpha}\in\{0,\ldots,q-1\}^v:
        \|\boldsymbol{\alpha}\|_1\leq r\}.
    \]
Our task is still to find $\yv$ with no zero entries and $\matB\yv=0$.
As in Figure~\ref{fig:rm-support-size}, extending $\yv$ by zero outside $S$ gives the equivalent task of finding a dual Reed--Muller codeword whose support is exactly $S$. 

In these terms, our task is to identify sufficient regimes of $(q, \rho, c)$ in which there are efficient quantum (or classical) algorithms. 
To do this, we use the counting exponent $\cntexp_q(\rho)$, which records how the number of these monomials grows:
\[
    n=q^{v(\cntexp_q(\rho)+o(1))}.
\]
Definition~\ref{def:count-exponent} gives its explicit formula, and Lemma~\ref{lem:count-asymptotics} gives this estimate.
Multiplying two or three of these polynomials at most doubles or triples the degree, so $\cntexp_q(2\rho)$ and $\cntexp_q(3\rho)$ control the dimensions of the respective product spaces.
Our tightest results give the following sufficient conditions (Theorems~\ref{thm:rm-pair-dqi} and~\ref{thm:cubecond}):
\begin{equation}\label{eq:rm-overview-thresholds}
\begin{aligned}
    \text{Quantum:}\qquad
    c&>\frac{\max\{\cntexp_q(2\rho),\,\cntexp_q(\rho)+\rho/(q-1)\}}
                 {\cntexp_q(\rho)},\\
    \text{Classical:}\qquad
    c&>\frac{\max\{\cntexp_q(3\rho),\,\cntexp_q(\rho)+\rho/(q-1)\}}
                 {\cntexp_q(\rho)}.
\end{aligned}
\end{equation}
Each condition gives an algorithm running in $\poly(m)$ time which, with probability $1-o(1)$ over the randomness of $S$, outputs a solution with probability $1-o(1)$ as $v\to\infty$.
The sole difference is the doubled versus tripled degree in the first term of each maximum.
The shared second term ensures that the sampled code $C^\perp$ has sufficiently large distance: the relative distance before puncturing is at least $q^{-\lfloor\rho v\rfloor/(q-1)}$, and this term requires $m q^{-\lfloor\rho v\rfloor/(q-1)}$ to grow faster than $n$ in the nontrivial regime.

\input{figures/rm-phase-diagrams}

Figure~\ref{fig:rm-phase-diagrams} plots these guarantees for $q=3,5,7$.
Above the dashed curve $c\cntexp_q(\rho)=1$, we eventually have $m=q^v$: the sampled set is the entire grid and the all-ones vector is a solution.
This trivial region is gray.
Since $\cntexp_q(\rho)=1$ for $\rho\geq(q-1)/2$, larger $\rho$ than those plotted also give the trivial case for every fixed $c>1$.

The blue region without green hatching is covered by our quantum guarantee but not by our three-fold classical guarantee.
Its portion with $c<2$ also lies below the quadratic scale of the prior classical algorithms discussed in the introduction~\cite{ivanyos2012nil2,ivanyosSantha2015diagonal,imran2024zerosum,kothari2025sis}, giving the candidate regimes for quantum advantage.
All these boundaries describe sufficient conditions; neither the blue-only region nor the unshaded region asserts classical hardness.

\subsection{Histogram-Local Constraints}\label{sec:overview-histogram}

We retain the notation from Section~\ref{sec:introcontribution2}. As in Part~\ref{part:rs}, $n$ denotes an alphabet size and $\rho$ a code rate in this subsection; they are unrelated to $n$ and $\rho$ in Section~\ref{sec:overview-quantum-decoder}. We first discuss the plain-model setting. Let
$$\calT=\{\yv\in\F_q^m:c_a(\yv)\in \calB_a\text{ for every }a\in\F_q\}$$
be the set of all $\yv \in \F_q^m$ whose histograms satisfy the given constraints, and let
$$A=\{\xv \in \FF_q^m:(\pi_1(x_1),\ldots,\pi_m(x_m))\in \calT\}$$
be the set of all words that map to a word in $\calT$ under the relabeling given by the $\pi_i$'s. We need $\ket{A}$ to be efficiently preparable. Theorem~\ref{thm:Aprep} prepares $\ket{\calT}$ using dynamic programming and Grover--Rudolph state preparation~\cite{groverRudolph2002}, and applying $\pi_i^{-1}$ to each coordinate $i$ then gives $\ket{A}$. For the quantum algorithm, it remains to solve $\QDP_{C^\perp, \QFT^\dag \ket{A}}$.

\paragraph{From list decoding to replacement stability.}

Assuming $\calT$ is nonempty, our goal is to show that the Fourier transform of $\ket{A}$ has noticeable mass within the radius of an efficient classical decoder for $C^\perp$ (we then solve $\QDP_{C^\perp, \QFT^\dag \ket{A}}$ by measuring in the standard basis and running that decoder). The independent uniformly random permutations $\{\pi_i\}$ make the distribution of $A$ translation-invariant, ensuring the applicability of the average-case Regev reduction. 

A useful observation---to our knowledge, first leveraged by~\cite{chailloux2024quantum}---is that even a \emph{list-}decoder for $C^\perp$ can be compiled into a $\QDP$ algorithm: the algorithm can simply run a list-decoding subroutine then output a uniformly random element of the list as its guess!\footnote{Previously,~\cite{yamakawa2024verifiable} also worked in the list-decoding regime, but they showed that their error distribution leads to list size 1 with overwhelming probability; we also want to handle cases where such a claim is hard to check or false.}
In more detail, let $\ev$ be the outcome of measuring $\QFT^\dag \ket{A}$ in the standard basis, and let $W=\wt(\ev)/m$. If a list decoder corrects up to a $\zeta$ fraction of errors with lists of size at most $L$, a uniformly random guess from its output list succeeds with probability at least
\[
    P_\dec\geq\frac{\Pr[W\leq\zeta]}{L}.
\]
Thus constant mass below $\zeta$ and a polynomial list bound are enough for the average-case Regev reduction (we formally state this in Corollary~\ref{cor:ctlistdecode}). By Markov's inequality (see below), it therefore suffices to bound the expected relative Hamming weight of $\ev$; no concentration estimate is needed.

To obtain that bound, consider the walk on $\F_q^m$ that replaces a uniformly chosen coordinate by an independent uniform field element, and let $\calK$ be its transition operator on the computational basis. This walk diagonalizes in the Fourier basis: averaging coordinate $i$ leaves $\QFT\ket{\ev}$ unchanged if $e_i=0$ and annihilates it otherwise, by character orthogonality. Averaging over the choice of $i$ therefore gives
\[
    \calK\QFT\ket{\ev}
    =\left(1-\frac{\wt(\ev)}m\right)\QFT\ket{\ev}.
\]
This is the standard relation between coordinate averaging and Fourier support~\cite[Corollary~8.21 and Proposition~8.23]{odonnell2014analysis}. Now choose $\xv$ uniformly from $A$ and take one step to obtain $\xv'$. Expanding $\ket{A}$ in the Fourier basis gives
\begin{align*}
    \Pr[\xv'\in A]
    &=\braket{A|\calK|A}\\
    &=\sum_{\ev\in\F_q^m}
        |\braket{\ev|\QFT^\dag|A}|^2
        \left(1-\frac{\wt(\ev)}m\right)
      =1-\EE[W].
\end{align*}
Thus $\EE[W]$ is exactly the one-step exit probability (Corollary~\ref{cor:fourier-exit}). Under the coordinate permutations, this experiment becomes uniform coordinate replacement in $\calT$ (Lemma~\ref{lem:foldpreservesexit}): a uniform element of $A$ maps to a uniform element of $\calT$, and a fresh uniform field element remains uniform after applying any $\pi_i$. If $p_\mathsf{stay}$, the \emph{replacement stability} of $\calT$, is the probability that this replacement stays in $\calT$, then
\[
    \EE[W]=1-p_{\mathsf{stay}},\text{ so Markov's inequality yields }
    \Pr[W\leq\zeta]\geq 1-\frac{1-p_\mathsf{stay}}{\zeta}.
\]
Therefore it suffices that $1-p_\mathsf{stay}\leq\zeta-\delta$ for a constant $\delta>0$.

\paragraph{Computing the replacement stability.}
Replacing a symbol $a$ (the ``source symbol'') by a different symbol $b$ (the ``target symbol'') changes only two counts: $c_a$ decreases by one and $c_b$ increases by one. A slight hindrance is that the counts of a uniform word in $\calT$ are dependent, since they sum to $m$. Moreover, histograms are not equally likely: a legal histogram $(c_a)_{a\in\F_q}$ occurs with probability proportional to
\[
    \binom{m}{c_0, \ldots, c_{q-1}} = \frac{m!}{\prod_{a\in\F_q}c_a!}.
\]
We use Poissonization~\cite{kolchin1978random} to represent this distribution by independent tilted Poisson variables conditioned on their sum. For a tilting parameter $\nu>0$, take independent $\operatorname{Poisson}(\nu)$ variables and condition each to lie in its allowed set. The resulting loads $K_a$ (one per symbol $a$, modeling its count $c_a$) satisfy
\begin{equation}\label{eq:kalaw}
    \Pr[K_a=k]
    =\begin{cases}
        \frac{\nu^k/k!}{\sum_{\ell\in \calB_a}\nu^\ell/\ell!} \text{, if }k\in \calB_a; \\
    0\text{, otherwise.}
    \end{cases}
\end{equation}
Conditioning further on $\sum_a K_a=m$ gives each legal histogram $(c_a)_{a\in\F_q}$ probability proportional to $\nu^m/\prod_a c_a!$. Since $\nu^m$ is common to every such histogram, this is exactly the histogram distribution of a uniform word in $\calT$, for every $\nu>0$.

The representation is exact, but conditioning makes the loads dependent. To navigate this, we choose $\nu$ so that the unconditioned sum has mean $m+O(1)$ (i.e., close to the sum $m$ being conditioned on).  A local central limit theorem then lets us approximately remove the conditioning on the sum being $m$ when computing the source and target survival probabilities (Theorem~\ref{thm:maininfinite}).

We illustrate this with the three-class example from Section~\ref{sec:introcontribution2}, where $m=q-1$. Write $(\alpha_1,\alpha_2,\alpha_3)=(1/2,1/3,1/6)$ for the class proportions and $B_1=\{0\}$, $B_2=\{1,2\}$, $B_3=\{3,4\}$ for their allowed load sets; class $r$ contains $\alpha_rq+O(1)$ symbols. Let $\mu_r$ be the mean of the load distribution with allowed set $B_r$ as specified in Equation~\eqref{eq:kalaw}. We choose $\nu$ so that the mean load per symbol is one; a direct computation yields:
\begin{align*}
    \nu &= \sqrt{5} - 1, \\
    (\mu_1, \mu_2, \mu_3) &= \left(0, \frac{2(1+\nu)}{2+\nu}, \frac{4(3+\nu)}{4+\nu}\right) \\
    &= \left(0, \frac{5-\sqrt{5}}{2}, 1+\sqrt{5}\right).
\end{align*}
In this case, we show that the probability that the following experiment outputs 1 equals $p_\mathsf{stay}$ up to an additive $O(q^{-1/2})$ error:
\begin{itemize}
    \item \textbf{Choose the target symbol.} Draw its class $I_\tar=r \in \{1, 2, 3\}$ with probability $\alpha_r$, then draw its load $K_\tar$ according to
    \[
        \Pr[K_\tar=k\mid I_\tar=r]
        =\begin{cases}\frac{\nu^k/k!}{\sum_{\ell\in B_r}\nu^\ell/\ell!},\text{ if }k\in B_r;\\
        0,\text{ otherwise.}
        \end{cases}
    \]
    The target symbol is uniform over the alphabet, so its class is sampled according to the class proportions.
    \item \textbf{Choose the source independently.} Draw its class $I_\src=r \in \{2, 3\}$ with probability $\alpha_r\mu_r$, then draw its load $K_\src$ according to
    \[
        \Pr[K_\src=k\mid I_\src=r]
        =\begin{cases}
            \frac{k}{\mu_r}\,
          \frac{\nu^k/k!}{\sum_{\ell\in B_r}\nu^\ell/\ell!},\text{ if }k\in B_r; \\
          0,\text{ otherwise.}
        \end{cases}
    \]
    The source is selected by choosing a uniform coordinate, so symbols are sampled in proportion to their loads. This biases both the class $r \in \{2, 3\}$ and the conditional sample $K_\src$.
    \item \textbf{Check both load changes.} Output $1$ if and only if $K_\src-1\in B_{I_\src}$ and $K_\tar+1\in B_{I_\tar}$.
\end{itemize}
In summary, we have the following (which we state in Theorem~\ref{thm:maininfinite}):
\[
    p_\mathsf{stay}
    =\Pr[K_\src-1\in B_{I_\src}]
     \Pr[K_\tar+1\in B_{I_\tar}]
     +O(q^{-1/2}).
\]
For our three-class example, decrements are legal at loads $2$ and $4$, while increments are legal at loads $1$ and $3$. Using the mean-one condition in the form $\frac1{\nu+2}+\frac1{\nu+4}=\frac12$, the two factors are
\begin{align*}
    \Pr[K_\src-1\in B_{I_\src}]
    &=\frac13\frac{2\nu}{2+\nu}
      +\frac16\frac{4\nu}{4+\nu}
      =\frac{\nu}{3},\\
    \Pr[K_\tar+1\in B_{I_\tar}]
    &=\frac13\frac{2}{2+\nu}
      +\frac16\frac{4}{4+\nu}
      =\frac13.
\end{align*}
Consequently, $p_\mathsf{stay}=(\sqrt5-1)/9+O(q^{-1/2})$.

More generally, under the hypotheses of Theorem~\ref{thm:maininfinite}, it suffices for one class occupying a constant fraction of the alphabet to allow two consecutive loads $k,k+1$. There is then constant probability that the target symbol has load $k$ and the source symbol has load $k+1$ in this class, allowing both changes. Thus $p_\mathsf{stay}$ is bounded away from zero even if the other classes have rigid load constraints. A single exceptional symbol with this flexibility would not give the same guarantee.

\paragraph{Specializing to Reed--Solomon codes in the plain model.}
For a Reed--Solomon code of asymptotic rate $\rho$, the dual has rate approaching $1-\rho$. Capacity list decoding of the dual permits any fixed radius below $\rho$. Consequently, our aforementioned stability estimate suffices whenever we can choose, with fixed positive gaps,
\[
    1-p_\mathsf{stay}<\zeta<\rho.
\]
For the three-class example, our calculation gives
\[
    1-p_\mathsf{stay}
    =\frac{10-\sqrt5}{9}+O(q^{-1/2})<\frac78
\]
for all sufficiently large alphabets. This explains the rate $7/8$ in the informal theorem.

In the plain model, the permutations are given by their complete truth tables and $m=q-1$. The capacity decoder of Brakensiek et al.~\cite{brakensiek2026capacity}, together with Jeronimo's field-size refinement~\cite{jeronimo2026capacity}, applies to these ordinary Reed--Solomon codes over prime fields (Theorem~\ref{thm:rscapacity}). Its polynomial list bound gives inverse-polynomial average success. The restriction to prime fields and the averaging over random permutations are both part of this guarantee; the stability bound holds for every choice of permutations, but the average-case reduction still uses their randomness, so we do not get a worst-case algorithm for fixed permutations.

\paragraph{A folded Reed--Solomon oracle separation a la Yamakawa--Zhandry.}

For the oracle separation, the relabeling of $\xv$ into $\yv$ instead proceeds as follows: given $\xv \in \F_q^{ms}$, we group its coordinates into $m$ blocks of length $s$. 
Then the relabeled string $\yv \in [n]^m$ is obtained by applying $H_i:\F_q^s\to[n]$ to the $i$th block. These functions are random oracles. Grouping coordinates into blocks (folding the Reed--Solomon code) gives each oracle the exponentially large domain $\F_q^s$, so a classical algorithm cannot read the functions in full; the classical lower bound then uses list recoverability of the folded code. Here, $n = m + O(1)$. The histogram constraints are then enforced on $\yv$, while requiring $\xv$ to be a Reed--Solomon codeword. The same Fourier stability argument we sketched in the plain model extends readily to show that Regev's reduction succeeds in this setting.

To show that classical algorithms fail in this random oracle model (classical soundness), we require a different property from replacement stability. A closer look at the existing result by Yamakawa--Zhandry indicates that at the end of the day it only requires two things:
\begin{itemize}
    \item Sufficiently good list-recoverability of the code $C$, which we know for folded Reed--Solomon codes~\cite{guruswami2008explicit,rudra2007list}: given a small list of values for each block, few codewords lie in the lists on most blocks. Taking the lists to be the queried points leaves a classical algorithm with few candidate codewords;
    \item $\calT$ should be resistant to random completions: consider an experiment where an adversary fixes a large constant fraction of the entries of $\yv \in [n]^m$ however they want, then the remaining entries are sampled independently and uniformly at random (these are oracle values at unqueried points). Then the property needed by~\cite{yamakawa2024verifiable} is that the resulting word falls in $\calT$ with probability $2^{-\Omega(m)}$.
    
    In the case of~\cite{yamakawa2024verifiable} where the constraints enforced on $\yv$ are local, this is straightforward to show. In our case, it is more delicate but nevertheless we can show the desired resistance to random completions. In fact, it holds whenever some class occupying a constant fraction of the alphabet has a fixed allowed-count set that omits infinitely many counts (Lemma~\ref{lem:prangeresistance}). In particular, this holds when that class allows only a bounded range of counts, as in each of the three classes in our example.
\end{itemize}

When some class occupying a constant fraction of the alphabet has bounded allowed counts, the reason is particularly simple: even after an adversary fixes most of the entries, the remaining random entries can cause a symbol to occur too often. Poissonizing the number of fresh entries, with its mean unchanged, makes their occurrence counts independent across symbols. Since $m=\Theta(n)$ and a constant fraction of entries remain fresh, each symbol receives a Poisson number of additional occurrences with constant positive mean. Thus each symbol in this class has a constant probability of exceeding its allowed upper bound, regardless of the fixed entries. There are $\Omega(n)$ such symbols, so the probability that all avoid this violation is $2^{-\Omega(n)}$. Conditioning the total number of fresh entries back to its original value costs only a polynomial factor, preserving the exponentially small bound.

A consequence of the above is that our concrete three-class example remains classically hard in the random-oracle model even if we discard any two of the three constraint classes. Each retained class still occupies a constant fraction of the alphabet and imposes a constant upper bound on its symbols' occurrence counts: $0$, $2$, or $4$, respectively. Thus satisfying any one class remains classically hard, while our quantum algorithm satisfies all three simultaneously.

\section{Related Work}\label{sec:related-work}

\paragraph{Quantum decoders for Regev's reduction and classical algorithms.}
Chen, Liu, and Zhandry~\cite{clz21} combine filtered quantum measurements with Arora--Ge relinearization~\cite{arorage}; Section~\ref{sec:overview-quantum-decoder} describes both.
Our quantum decoder retains these ingredients, but its guarantee depends on $s_2$ instead of the number of quadratic monomials.
The single-coordinate measurements also belong to the literature on unambiguous measurements and zero-error list decoding~\cite{chailloux2023quantum,buzet2025finegrained,dalai2026zeroerror}.
A different use of quantum decoding appears in Shutty et al.~\cite{shutty2026optimizationusinglocallyquantumdecoders}, who exploit coherent errors in LDPC codes and identify a classical algorithm, an enhanced version of Prange's information-set decoding, that matches their optimization guarantee.

Classical algorithms for related finite-field zero-sum and polynomial-equation problems appear in work of Ivanyos, Sanselme, and Santha; Ivanyos and Santha; and Imran and Ivanyos~\cite{ivanyos2012nil2,ivanyosSantha2015diagonal,imran2024zerosum}.
Kothari, O'Donnell, and Wu give classical algorithms covering the constrained-solution regimes of Chen--Liu--Zhandry~\cite{kothari2025sis}.
For the problem of finding a kernel vector with no zero coordinates, these algorithms kick in at the quadratic scale $m=\Omega(n^2)$ for uniformly random dense matrices.
Our columns instead come from evaluations of low-degree monomials at random points, and our quantum guarantee extends to suitable subquadratic regimes. 
To our knowledge, the preceding classical guarantees do not apply in such subquadratic regimes.

\paragraph{Multiplication in codes.}
Coordinate-wise multiplication is central to classical algebraic decoding and applications of algebraic codes in theoretical computer science.
For Reed--Solomon codes, multiplying the received word by the evaluations of an error-locator polynomial eliminates the errors, while the product of the locator and message polynomials has controlled degree.
The Berlekamp--Welch algorithm exploits this to recover the locator and product by linear algebra~\cite{guruswami2025essential}.
Error-correcting pairs abstract the required product containment and distance conditions~\cite{pellikaan1992decoding,marquez2012pairs}.
The list decoders of Sudan and Guruswami--Sudan also use degree bounds on products and powers~\cite{sudan1997decoding,guruswamiSudan1999}: weighted-degree interpolation and root bounds ensure that substituting a sufficiently agreeing message polynomial yields an identically zero polynomial, after which polynomial factorization recovers the candidate messages. 
Multiplication properties also appear in probabilistically checkable proofs (PCPs) and testing of Reed--Muller codes~\cite{meir2016combinatorial,harsha2016robust}.

Recent constructions combine Reed--Solomon local codes with expanders to build new codes with multiplication properties.
Dinur, Liu, and Zhang~\cite{dinur2025newcodes} place these constraints on a high-dimensional expander, so that local product containment passes to the global code.
They obtain constant rate and relative distance in one parameter regime, and local testability in another.
Golowich and Lin~\cite{golowich2025products} construct classical locally testable codes with a multiplication property using algebraic codes on spectral expanders and balanced products, alongside their quantum-code constructions with transversal gates.
Kopparty and Tamo~\cite{kopparty2026algebraic} construct codes on spectral expander graphs with positive rate and relative distance even when their Reed--Solomon local codes have arbitrarily small fixed positive rate.

Our quantum decoder uses the low dimension of the two-fold product space to reduce the number of unknowns, and our classical algorithm uses orthogonality to triple products to cancel quadratic terms.
The resulting two-fold and three-fold conditions give different sufficient guarantees for finding full-support codewords.
Applying our algorithms to the newer expander constructions would require appropriate bounds on distance and product-space dimension over the relevant fields.
We give further connections to restricted syndrome decoding and constrained-codeword search in Section~\ref{sec:classical-families}.

\paragraph{DQI beyond coordinate-wise constraints.} 
Our work is not the first to consider the applicability of Regev's reduction beyond coordinate-wise nonlinear constraints.
Recently, Krajenbrink et al.~\cite{krajenbrink2026beyond} extend DQI to translation association schemes and develop a rank-metric construction that seeks an encoded matrix close to a target in rank distance.
Since rank depends jointly on the matrix entries, this is a substantial departure from coordinate-wise nonlinear constraints.
Their construction uses superpositions of fixed-rank matrices and rank-metric decoding, and they do not claim quantum advantage.
Other extensions include weighted optimization using multivariate states~\cite{bu2026multivariate} and function intersection on algebraic curves~\cite{matthews2026cofi}.
These works retain allowed-value predicates at individual coordinates.
Their prepared states need not be tensor products: conditioning on a weight or superposing states of different Hamming weights already introduces dependence.
Our histogram constraints instead restrict the entire vector of symbol counts in the output.
This is a different global target, analyzed through Hamming-metric decoding and its stability when one coordinate is replaced, for which we also obtain a random-oracle separation.

\paragraph{Recent advances on OPI.}
Jordan et al.\ introduce DQI and its application to OPI for Reed--Solomon codes~\cite{jordan2024optimization}. 
Their algorithm targets worst-case instances and therefore is limited to regimes where $C^\perp$ can be uniquely decoded. Chailloux and Tillich~\cite{chailloux2024quantum} relax this to average-case instances, which allows list decoding. Our work builds heavily on their average-case reduction.

An exciting recent line of work has built a bridge towards a best-of-both-worlds result, showing algorithms for \emph{worst}-case OPI beyond the unique-decoding regime (but not quite covering the full list-decoding regime yet). This was kickstarted by improved existential bounds for worst-case OPI by Sun and Wootters~\cite{sunwootters2026}.
Later, Horinaga and Yamakawa and, independently, Jo give quantum algorithms exceeding the original DQI guarantees in this setting~\cite{horinaga2026worstcase,jo2026efficient}; Horinaga and Yamakawa obtain such guarantees for both partial and full satisfaction.
Our work is in the average-case, list-decoding regime; we discuss worst-case extensions in Section~\ref{sec:open-questions}.

\paragraph{Decoders for Reed--Solomon and Reed--Muller codes.} 
By the 1970s, Reed--Solomon codes had efficient algebraic decoders correcting any error pattern of weight below half their minimum distance~\cite{sugiyama1975method,justesen1976complexity}.
Sudan~\cite{sudan1997decoding} broke this unique-decoding barrier at low rates using polynomial-time list decoding, and Guruswami--Sudan~\cite{guruswamiSudan1999} extended this to the Johnson radius at all rates.
Guruswami--Rudra~\cite{guruswami2008explicit} subsequently achieved list-decoding capacity for folded Reed--Solomon codes.

The landscape of list-decoding for Reed--Solomon codes changed recently with the breakthrough work of Brakensiek et al.~\cite{brakensiek2026capacity} obtaining capacity list decoding for ordinary Reed--Solomon codes at every fixed rate, with arbitrary evaluation points over sufficiently large prime fields.
For sufficiently large block lengths, Jeronimo~\cite{jeronimo2026capacity} refined this to allow every prime field at least as large as the block length.
We use this refinement for the length-$q-1$ ordinary Reed--Solomon codes in our plain-model result.

On the Reed--Muller side, Saptharishi, Shpilka, and Volk~\cite{saptharishi2017decoding} decode a fraction $1/2-o(1)$ of random errors in time polynomial in the full block length for unpunctured binary codes of degree $o(\sqrt{v})$, where $v$ is the number of variables.
See Abbe, Shpilka, and Ye~\cite{abbe2020rm} for a broader survey of Reed--Muller decoding.
These guarantees do not suffice for our randomly punctured codes of degree $\Theta(v)$ over odd prime fields; our quantum decoder instead relies on the two-fold multiplication property of these codes without leaning on classical decoding.

\paragraph{Random-oracle separations for Regev's reduction.}
Yamakawa and Zhandry prove quantum advantage for code intersection relative to a random oracle~\cite{yamakawa2024verifiable}.
Their classical lower bound uses list recovery of folded Reed--Solomon codes~\cite{guruswami2008explicit,rudra2007list}.
We retain this oracle and list-recovery architecture, but replace their bound for random completions of coordinate-wise constraints (Section~\ref{sec:overview-histogram}) with a bound on the probability that fresh random symbols complete an adversarially fixed partial histogram.
This permits constraints on occurrence counts and shows that each individual class in our three-class example remains classically hard even when the other two classes are discarded.

\section{Open Questions}\label{sec:open-questions}

The obvious open question is whether the problems solved by our quantum decoder and by our plain-model histogram algorithm are actually classically hard, or whether these algorithms can be dequantized. 
Most tantalizingly (in our view), can a classical search algorithm exploit only the two-fold multiplication property, or otherwise avoid the three-fold condition used by our classical algorithm? This could give a full dequantization of our quantum algorithm for Reed--Muller codes. 
Beyond that, we catalog other interesting directions for future work below:

\paragraph{Combining our two contributions.}
Can Regev's reduction give a candidate quantum advantage using both a quantum decoder and global nonlinear constraints?
Each of our two applications removes only one of the two restrictions from the introduction. 
There is a fundamental technical challenge here; the analysis of our quantum decoder relies heavily on the fact that it works with independent single-coordinate states, which arise from coordinate-wise constraints. 
A quantum decoder for global nonlinear constraints would likely need to work with highly entangled states, so any measurements on a single symbol would disturb the remainder of the state.

\paragraph{Worst-case guarantees for histogram instances.}
Can the plain-model guarantee be strengthened from success averaged over random coordinate permutations to success on most such instances, or to a worst-case guarantee under suitable promises? 
Our use of Regev's reduction still relies on averaging in order to freely lean on list-decoding results, but the works mentioned in Section~\ref{sec:related-work} on worst-case algorithms for OPI in some list-decoding regimes inspire us to hope that this may be improved. Even worst-case existential results would be very interesting. 
One particularly appealing special case is where we just remove the relabeling step from Figure~\ref{fig:histogram-csp} (or equivalently, set all $\pi_i$'s to be the identity): do our methods give an efficient quantum algorithm in this setting?

\paragraph{Applying our methods to more general codes with multiplication structure.} 
As discussed in Section~\ref{sec:related-work}, the study of linear codes with multiplication structure is very rich and extends well beyond Reed--Solomon and Reed--Muller codes. It is therefore interesting to ascertain whether our methods give new quantum or classical algorithms for problems related to these codes.

\section{Preliminaries}

Section~\ref{sec:prelim-coding} collects facts about the codes used in both parts. Section~\ref{sec:prelim-regev} states Regev's reduction in the two forms we use: a worst-case form (Theorem~\ref{thm:worstcasecip}) in Part~\ref{part:rm}, and an average-case form based on list decoding (Corollary~\ref{cor:ctlistdecode}) in Part~\ref{part:rs}. Sections~\ref{sec:prelim-quantum} and~\ref{sec:prelim-probability} give quantum and probabilistic tools used only in Part~\ref{part:rs}.

\paragraph{Notation.}
For a nonnegative function $f$, we write $f(n)=\negl(n)$ to mean $f(n)=n^{-\omega(1)}$: for every constant $c>0$, there is an $n_0$ such that $f(n)\leq n^{-c}$ for all $n\geq n_0$.
The same convention applies when another asymptotic parameter replaces $n$.
Over $\F_q$, write $p=\operatorname{char}(\F_q)$, $\omega_p=e^{2\pi i/p}$, and $\Tr=\Tr_{\F_q/\F_p}$ for the field trace, which is $\F_p$-linear and onto $\F_p$; hence $\sum_{x\in\F_q}\omega_p^{\Tr(ax)}$ is $q$ if $a=0$ and $0$ otherwise.
For $\xv,\yv\in\F_q^M$, put $\langle\xv,\yv\rangle=\sum_jx_jy_j$ and use
\begin{equation}\label{eq:fourier-convention}
    \QFT\ket{\yv}
    =q^{-M/2}\sum_{\xv\in\F_q^M}
       \omega_p^{\Tr(\langle\xv,\yv\rangle)}\ket{\xv}.
\end{equation}
Thus $\QFT^\dag\ket{\yv}=\QFT\ket{-\yv}$.
We write $[n]=\{1,\ldots,n\}$ and, for integers $a\leq b$, $[a,b]=\{a,\ldots,b\}$.
For $\xv\in\F_q^M$, $\wt(\xv)$ is its Hamming weight and $\supp(\xv)$ its set of nonzero coordinates; $\dist(\xv,\yv)=\wt(\xv-\yv)$ is the Hamming distance, and $\dist(C)$ is the minimum distance of a code $C$.
For vectors of equal length, $\xv\odot\yv$, $\xv^{\odot k}$ and, when every entry is nonzero, $\xv^{-1}$ are taken coordinate-wise; $\one$ is the all-ones vector and $\diag(\xv)$ the diagonal matrix with diagonal $\xv$.
For $\cv\in\F_q^M$, $X^{\cv}$ denotes the unitary $\ket{\xv}\mapsto\ket{\xv+\cv}$.
A function from a finite set $D$ to $[n]$ is \emph{balanced} if every element of $[n]$ has exactly $|D|/n$ preimages.

\subsection{Coding Theory}\label{sec:prelim-coding}

Much of this section is adapted from~\cite{yamakawa2024verifiable}.

\begin{definition}[Folded linear codes]
    A code of length $M$ over an alphabet $\Sigma$ is a subset $C \subseteq \Sigma^M$.
    A linear code $C$ is one where $\Sigma = \FF_q$ for some prime power $q$ and $C$ is a linear subspace of $\FF_q^M$.
    Finally, for a parameter $s$ dividing $M$, the $s$-folded linear code $C^{(s)}$ is defined to have length $m := M/s$ over the alphabet $\Sigma = \FF_q^s$.
    Any codeword $(x_1, \ldots, x_M) \in C$ is converted into a codeword in $C^{(s)}$ by clustering blocks of symbols together:
    $$((x_1, \ldots, x_s), (x_{s+1}, \ldots, x_{2s}), \ldots, (x_{M-s+1}, \ldots, x_M)).$$
\end{definition}

\begin{definition}[Rate]
    The rate of a code $C\subseteq\Sigma^M$ is $\frac{\log |C|}{M \log |\Sigma|}$.
    When $C$ is a linear code, it can also be characterized as $\frac{\dim C}{M}$.
\end{definition}

\paragraph{Reed--Solomon codes.}
\begin{definition}[Generalized Reed--Solomon codes]\label{def:grs}
    Let $\FF_q$ be a field, $k \in [1, q-1]$ an integer, $\gamma$ a generator of $\FF_q^*$, and $\vv \in \FF_q^{q-1}$ some vector of nonzero elements.
    Then the generalized Reed--Solomon code $\GRS_{\FF_q, \gamma, k, \vv}$ is defined as:
    $$\left\{(v_1 f(\gamma), v_2 f(\gamma^2), \ldots, v_{q-1} f(\gamma^{q-1})): f \in \FF_q[X], \deg f \leq k-1\right\}.$$
    It is linear, has length $q-1$ over alphabet $\FF_q$, and has rate $k/(q-1)$.
    When $\vv$ is the all 1's vector, we simply refer to this as the Reed--Solomon code $\RS_{\FF_q, \gamma, k}$.
\end{definition}

\begin{lemma}[Generalized Reed--Solomon duality~{\cite[Exercise~5.10(2)]{guruswami2025essential}}]\label{lem:grsduality}
Let $\gamma$ generate $\F_q^*$, let $1\leq k<M=q-1$, and let $\vv\in\F_q^M$ have nonzero entries.
For $i\in[M]$, put $a_i=\gamma^i$ and
\[
    u_i=\left(v_i\prod_{j\in[M]\setminus\{i\}}(a_i-a_j)\right)^{-1}.
\]
Then
$\GRS_{\F_q,\gamma,k,\vv}^{\perp}
 =\GRS_{\F_q,\gamma,M-k,\uv}$.
\end{lemma}
Dividing coordinate $i$ by $v_i$ preserves Hamming agreement, including agreement of entire folded blocks, so the decoders below apply to generalized codes.

\begin{definition}[Asymptotic rate]
    In all Reed--Solomon parameter settings we consider in this paper, $k, q$ will grow together and we will have $k/(q-1) \to \rho$ as $k, q \to \infty$.
    We refer to $\rho \in (0, 1)$ as the asymptotic rate.
\end{definition}

\begin{definition}[Folded Reed--Solomon codes]\label{def:frs}
    For $\gamma$ and $k$ as in Definition~\ref{def:grs} and $s$ dividing $q-1$, we define the folded Reed--Solomon code $\FRS_{\FF_q, \gamma, k, s}$ to be the $s$-folding $\RS_{\FF_q, \gamma, k}^{(s)}$.
    We regard its alphabet as $\F_q^s$.  The code remains linear over $\F_q$; no linearity over an extension field of size $q^s$ is asserted.
\end{definition}

\paragraph{Reed--Muller codes.}
\begin{definition}[Reed--Muller codes~{\cite[Definition~9.1.3]{guruswami2025essential}}]
    Let $\FF_q$ be a field, $v$ a positive integer, and $r \in [0, v(q-1)-1]$ an integer.
    Then the Reed--Muller code $\RM_{\FF_q, v, r} \subseteq \FF_q^{q^v}$ is defined as:
    $$\left\{(f(\xv): \xv \in \FF_q^v): f \in \FF_q[X_1, \ldots, X_v],\ \deg f \leq r,\ \deg_{X_i} f\leq q-1\text{ for all }i\in[v]\right\}.$$
    Its dimension is $\cnt_q(v,r)$ (Definition~\ref{def:count-exponent})~\cite[Proposition~9.4.4]{guruswami2025essential}.
    Given a subset $S \subseteq \FF_q^v$ of size $m$, we define $\RM_{\FF_q, v, r, S} \subseteq \FF_q^m$ analogously, except we only write down the evaluations of the polynomial $f$ at points in $S$.
    In the notation of Definition~\ref{def:puncturing}, this is the punctured code $\RM_{\FF_q, v, r}|_S$.
\end{definition}

\begin{lemma}[{\cite[Proposition~4.2 and Remark~4.7]{heijnenPellikaan1998}}]\label{lem:rmdual}
    For all integers $v\geq1$ and $0\leq r\leq v(q-1)-1$, we have $\RM_{\FF_q, v, r}^\perp = \RM_{\FF_q, v, v(q-1)-r-1}$.
\end{lemma}

\begin{lemma}[Lemma 9.4.1 of~\cite{guruswami2025essential}]\label{lem:rmdistance}
    For all integers $v\geq1$ and $0\leq r\leq v(q-1)-1$, the minimum distance of $\RM_{\FF_q, v, r}$ is at least $q^{v-r/(q-1)}$.
\end{lemma}

\begin{definition}[Reduced-monomial count and counting exponent]\label{def:count-exponent}
    For integers $q\geq 2$, $v\geq 1$, and a real threshold $r\geq 0$, define
    \[
        \cnt_q(v,r)
        :=\#\{\boldsymbol{\alpha}\in\{0,\ldots,q-1\}^v:
        \|\boldsymbol{\alpha}\|_1\leq r\}.
    \]
    Thus $\cnt_q(v,r)=\cnt_q(v,\lfloor r\rfloor)$.
    For a real $x\geq 0$, define the counting exponent
    \[
        \cntexp_q(x)
        :=\inf_{0<z\leq 1}
        \left\{\log_q(1+z+\cdots+z^{q-1})-x\log_q z\right\}.
    \]
\end{definition}

Since $\cnt_q(v,r)=q^v\Pr[S_v\leq r]$, where $S_v$ is a sum of $v$ independent uniform elements of $\{0,\ldots,q-1\}$, the following estimate is the large deviation theorem of Flajolet and Sedgewick~\cite[Theorem~IX.15]{flajolet2009analytic}, applied to the probability generating function $B(u)^v$ with $B(u)=(1+u+\cdots+u^{q-1})/q$.
Its rate function is $W(\rho)=(1-\cntexp_q(\rho))\ln q$.
(The cases $\rho=0$ and $\rho>(q-1)/2$, the shift from $\rho v$ to $r$, and the final claims follow directly from Definition~\ref{def:count-exponent}.)

\begin{lemma}[Reduced-monomial count asymptotics]\label{lem:count-asymptotics}
    Fix an integer $q\geq 2$ and a real $\rho\geq 0$, and let $r=r(v)\geq 0$
    satisfy $r=\rho v+O(1)$. Then, as $v\to\infty$,
    \[
        \cnt_q(v,r)
        =q^{v(\cntexp_q(\rho)+o(1))}.
    \]
    Moreover, $\cntexp_q(0)=0$, $0<\cntexp_q(\rho)<1$ for
    $0<\rho<(q-1)/2$, and $\cntexp_q(\rho)=1$ for $\rho\geq(q-1)/2$.
\end{lemma}

\paragraph{List-decodability.}
\begin{definition}
    A code $C \subseteq \Sigma^M$ is $(\zeta, L)$-list-decodable if for all $\xv \in \Sigma^M$, we have:
    $$\left|\{\cv \in C: \dist(\cv, \xv) \leq \zeta M\}\right| \leq L.$$
    We say $C$ is \emph{efficiently} $(\zeta, L)$-list-decodable if $L = \poly(M, |\Sigma|)$ and there exists a $\poly(M, |\Sigma|)$ time algorithm that, on input $\xv$, outputs every codeword within distance $\zeta M$ of $\xv$.
\end{definition}

The algorithmic part of the following theorem is the Guruswami--Sudan decoder~\cite{guruswamiSudan1999}.
The list-size estimate is the large-alphabet Johnson bound~\cite[Theorem~3.3]{guruswami2006algorithmic}.
\begin{theorem}[Guruswami--Sudan and the large-alphabet Johnson bound]\label{thm:guruswamisudan}
    Let $C = \GRS_{\FF_q, \gamma, k, \vv} \subseteq \FF_q^M$, where $M=q-1$, and write $R_C=k/M$.
    For every $\delta>0$, there is an algorithm running in time $\poly(q,1/\delta)$ that outputs every codeword at relative Hamming distance at most $1-\sqrt{R_C}-\delta$ from a given word.
    If $R_C$ is bounded away from $0$, the output list has size $O(1/\delta)$.
\end{theorem}

For the oracle application, we count errors by folded blocks, which allows radius $1-R_C-\epsilon$ instead of the Johnson radius $1-\sqrt{R_C}$ of Theorem~\ref{thm:guruswamisudan}.
There the folding parameter $s$ grows with $\lambda$, so a running time polynomial in the folded alphabet size $q^s$ would not be polynomial; we therefore state it in terms of the scalar length and base field.

\begin{lemma}[Folded decoding with a constant final list]\label{lem:frscapacity}
    Let $M=q-1$, let $s\mid M$, and let
    $C=\GRS_{\F_q,\gamma,k,\vv}$ have rate $R_C=k/M$.
    Group the coordinates into $m=M/s$ consecutive blocks.
    Fix $\epsilon\in(0,1-R_C)$, put $t=\lceil2/\epsilon\rceil$, and assume
    \[
        s\geq t+\left\lceil\frac{2t}{\epsilon}\right\rceil.
    \]
    There is a deterministic algorithm that lists all codewords differing from a given word in at most $(1-R_C-\epsilon)m$ blocks.
    Its output list has size at most $t$, and its running time is
    \[
        q^{t-1}\poly(M,t,\log q).
    \]
    These claims hold over every finite field $\F_q$.
\end{lemma}
\begin{proof}
    Divide coordinate $i$ of the received word by $v_i$.
    This converts $C$ into an ordinary Reed--Solomon code and preserves agreement in every block.
    For this folded code, the linear-algebraic decoder of~\cite[Theorem~7]{guruswami2011linear} computes an affine space over $\F_q$ of dimension at most $t-1$ containing every message polynomial whose codeword has relative block distance less than
    \[
        r_{s,t}=\frac{t}{t+1}\left(1-\frac{sR_C}{s-t+1}\right).
    \]
    Computing this space takes $\poly(M,t,\log q)$ time.
    Here $t$ is the interpolation parameter; it need not grow with the folding parameter $s$.
    The chosen parameters give
    \begin{align*}
        (1-R_C)-r_{s,t}
        &=\frac{1-R_C}{t+1}
          +\frac{tR_C(t-1)}{(t+1)(s-t+1)}\\
        &\leq\frac1{t+1}+\frac{t-1}{s-t+1}<\epsilon,
    \end{align*}
    since $t+1>2/\epsilon$ and $s-t+1>2t/\epsilon$ make each term less than $\epsilon/2$.
    Thus the affine space contains every codeword in the required closed decoding ball.

    Enumerate the at most $q^{t-1}$ elements of this space, encode them, and retain precisely those at relative block distance at most $1-R_C-\epsilon$.
    This gives a complete list in the stated time.
    The combinatorial bound of Chen and Zhang~\cite[Theorem~2.20]{chenzhang2025folded} says that the list at radius $r_{s,t}$ has size at most $t$; hence so does the retained list.
    Finally, multiply coordinate $i$ of each output by $v_i$ to return codewords of $C$.
\end{proof}

For ordinary Reed--Solomon codes, Brakensiek et al.~\cite[Corollary~5.1]{brakensiek2026capacity} give deterministic capacity decoding over prime fields: at fixed rate and fixed positive slack $\epsilon$ (the gap between the decoding radius and $1-R_C$), both time and list size are polynomial in $q$.
For general rates, their result requires $q\geq KM$ for a sufficiently large $K$ depending on the rate and slack; the evaluation points may be any $M$ distinct field elements.
To retain length $q-1$ in the plain model, we use the following refinement, which permits every prime field large enough to contain the evaluation set.

\begin{theorem}[Ordinary RS capacity decoding~{\cite[Theorem~1.1]{jeronimo2026capacity}}]\label{thm:rscapacity}
    Fix $\epsilon\in(0,1)$.
    For all sufficiently large $M$, every prime $q\geq M$, every set of $M$ distinct evaluation points in $\F_q$, and every $1\leq k\leq(1-\epsilon)M$, there is a deterministic algorithm listing all polynomials of degree below $k$ whose evaluations agree with a given received word in at least $k+\lceil\epsilon M\rceil$ coordinates.
    The list has size at most $M^{O_\epsilon(1)}$, and the algorithm uses $q^{O_\epsilon(1)}$ field operations.
    Equivalently, its relative decoding radius is $1-k/M-\epsilon$.
    The same bounds hold for generalized Reed--Solomon codes after dividing out the nonzero coordinate multipliers.
\end{theorem}

Theorem~\ref{thm:rscapacity} applies only to prime fields.
The characteristic-two fields in the oracle construction require the folded decoder instead.
The slack $\epsilon$ is fixed as the length grows, and the exponents may depend on it.
The theorem gives polynomial time in the input length when $q=M^{O(1)}$, but gives a polynomial rather than a constant bound on the output list.

Table~\ref{tab:rsdecoders} summarizes the three algorithmic guarantees used in the paper.
The code rate in the table is the rate of the code being decoded, which is the dual code in our applications.

\begin{table}[tb]
\centering
\small
\setlength{\tabcolsep}{4pt}
\renewcommand{\arraystretch}{1.2}
\begin{tabular}{@{}
>{\raggedright\arraybackslash}p{2.6cm}
>{\raggedright\arraybackslash}p{2.75cm}
>{\raggedright\arraybackslash}p{2.5cm}
>{\raggedright\arraybackslash}p{1.65cm}
>{\raggedright\arraybackslash}p{2.85cm}
>{\raggedright\arraybackslash}p{2.4cm}@{}}
\toprule
Decoder & Field and length & Relative radius & Final list & Running time & Application\\
\midrule
Guruswami--Sudan\newline Theorem~\ref{thm:guruswamisudan}
& Any $\F_q$;\newline $M=q-1$
& $1-\sqrt{R_C}-\epsilon$\newline scalar errors
& $O(1/\epsilon)$
& $\poly(q,1/\epsilon)$
& Plain model,\newline constant success\\
Ordinary RS capacity\newline Theorem~\ref{thm:rscapacity}
& Prime $q\geq M$;\newline sufficiently large $M$
& $1-R_C-\epsilon$\newline scalar errors
& $M^{O_\epsilon(1)}$
& $q^{O_\epsilon(1)}$\newline field operations
& Plain model,\newline inverse-polynomial success\\
Folded RS\newline Lemma~\ref{lem:frscapacity}
& Any $\F_q$;\newline $M=q-1$;\newline $s\mid M$
& $1-R_C-\epsilon$\newline block errors
& At most $t$
& $q^{t-1}$ times\newline $\poly(M,t,\log q)$
& Oracle model,\newline constant success\\
\bottomrule
\end{tabular}
\caption{Decoder guarantees at positive rate $R_C=k/M$ and positive slack $\epsilon$.
The Guruswami--Sudan list bound assumes $R_C$ is bounded away from zero.
In the folded row, $t=\lceil2/\epsilon\rceil$ and $s\geq t+\lceil2t/\epsilon\rceil$.
The ordinary capacity row uses Jeronimo's refinement of Brakensiek et al.
The folded algorithm enumerates $q^{t-1}$ affine candidates before filtering to at most $t$ outputs; its exponent depends on the slack, not on $s$. (Theorem~\ref{thm:guruswamisudan} calls this slack $\delta$.)}
\label{tab:rsdecoders}
\end{table}

\paragraph{List-recoverability.}
The classical lower bound uses the following combinatorial list-recovery guarantee for the primal folded code.
It bounds the number of compatible codewords; our soundness proof does not require an algorithm that enumerates them.
\begin{definition}
    A code $C \subseteq \Sigma^m$ is $(\zeta, \ell, L)$-list-recoverable if for any subsets $S_1, \ldots, S_m \subseteq \Sigma$ each of size at most $\ell$, we have:
    $$\left|\{\cv \in C: c_i \in S_i\text{ for at least }(1-\zeta)m\text{ values of }i\}\right| \leq L.$$
\end{definition}

The following folded-Reed--Solomon list-recovery bound is due to Guruswami and Rudra~\cite{guruswami2008explicit}; see also Rudra's thesis~\cite[Section~3.6]{rudra2007list}.
We use the formulation in~\cite[Lemma~4.3]{yamakawa2024verifiable}, with $\zeta$ denoting the fraction of positions at which a codeword may fall outside the input lists.
\begin{theorem}[Folded Reed--Solomon list recovery~\cite{guruswami2008explicit}]\label{thm:rudra}
    Let $q, \gamma, k, s$ be as in Definition~\ref{def:frs}.
    Let $\ell, r, d$ be positive integers such that $d \leq s$ and let $\zeta \in (0, 1)$ be real.
    Assume both of the following inequalities hold:
    \begin{align}
        \frac{(1-\zeta)(q-1)}{s} &\geq \left(1+\frac{d}{r}\right)\frac{\sqrt[d+1]{(q-1)\ell (k-1)^d}}{s-d+1} \label{eq:gr1} \\
        (r+d) \cdot \sqrt[d+1]{\frac{(q-1)\ell}{k-1}} &< q.\label{eq:gr2}
    \end{align}
    Then $\FRS_{\F_q, \gamma, k, s}$ is $(\zeta, \ell, q^d)$-list-recoverable.
\end{theorem}

\paragraph{Shortened and punctured codes.}

\begin{definition}
    For a vector $\vv \in \F_q^M$, let $Z(\vv) \subseteq [M]$ be the set of indices $i$ with $v_i = 0$.
\end{definition}

\begin{definition}[Shortening]
    For a linear code $C \subseteq \F_q^M$ and a subset $S \subseteq [M]$, we define the \emph{shortened} code $C[S]$ to be:
    \[
        \left\{\cv \in C: c_i = 0\ \forall i \in S\right\}.
    \]
    Note that $C[S]$ is a linear subcode of $C$.
    We will often abuse notation and remove the coordinates in $S$ to view $C[S]$ as a code in $\F_q^{M-|S|}$.
\end{definition}

\begin{definition}[Puncturing]\label{def:puncturing}
    For a linear code $C \subseteq \F_q^M$ and a subset $S \subseteq [M]$, we define the \emph{punctured} code $C|_{S} \subseteq \F_q^{|S|}$ to be:
    \[
        \left\{(c_i: i \in S): \cv \in C\right\}.
    \]
    Note that $C|_S$ is also a linear code, and its dimension is at most $\dim C$.
\end{definition}

\begin{lemma}[Puncturing and shortening are duals]\label{lem:puncshort}
    For any linear code $C\subseteq\F_q^M$ and subset $S \subseteq [M]$, we have $(C[[M] \backslash S])^\perp = (C^\perp)|_{S}$.
    (Here, we are viewing $C[[M] \backslash S]$ as a subspace of $\F_q^{|S|}$.)
\end{lemma}
\begin{proof}
    It is well-known that for any subspaces $A, B \subseteq \F_q^M$, we have $(A \cap B)^\perp = A^\perp + B^\perp$.
    To see this in one line, let $H_A, H_B$ be parity check matrices for $A, B$, then we have:
    $$A \cap B = \ker \begin{bmatrix} H_A \\ H_B \end{bmatrix} \Rightarrow (A \cap B)^\perp = \mathrm{im} \begin{bmatrix} H_A^\top & H_B^\top \end{bmatrix} = \mathrm{im} H_A^\top + \mathrm{im} H_B^\top = A^\perp + B^\perp.$$

    We will show the lemma by taking $A = C$ and $B = \left\{\vv \in \F_q^M: v_i = 0\ \forall i \in [M] \backslash S\right\}$.
    In the following, we will temporarily view $C[[M] \backslash S]$ as a subset of $\F_q^M$.
    We have $A \cap B = C[[M] \backslash S]$ by definition, so the left side of the identity is $(C[[M] \backslash S])^\perp$.
    On its right side, $A^\perp = C^\perp$ and $B^\perp$ consists of all vectors supported on (a subset of) $[M] \backslash S$.
    Hence a vector $\vv$ is in $(C[[M] \backslash S])^\perp$ if and only if there exists a vector in $C^\perp$ that agrees with $\vv$ in all coordinates in $S$.
    Puncturing out coordinates outside $S$ (and thus returning to the convention where $C[[M] \backslash S]$ is a subset of $\F_q^{|S|}$) implies the conclusion.
\end{proof}

\begin{lemma}[Random puncturing approximately preserves distance]\label{lem:randompuncturedistance}
    Let $C \subseteq \F_q^M$ be a nonzero linear code of minimum distance $d$.
    Let $S \subseteq [M]$ be uniformly random among subsets of size $m$.
    Then, for every $\eta \in (0, 1)$,
    \[
        \Pr\left[\exists \cv \in C: \wt(\cv|_S) < (1-\eta) \frac{\wt(\cv) m}{M}\right]
        \leq
        \min\left(
            1,\,
            \frac{|C|-1}{q-1}
            \exp\left(-\frac{\eta^2dm}{2M}\right)
        \right).
    \]
    Consequently, the same expression also upper bounds $\Pr\left[\dist(C|_S) < (1-\eta)\frac{dm}{M}\right]$.
\end{lemma}

\begin{proof}
    We will take a union bound over all nonzero codewords $\cv$ of the probability that the puncturing of that codeword has Hamming weight at most $(1-\eta)\wt(\cv) m/M$.
    Since codewords that are a constant nonzero multiple of each other have the same support under puncturing, we only need to account for $(|C|-1)/(q-1)$ such codewords.
    Fix a nonzero $\cv\in C$ and let $w=\wt(\cv)\geq d$.
    The random variable $X_{\cv}:=\wt(\cv|_S)$
    is hypergeometric, with expected value $\frac{wm}{M}\geq\frac{dm}{M}$.
    The multiplicative Chernoff bound applies to sampling without replacement~\cite[Theorem~4]{hoeffding1963probability} and gives
    \[
        \Pr\left[
            X_{\cv}\leq(1-\eta)\EE[X_{\cv}]
        \right]
        \leq
        \exp\left(-\frac{\eta^2\EE[X_{\cv}]}{2}\right)
        \leq
        \exp\left(-\frac{\eta^2dm}{2M}\right).
    \]
    Taking the aforementioned union bound then proves the result.
\end{proof}

\paragraph{Schur products.}

\begin{definition}[Schur products of codes]
    For codes $C_1,\ldots,C_k \subseteq \F_q^M$, let $C_1 * \cdots * C_k \subseteq \F_q^M$ denote the linear subspace spanned by the coordinate-wise products $\cv_1 \odot \cdots \odot \cv_k$ with $\cv_j \in C_j$ for every $j$.
    When $C_1 = \cdots = C_k = C$, we write $C^{*k}$; we call $C^{*2}$ the \emph{Schur square} of $C$ and $C^{*3}$ its \emph{Schur cube}.
\end{definition}

\subsection{Regev's Reduction}\label{sec:prelim-regev}

\paragraph{Computational conventions.}
Codes are given by explicit generator or parity-check matrices over an explicitly represented finite field.
A quantum decoder is given as a uniform circuit.
Its measurements and random choices are deferred to work registers, so it is a unitary whose inverse is also available.
We count the cost of this coherent implementation, including its work registers.
Classical function oracles use reversible evaluation,
$\ket{x,z}\mapsto\ket{x,z\oplus H(x)}$, with a fixed binary encoding of the output.
Their inverses have the same query cost, and controlled evaluations use at most two queries.
Thus reversing a decoder circuit is also possible in the oracle model.

Preparation error means trace distance, with the convention
$D(\sigma,\tau)=\frac12\norm{\sigma-\tau}_1$.
Quantum Fourier transforms, rotations, and other nonexact gates are compiled to the stated precision; their gate costs include the dependence on the logarithm of the inverse precision.
The following elementary bound records the effect of conditioning.

\begin{lemma}[Approximation under postselection]\label{lem:postselection-error}
Suppose an ideal input state and an actual input state have trace distance at most $\delta$, and apply to both the same quantum operation, which also outputs a success flag.
If the success probability on the ideal input is $p>\delta$, then the success probability on the actual input is at least $p-\delta$, and the two normalized states conditioned on success have trace distance at most $\delta/(p-\delta)$.
\end{lemma}
\begin{proof}
Let $\sigma_0,\sigma_1$ be the subnormalized success states on the ideal and actual inputs, with traces $p,p'$, respectively.
Apply the channel that keeps the output on success and replaces it on failure by a fixed flag state orthogonal to every success output.
Contractivity of trace distance gives
\[
    \norm{\sigma_0-\sigma_1}_1+|p-p'|\leq2\delta,
    \qquad |p-p'|\leq\delta,
\]
where the second bound follows from the first because $|p-p'|=|\Tr(\sigma_0-\sigma_1)|\leq\norm{\sigma_0-\sigma_1}_1$.
Consequently,
\[
 \frac12\left\|\frac{\sigma_0}{p}-\frac{\sigma_1}{p'}\right\|_1
 \leq\frac{\norm{\sigma_0-\sigma_1}_1+|p-p'|}{2p'}
 \leq\frac{\delta}{p'}\leq\frac{\delta}{p-\delta}.
\]
The same argument applies when $\delta$ also includes circuit-approximation errors before the success test.
\end{proof}

\begin{definition}[Quantum decoding problem]
    Let $C \subseteq \F_q^M$ be a linear code, and $\ket{\psi_\mathsf{err}} = \sum_{\ev \in \F_q^M} \alpha(\ev) \ket{\ev}$ be some state.
    The quantum decoding problem $\QDP_{C^\perp, \ket{\psi_{\err}}}$ is defined as follows:
    \begin{itemize}
        \item Input: for a uniformly random codeword $\cv\gets C^\perp$, the input is the state
        \[
            X^\cv\ket{\psi_\err}
            =\sum_{\ev\in\F_q^M}\alpha(\ev)\ket{\cv+\ev}.
        \]
        \item Output: the algorithm succeeds if and only if it outputs $\cv$.
    \end{itemize}
\end{definition}

\begin{definition}[Code intersection problem]\label{def:cip}
    Let $C\subseteq\F_q^M$ be a linear code.
    An instance representation $\mathsf I$ consists of public classical data and, optionally, access to fixed classical function oracles; it specifies a target set $A(\mathsf I)\subseteq\F_q^M$.
    For a distribution $\calD$ on such representations, $\CIP_{C,\calD}$ asks, on input $C$ and $\mathsf I\gets\calD$, for a vector in $C\cap A(\mathsf I)$.
    We write $\CIP_{C,A}$ for a fixed target with a specified representation (which we will refer to as \emph{worst-case} $\CIP$).
    If the intersection is empty, no output is successful.
\end{definition}

\begin{definition}[Translation-invariant represented instances]\label{def:translation-invariant}
    A translation action assigns a representation $\mathsf I+\vv$ to each $\mathsf I$ and $\vv\in\F_q^M$, with
    \[
        A(\mathsf I+\vv)=A(\mathsf I)+\vv,\qquad
        (\mathsf I+\vv)+\wv=\mathsf I+(\vv+\wv),\qquad
        \mathsf I+\mathbf0=\mathsf I.
    \]
    The distribution $\calD$ is translation-invariant if $\mathsf I+\vv$ has law $\calD$ for every fixed $\vv$ when $\mathsf I\gets\calD$.
    We require that the public data and oracle access for $\mathsf I+\vv$ can be constructed or simulated from $\mathsf I,\vv$.
    All costs of this simulation and any preprocessing are included in the preparation and decoding costs below.
\end{definition}

For explicit sets with canonical descriptions, this recovers the usual definition of a translation-invariant distribution on sets.
For our oracle instances, the action is
$H_i^{(\vv)}(x)=H_i(x-v_i)$, with $v_i$ a scalar or block as appropriate.
A translated query uses one original query and reversible subtraction and addition on the input register.
For public truth tables the same action is implemented by reindexing or reversible lookup.
In each application we verify invariance of these representations, not only of the resulting target sets.

\paragraph{Average-case $\CIP$ using Regev's reduction.}

\begin{theorem}[Average-case $\CIP$~\cite{chailloux2024quantum,chaillouxhermouet,chailloux2025opixsoftdecoders}]\label{thm:chailloux}
    Let $C\subseteq\F_q^M$ be a linear code and let $\calD$ be translation-invariant in the sense of Definition~\ref{def:translation-invariant}, with $A(\mathsf I)\neq\emptyset$ on its support.
    Write $\ket{A}$ for the uniform state on $A=A(\mathsf I)$.
    Assume, uniformly for every represented instance in this support:
    \begin{itemize}
        \item A coherent decoder, as in the computational conventions above, solves $\QDP_{C^\perp,\QFT^\dag\ket A}$ with probability at least a known lower bound $P_\dec>0$, using $T_\dec$ gates and $Q_\dec$ oracle queries.
        \item For every $\delta>0$, a state $\delta$-close to $\ket A$ can be prepared using $T_\prep(q,M,\log(1/\delta))$ gates and $Q_\prep(q,M,\log(1/\delta))$ queries.
    \end{itemize}
    These are worst-case requirements over represented instances; decoder success is averaged over the uniform hidden codeword as in the definition of $\QDP$.
    Fix $\epsilon\in(0,1/4)$ and put $\tau=\log_2(16/(\epsilon P_\dec))$.
    There is a quantum algorithm for $\CIP_{C,\calD}$ with success probability at least $P_\dec-\epsilon$, expected gate complexity
    \[
        O\left(\frac{T_\dec+T_\prep(q,M,\tau)
                    +\poly(M,\log q,\tau)}{P_\dec}\right),
    \]
    and expected query complexity
    \[
        O\left(\frac{Q_\dec+Q_\prep(q,M,\tau)}{P_\dec}\right).
    \]
    A bounded-time version has the same success guarantee and gate and query bounds larger by a factor $O(\log(1/\epsilon))$.
\end{theorem}
\begin{proof}
    First suppose preparation and the other nondecoder operations are exact.
    Theorem~1 of~\cite{chailloux2025opixsoftdecoders}, based on~\cite{chaillouxhermouet}, takes a fixed target $A'$ and a uniform syndrome for a full-row-rank parity-check matrix of $C$, and uses a decoder for $\QDP_{C^\perp,\QFT^\dag\ket{A'}}$, which our hypothesis provides.
    It finds an element of the corresponding coset intersected with $A'$, with average success at least $P_\dec$ when the prepared state has exact support in $A'$.
    A uniform $\vv\in\F_q^M$ has a uniform syndrome, and $\xv\in(C-\vv)\cap A'$ if and only if $\xv+\vv\in C\cap(A'+\vv)$.
    So the cited reduction equivalently finds an element of $C\cap(A'+\vv)$ for uniform $\vv$.

    On input $\mathsf I\gets\calD$, sample a uniform $\vv$ and give the reduction the translated representation $\mathsf I'=\mathsf I-\vv$.
    For every representation $\iota$ and vector $v$,
    \[
        \Pr[\mathsf I'=\iota,\ \vv=v]
        =q^{-M}\Pr[\mathsf I=\iota+v]
        =q^{-M}\Pr[\mathsf I=\iota],
    \]
    where the second equality uses the translation invariance of $\calD$ (Definition~\ref{def:translation-invariant}).
    Thus $\mathsf I'$ and $\vv$ are independent, with laws $\calD$ and uniform, respectively.
    Conditional on any $\mathsf I'$, the cited reduction applies to the fixed set $A'=A(\mathsf I')$ and the uniform translate $\vv$.
    Its output lies in $C\cap(A'+\vv)=C\cap A(\mathsf I)$ whenever successful.
    Averaging over $\mathsf I'$ gives success at least $P_\dec$.

    We recall the structure of the cited reduction, which the error analysis below uses.
    It keeps a uniform hidden codeword in one register while running the decoder on its noisy encoding.
    It compares the decoded word with the retained register and restarts on inequality.
    This test succeeds with the decoder's actual average success probability, at least $P_\dec$, independently of the chosen syndrome.
    If they are equal, it erases the retained copy, applies the inverse decoder, and Fourier transforms.
    These are coherent circuit operations with the same oracle access.
    In particular, the test uses the retained codeword and does not require an oracle that recognizes correct decoding.

    Now allocate total trace-distance error at most $\delta=\epsilon P_\dec/8$ before the equality test: use preparation error at most $\delta/2$ and compile the other operations there to total error at most $\delta/2$.
    Since $\log_2(2/\delta)=\tau$, this preparation uses $T_\prep(q,M,\tau)$ gates and $Q_\prep(q,M,\tau)$ queries.
    By Lemma~\ref{lem:postselection-error}, the test succeeds with probability at least $P_\dec-\delta\geq P_\dec/2$, and conditioning changes the state by at most
    \[
        \frac{\delta}{P_\dec-\delta}
        \leq \frac{\epsilon}{8-\epsilon}<\frac{\epsilon}{4}.
    \]
    Compile the remaining Fourier transforms and arithmetic to total error at most $\epsilon/4$.
    The successful-trial output distribution is therefore within $\epsilon/2$ of the ideal one.
    This comparison does not assume that the approximate state is supported on $A'$.

    The expected number of trials is at most $2/P_\dec$.
    Each trial costs $T_\dec+T_\prep(q,M,\tau)+\poly(M,\log q,\tau)$ gates and $Q_\dec+Q_\prep(q,M,\tau)$ queries, up to constant factors.
    The inverse decoder and final operations have no larger cost.
    This proves the expected bounds and success at least $P_\dec-\epsilon/2$.
    Alternatively, abort after
    \[
        \left\lceil\frac{2}{P_\dec}\ln\frac{2}{\epsilon}\right\rceil
    \]
    trials. The abort probability is at most $\epsilon/2$, proving the bounded-time claim.
\end{proof}

\begin{corollary}[Code intersection from block list decoding]\label{cor:ctlistdecode}
    Use the represented-instance model of Theorem~\ref{thm:chailloux}.
    Assume $\calD$ is translation-invariant and, for every $\mathsf I$ in its support, $A=A(\mathsf I)$ is nonempty and $\ket A$ is preparable to precision $\epsilon$ with $\poly(q,M,\log(1/\epsilon))$ gates and oracle queries.
    Fix a divisor $s$ of $M$, put $m=M/s$, and let $\wt_s$ count nonzero consecutive blocks of length $s$.
    Let $\zeta\in(0,1)$ be a constant, and assume:
    \begin{itemize}
        \item For every represented instance in the support of $\calD$, measuring $\QFT^\dag\ket A$ gives a vector of block weight at most $\zeta m$ with probability at least a known bound $P_\sparse\geq1/\poly(q,M)$.
        \item A deterministic algorithm lists all codewords of $C^\perp$ within $\zeta m$ block errors of any received word, in time $T_{\mathsf{list}}(q,M)\leq\poly(q,M)$, with list size at most $L\leq\poly(q,M)$.
    \end{itemize}
    Then a quantum algorithm for $\CIP_{C,\calD}$ succeeds with probability at least
    \[
        \frac{P_\sparse}{L}-2^{-\Omega(q+M)}
    \]
    and has expected running time
    \[
        O\left(\frac{L}{P_\sparse}
          \left(T_{\mathsf{list}}(q,M)+\poly(q,M)\right)\right).
    \]
    In the oracle model its expected query cost is $(L/P_\sparse)\poly(q,M)$.
    A bounded-time version has the same success guarantee and gate and query bounds larger by a factor $O(q+M)$.
    The case $s=1$ gives the ordinary Hamming-weight statement.
    For $s>1$, the Fourier transform and dual code are still over $\F_q$; the blocks only specify the decoding metric.
\end{corollary}
\begin{proof}
    An ideal decoder measures the input to $\QDP_{C^\perp,\QFT^\dag\ket A}$ in the standard basis to obtain $\cv+\ev$, runs the list decoder, and chooses a uniformly random list element; an empty list gives $\bot$.
    Whenever $\wt_s(\ev)\leq\zeta m$, the list contains $\cv$ and has at most $L$ elements, so this guess is correct with probability at least $1/L$.
    Thus ideal decoding succeeds with probability at least $P_\sparse/L$ on every represented instance.

    Retain the measurement outcome, random coins, and intermediate work to implement this decoder coherently.
    For a list of size $\ell\leq L$, a uniform index in $[\ell]$ can be prepared to error $\delta$ with $\poly(L,\log(1/\delta))$ gates.
    Take $\delta=2^{-2(q+M)}P_\sparse/(16L)$.
    The implemented decoder has success at least $P_\dec=P_\sparse/L-\delta>0$.
    Apply Theorem~\ref{thm:chailloux} with $\epsilon=\delta$.
    Its parameter $\tau=\log_2(16/(\delta P_\dec))$ is $O(q+M)$, because $L/P_\sparse$ is polynomial.
    The success is at least $P_\sparse/L-2\delta$, and the expected restart factor is $O(L/P_\sparse)$.
    The bounded-time version costs the additional factor $O(\log(1/\delta))=O(q+M)$.
    Reversible list decoding and all remaining work are polynomial, giving the claimed bounds.
\end{proof}

\paragraph{Worst-case $\CIP$ using Regev's reduction.}

\begin{theorem}[Worst-case $\CIP$~\cite{Regev05,jordan2024optimization}]\label{thm:worstcasecip}
    Let $C\subseteq\F_q^M$ be a linear code, let $A\subseteq\F_q^M$ be nonempty, and let $\ket{\psi_A}$ be a normalized state supported on $A$.
    Assume:
    \begin{itemize}
        \item A coherent decoder solves $\QDP_{C^\perp,\QFT^\dag\ket{\psi_A}}$ with probability at least $P_\dec\geq1/2$ and gate cost $T_\dec(q,M)$.
        \item For every $\delta>0$, preparation to trace-distance error $\delta$ takes time $T_\prep(q,M,\log(1/\delta))$.
    \end{itemize}
    For every $\epsilon\in(0,1/4)$, there is an algorithm for $\CIP_{C,A}$ with success probability at least
    \[
        1-2|C^\perp|(1-P_\dec)-4\epsilon
    \]
    and expected gate complexity
    \[
        O\left(T_\dec(q,M)+T_\prep(q,M,\log(4/\epsilon))
                    +\poly(M,\log q,\log(1/\epsilon))\right).
    \]
    Aborting after $\lceil4\ln(1/\epsilon)\rceil$ trials gives a bounded-time algorithm with the same cost multiplied by $O(\log(1/\epsilon))$ and success at least $1-2|C^\perp|(1-P_\dec)-5\epsilon$.
    The factor $|C^\perp|$ requires decoder success extremely close to one for these bounds to be nontrivial.
\end{theorem}
\begin{proof}
    We give the circuit from~\cite[Section~7]{shutty2026optimizationusinglocallyquantumdecoders} explicitly.
    First assume exact preparation and exact Fourier transforms.
    For $\dv\in C^\perp$, write
    $\ket{\phi_\dv}=X^\dv\QFT^\dag\ket{\psi_A}$.
    A coherent decoder $U_\dec$ has the form
    \[
        U_\dec\ket{\phi_\dv}\ket0
        =\sqrt{1-\epsilon_\dv}\ket{\xi_\dv}\ket{\dv}
         +\sqrt{\epsilon_\dv}\ket{\delta_\dv},
    \]
    where the decoded-word register of $\ket{\delta_\dv}$ is orthogonal to $\ket{\dv}$.
    Let $p$ be the actual average decoder success, so
    \[
        p=\frac1{|C^\perp|}\sum_{\dv\in C^\perp}(1-\epsilon_\dv)
        \geq P_\dec.
    \]
    Each trial prepares a uniform codeword $\dv$ in a retained register and $\QFT^\dag\ket{\psi_A}$ in a second register, adds $\dv$ to the second register, and runs $U_\dec$.
    Measure only whether its decoded-word register equals the retained register.
    On inequality, restart; equality occurs with probability $p$.
    On equality, subtract the decoded word from the retained register and discard that zero register, then apply $U_\dec^\dag$.
    The resulting normalized state is
    \[
        \ket\Psi
        =\frac1{\sqrt{p|C^\perp|}}\sum_{\dv\in C^\perp}
          \sqrt{1-\epsilon_\dv}\,U_\dec^\dag\ket{\xi_\dv}\ket{\dv}
        =\ket{\Phi_{\mathsf{success}}}-\ket\Delta,
    \]
    where
    \[
        \ket{\Phi_{\mathsf{success}}}
        =\frac1{\sqrt{p|C^\perp|}}\sum_{\dv\in C^\perp}\ket{\phi_\dv}\ket0,
        \qquad
        \ket\Delta
        =\frac1{\sqrt{p|C^\perp|}}\sum_{\dv\in C^\perp}
              \sqrt{\epsilon_\dv}\,U_\dec^\dag\ket{\delta_\dv}.
    \]
    Character orthogonality shows that applying $\QFT$ to the first register of $\ket{\Phi_{\mathsf{success}}}$ gives $\sqrt{|C^\perp|/p}$ times the projection of $\ket{\psi_A}$ onto $C$, with zero work registers.
    It is therefore supported on $C\cap A$.
    Let $\Pi$ project onto states whose first register lies in $C\cap A$ after the final $\QFT$.
    By the previous sentence, $(\id-\Pi)\ket{\Phi_{\mathsf{success}}}=0$, so
    \[
        \Pr[\text{failure}]
        =\norm{(\id-\Pi)\ket\Psi}^2
        =\norm{(\id-\Pi)\ket\Delta}^2
        \leq\norm{\ket\Delta}^2.
    \]
    The triangle and Cauchy--Schwarz inequalities give
    \begin{align*}
        \norm{\ket\Delta}
        &\leq\frac1{\sqrt{p|C^\perp|}}
              \sum_{\dv\in C^\perp}\sqrt{\epsilon_\dv}\\
        &\leq\frac{\sqrt{|C^\perp|\sum_\dv\epsilon_\dv}}
                         {\sqrt{p|C^\perp|}}
         =\sqrt{\frac{|C^\perp|(1-p)}p}.
    \end{align*}
    Hence the ideal failure probability is at most
    $|C^\perp|(1-p)/p\leq2|C^\perp|(1-P_\dec)$.

    A generator matrix prepares the uniform codeword state using polynomially many field operations, without enumerating $C^\perp$.
    A trial uses one preparation, one decoder call, and polynomially many gates for codeword addition and the equality test.
    After acceptance it uses one inverse-decoder call and the final Fourier transform.
    The expected number of exact trials is $1/p\leq2$.

    For the actual implementation, allocate error $\epsilon/4$ to state preparation and $\epsilon/4$ to the other operations before the equality test.
    Its acceptance probability is at least $p-\epsilon/2>1/4$.
    Lemma~\ref{lem:postselection-error} bounds the conditioned-state error by
    \[
        \frac{\epsilon/2}{p-\epsilon/2}
        \leq\frac{\epsilon}{1-\epsilon}.
    \]
    Allocate an additional error at most $\epsilon/2$ to the final operations.
    The change in success probability is at most
    $\epsilon/(1-\epsilon)+\epsilon/2<4\epsilon$.
    Thus at most four trials are needed in expectation, giving the stated cost.
    If instead we stop after $\lceil4\ln(1/\epsilon)\rceil$ trials, the probability of no acceptance is at most
    $(1-1/4)^{\lceil4\ln(1/\epsilon)\rceil}\leq\epsilon$.
    This proves the bounded-time version.
\end{proof}

\subsection{Other Quantum Computation Background}\label{sec:prelim-quantum}

\paragraph{Amplitude amplification.}

\begin{lemma}[Phase-correct amplitude amplification with auxiliary input]\label{lemma:ampamp}
    Let $M\geq 2$ and $N\geq 1$ be integers, and let $H: [MN] \to [M]$ be balanced, i.e., $|H^{-1}(y)| = N$ for every $y \in [M]$.
    For every $\delta\in(0,1)$, there is a unitary algorithm $U$ that makes $O(\sqrt{M})$ queries to $H$, uses
    \[
        O\left(\sqrt M\log(MN)+\log(1/\delta)\right)
    \]
    Clifford+$T$ gates, and has the following input-output behavior up to Euclidean error $\delta$, for all amplitudes $\alpha_y$ with $\sum_y|\alpha_y|^2=1$ and all states $\ket{\mathsf{aux}_y}$ on a fixed register:
    \begin{align*}
        U\left(\ket{\psi_{\mathsf{in}}} \ket{0}\right) &= \ket{\psi_{\mathsf{out}}} \ket{0},\text{ where}\\
        \ket{\psi_\mathsf{in}} &= \sum_{y \in [M]} \alpha_y \ket{y} \ket{\mathsf{aux}_y},\text{ and}\\
        \ket{\psi_{\mathsf{out}}} &= \frac{1}{\sqrt{N}} \sum_{x \in [MN]} \alpha_{H(x)} \ket{x} \ket{\mathsf{aux}_{H(x)}}.
    \end{align*}
\end{lemma}
\begin{proof}
    Let $K=2^{\lceil\log_2(MN)\rceil}$, so $MN\leq K<2MN$.
    Define:
    \begin{align*}
        \ket{\psi_1} &= \frac{1}{\sqrt{K}} \sum_{y \in [M]} \sum_{x \in [K]} \alpha_y \ket{x} \ket{y} \ket{\mathsf{aux}_y} \\
        \ket{\psi_2} &= \frac{1}{\sqrt{N}} \sum_{x \in [MN]} \alpha_{H(x)} \ket{x} \ket{H(x)} \ket{\mathsf{aux}_{H(x)}}.
    \end{align*}
    We now describe $U$ in three steps.
    First, we append a register in $\ket{+_{[K]}}=K^{-1/2}\sum_{x\in[K]}\ket{x}$ and reorder registers to obtain $\ket{\psi_1}$.
    This uniform state is prepared exactly by $\log_2 K=O(\log(MN))$ Hadamard gates.

    Secondly, once we have $\ket{\psi_1}$, we use exact amplitude amplification with a final phase adjustment~\cite[Section~2.1]{brassard2002quantum} to prepare $\ket{\psi_2}$. This is done with the following source and target projectors:
    \begin{align*}
        \Pi_\src &= \proj{+_{[K]}} \otimes \id \otimes \id \\
        \Pi_\tar &= \sum_{x \in [MN]} \proj{x} \otimes \proj{H(x)} \otimes \id.
    \end{align*}
    To justify the applicability of amplitude amplification, note that $\braket{\psi_1|\Pi_\tar|\psi_1} = N/K$ and $\ket{\psi_{2}} = \sqrt{K/N} \Pi_\tar \ket{\psi_1}$.
    We also have $\ket{\psi_1} \in \mathrm{im} \Pi_\src$ and
    \[
        \Pi_\src \Pi_\tar \Pi_\src = \frac{N}{K} \Pi_\src.
    \]
    Thus every source component has the same source--target angle, so the same phase adjustment applies simultaneously to the entire auxiliary superposition.

    Since $K/N<2M$, the procedure uses $O(\sqrt M)$ source and target reflections.
    A source reflection uses $O(\log(MN))$ Clifford+$T$ gates: conjugate a phase conditioned on the all-zero state by the Hadamards that prepare $\ket{+_{[K]}}$.
    A target reflection first checks whether $x\in[MN]$ and, if so, computes $H(x)$ into a work register, compares it with $y$, applies a conditional phase, and uncomputes.
    It uses two queries to $H$ and $O(\log(MN))$ additional gates.
    All but the final source and target phases are sign flips.
    The two nonstandard one-qubit phase rotations in the final iteration can each be synthesized to operator-norm error $O(\delta)$ using $O(\log(1/\delta))$ Clifford+$T$ gates~\cite{kliuchnikov2013asymptotically}.
    The ideal amplitude-amplification sequence is exact up to a common, input-independent global phase, which we fix as part of $U$.
    Replacing the two nonstandard rotations by their approximations changes the output by at most $\delta$ after adjusting the constant implicit in the synthesis precision.

    Finally, once we have $\ket{\psi_2}$ we can coherently uncompute $H(x)$ and swap some registers to arrive at the state $\ket{\psi_{\mathsf{out}}}$.
    This proves the stated query and gate bounds.
\end{proof}

\begin{corollary}[Decompression lemma]\label{cor:decompress}
    Let $M\geq 2$ and $N\geq 1$ be integers, and let $H_1, \ldots, H_\ell: [MN] \to [M]$ be balanced functions, i.e., $|H_i^{-1}(y)| = N$ for all $i \in [\ell], y \in [M]$.
    Suppose we have some normalized state:
    \[
        \ket{\psi_\mathsf{in}} = \sum_{y_1, \ldots, y_\ell \in [M]} \alpha_{y_1, \ldots, y_\ell} \ket{y_1, \ldots, y_\ell}.
    \]
    Then there is a quantum algorithm that prepares the state:
    \[
        \ket{\psi_\mathsf{out}} = \frac{1}{N^{\ell/2}} \sum_{x_1, \ldots, x_\ell \in [MN]} \alpha_{H_1(x_1), \ldots, H_\ell(x_\ell)} \ket{x_1, \ldots, x_\ell}.
    \]
    For every $\delta\in(0,1)$, it can prepare this state to Euclidean error $\delta$ using $O(\ell \sqrt{M})$ queries in total to the $H_i$'s and
    \[
        O\left(\ell\sqrt M\log(MN)+\ell\log(\ell/\delta)\right)
    \]
    Clifford+$T$ gates.
\end{corollary}
\begin{proof}
    Apply Lemma~\ref{lemma:ampamp} to each of the $\ell$ registers in sequence with error at most $\delta/\ell$ per application, treating all other registers as the auxiliary register; the errors add by the triangle inequality.
\end{proof}

\begin{remark}
    As written, the proof of Lemma~\ref{lemma:ampamp} relies on the fact that $H$ is exactly balanced to ensure that every source component has the same source--target angle.
    This is not strictly necessary; one can prove a variant of Corollary~\ref{cor:decompress} that handles imbalanced $H_1, \ldots, H_\ell$ as long as every preimage set occupies at least an inverse-polynomial fraction of the domain.
    To do this, one can switch out Lemma~\ref{lemma:ampamp} for a more general amplitude amplification procedure based on the quantum singular value transformation~\cite[Theorems 26 and 27]{10.1145/3313276.3316366}.
\end{remark}

\paragraph{Distinguishing random functions from balanced functions.}
\begin{lemma}\label{lem:distinguishbalanced}
    Let $M$ and $N$ be positive integers, and suppose a quantum algorithm is given oracle access to a function $H: [MN] \to [M]$ such that either:
    \begin{itemize}
        \item $H$ is uniformly random; or
        \item $H$ is uniformly random subject to the added constraint that it is balanced.
    \end{itemize}
    If the quantum algorithm makes $Q$ queries to $H$, it distinguishes these with advantage at most $O(Q^3/(MN))$.
\end{lemma}
\begin{proof}
    Let $G: [MN] \to [M]$ be a fixed balanced function.
    If $F: [MN] \to [MN]$ is uniformly random, then $G \circ F$ is a uniformly random function from $[MN] \to [M]$.
    If $F: [MN] \to [MN]$ is a uniformly random permutation, then $G \circ F$ is a uniformly random balanced function from $[MN] \to [M]$.
    So an algorithm that distinguishes the two cases with $Q$ queries to $H$ distinguishes a random permutation from a random function on $[MN]$ with $2Q$ queries to $F$.
    The conclusion now follows from~\cite{DBLP:journals/qic/Zhandry15}. 
\end{proof}

\subsection{Probability Lemmas}\label{sec:prelim-probability}

\paragraph{Tilted Poisson distributions.}

\begin{definition}[Tilted Poisson distribution]
    Fix a parameter $\nu > 0$ and a nonempty subset $B \subseteq \ZZ_{\geq 0}$.
    Define the tilted Poisson distribution $\tpoiss_{\nu, B}$ to be the distribution on $\ZZ_{\geq0}$ with probabilities:
    \begin{equation}
        p_k := \frac{1}{Z(\nu, B)} \one[k \in B] \frac{\nu^k}{k!},
    \end{equation}
    where we define the partition function:
    \begin{equation}
        Z(\nu, B) := \sum_{k \in B} \frac{\nu^k}{k!},
    \end{equation}
    so that the probabilities add to 1.
\end{definition}

\begin{definition}[Multi-class tilted Poisson distribution]
    Fix a parameter $\nu > 0$, nonempty subsets $B_1, \ldots, B_R \subseteq \ZZ_{\geq 0}$, and constants $\alpha_1, \ldots, \alpha_R \in \RR_{> 0}$ that sum to 1.
    Define the tilted Poisson distribution $\tpoiss_{\nuba}$ to be the labeled $(\alpha_1, \ldots, \alpha_R)$-weighted probabilistic mixture of $\tpoiss_{\nu, B_1}, \ldots, \tpoiss_{\nu, B_R}$.
    More formally, if we let $p_k^{(i)}$ denote the probability of $k$ under $\tpoiss_{\nu, B_i}$, we have
    $$\Pr_{(I, K) \gets \tpoiss_{\nuba}} [I = i \land K = k]  = \alpha_i p_k^{(i)}.$$
\end{definition}

\begin{lemma}[Tilted Poisson tail bound]\label{lem:tpoisstailbound}
    For a fixed choice of $\nu$, $B_1,\ldots,B_R$, and $\alpha_1,\ldots,\alpha_R$, let $X$ have distribution $\tpoiss_{\nu,B_r}$ for some $r$, or let $(I,X)$ have distribution $\tpoiss_{\nuba}$.
    Then $\EE[X^3 \one[X\geq n]]=\negl(n)$ as $n\to\infty$.
\end{lemma}
\begin{proof}
    In either case, the probability of $X$ being $k$ is given by
    $$p_k = \frac{w(k) \nu^k/k!}{Z},$$
    for a fixed bounded function $w: \ZZ_{\geq 0} \to \RR_{\geq 0}$ and a constant $Z = \sum_{k = 0}^{\infty} w(k)\nu^k/k!>0$.
    Let $M$ be a bound on $w$. Then
    \begin{align*}
        \EE_{X \gets \tpoiss}[X^3 \cdot \one[X \geq n]] &= \sum_{k = n}^\infty k^3p_k \\
        &\leq \frac{M}{Z} \sum_{k = n}^{\infty} \frac{k^3 \nu^k}{k!} \\
        &\leq O(n^3\nu^n/n!) \\
        &= \negl(n),
    \end{align*}
    as claimed.
    The second inequality holds because the ratio of consecutive terms, $\nu(k+1)^2/k^3$, is at most $1/2$ for large $k$, so the tail is at most twice its first term.
\end{proof}

\paragraph{Local central limit theorem.}

The following finite-type local central limit theorem is a standard specialization of the lattice local limit theory in~\cite[Chapter~VII]{Petrov1975}.
\begin{theorem}[Finite-type one-dimensional local central limit theorem]\label{thm:1dlocalclt}
    Fix some distributions $\calD_1, \ldots, \calD_R$ supported on $\ZZ$ and scalars $\alpha_1, \ldots, \alpha_R \in \RR_{> 0}$ summing to 1.
    For $r \in [R]$, let $\mu_r, \sigma_r$ be the mean and standard deviation of $\calD_r$.
    Let $n$ grow, and let $n_1, \ldots, n_R$ be such that $n_r = \alpha_r n + O(1)$ for all $r$ and $\sum_r n_r = n$.
    Let $\{Z_{r, i}: r \in [R], i \in [n_r]\}$ be $n$ independent random variables with $Z_{r, i} \sim \calD_r$ for all $i$.
    Let $Z = \sum_{r \in [R]} \sum_{i \in [n_r]} Z_{r, i}$.
    Assume all of the following:
    \begin{itemize}
        \item There exists $r \in [R]$ such that $\Var[\calD_r] > 0$;
        \item every $\calD_r$ has finite third absolute moment; and
        \item the additive subgroup generated by
        $\bigcup_{r\in[R]}(\operatorname{supp}(\calD_r)-\operatorname{supp}(\calD_r))$ is $\ZZ$.
    \end{itemize}
    Then for any $z \in \ZZ$, we have the following bound uniformly over $z$:
    \[
        \left|\Pr[Z = z] - \frac{1}{\sqrt{2\pi \sum_{r \in [R]} n_r \sigma_r^2}} \exp\left(-\frac{(z-\sum_{r \in [R]} n_r \mu_r)^2}{2\sum_{r \in [R]} n_r \sigma_r^2}\right)\right| = O\left(\frac{1}{n}\right).
    \]
\end{theorem}
The proof, including the Fourier-inversion estimate, is given in Appendix~\ref{app:1dlocalclt}.

\begin{corollary}[Probability near the center]
\label{cor:lclt-central-ratio}
    In the setting of Theorem~\ref{thm:1dlocalclt}, consider $z \in \ZZ$ such that $|z - n\sum_{r \in [R]} \alpha_r \mu_r| \leq O(n^{1/4})$.
    Then we have:
    \[
        \Pr[Z = z] = \frac{1}{\sqrt{2\pi n\sum_{r \in [R]} \alpha_r \sigma_r^2}} \left(1+O(1/\sqrt{n})\right).
    \]
\end{corollary}
\begin{proof}
    Note that we have:
    \begin{align*}
        |z - \sum_{r \in [R]} n_r \mu_r| &\leq \underbrace{|z - n\sum_{r \in [R]} \alpha_r \mu_r|}_{\leq O(n^{1/4})} + \sum_{r \in [R]} \underbrace{|n_r - n\alpha_r|}_{\leq O(1)} \cdot |\mu_r| \\
        &= O(n^{1/4}) \\
        \Rightarrow \frac{(z - \sum_{r \in [R]} n_r \mu_r)^2}{2\sum_{r \in [R]} n_r \sigma_r^2} &\leq \frac{O(\sqrt{n})}{\sum_{r \in [R]} n_r \sigma_r^2} \\
        &= O(1/\sqrt{n}) \\
        \Rightarrow \exp\left(-\frac{(z - \sum_{r \in [R]} n_r \mu_r)^2}{2\sum_{r \in [R]} n_r \sigma_r^2}\right) &= 1 + O(1/\sqrt{n}).
    \end{align*}
    Next, we have:
    \begin{align*}
        2\pi \sum_{r \in [R]} n_r \sigma_r^2 &= 2\pi n \sum_{r \in [R]} \alpha_r \sigma_r^2 + O(1) \\
        \Rightarrow \sqrt{2\pi \sum_{r \in [R]} n_r \sigma_r^2} &= \sqrt{2\pi n \sum_{r \in [R]} \alpha_r \sigma_r^2 + O(1)} \\
        &= \sqrt{2\pi n \sum_{r \in [R]} \alpha_r \sigma_r^2} \cdot \sqrt{1+O(1/n)} \\
        &= \sqrt{2\pi n \sum_{r \in [R]} \alpha_r \sigma_r^2} (1+O(1/n)).
    \end{align*}
    It follows that:
    \begin{align*}
        \frac{1}{\sqrt{2\pi \sum_{r \in [R]} n_r \sigma_r^2}} \exp\left(-\frac{(z-\sum_{r \in [R]} n_r \mu_r)^2}{2\sum_{r \in [R]} n_r \sigma_r^2}\right) &= \frac{1}{\sqrt{2\pi n \sum_{r \in [R]} \alpha_r \sigma_r^2}} (1+O(1/\sqrt{n})) \\
        \Rightarrow \Pr[Z = z] &= \frac{1}{\sqrt{2\pi n \sum_{r \in [R]} \alpha_r \sigma_r^2}} (1+O(1/\sqrt{n})) + O(1/n) \\
        &= \frac{1}{\sqrt{2\pi n \sum_{r \in [R]} \alpha_r \sigma_r^2}} (1+O(1/\sqrt{n})),
    \end{align*}
    as desired.
    The last step uses that the main term is $\Theta(n^{-1/2})$, so the additive $O(1/n)$ is a relative $O(n^{-1/2})$.
\end{proof}

\clearpage
\part{Reed--Muller Codes with Quantum Decoding}\label{part:rm}

\section{Problem Definition}\label{sec:classical-families}

Part~\ref{part:rm} studies one problem. Given a uniformly random set $S$ of $m$ points in $\FF_q^v$, find a vector with no zero entry in the kernel of the matrix that evaluates the monomials of degree at most $\rho v$ on $S$; equivalently, find a polynomial of bounded degree that is nonzero exactly on $S$. This section defines both formulations and shows that they are equivalent. Section~\ref{sec:q2} gives a quantum algorithm and, for odd $q$, Section~\ref{sec:classical} gives a classical one. Each succeeds when $m=\lfloor n^c\rfloor$, where $n$ is the number of monomials, for every $c$ above a threshold depending on $q$ and $\rho$ (Theorems~\ref{thm:rm-pair-dqi} and~\ref{thm:cubecond}).

The formulations below specialize established code-search problems to a random
set of Reed--Muller evaluation points. The dual problem is syndrome decoding
with the requested Hamming weight equal to the entire code length, as in
Bricout, Chailloux, Debris-Alazard, and Lequesne~\cite[Problem~1]{bricout2019ternary}.
Equivalently, it is restricted syndrome decoding with allowed alphabet
$\FF_q^\times$; see Baldi et al.~\cite[Problem~1]{baldi2023restricted} and
Burle and Udovenko~\cite[Problem~1]{burle2026restricted}.
The primal problem is constrained-codeword search in a shortened Reed--Muller
code, in the terminology of Chailloux~\cite[Problems~3--4]{chailloux2025opixsoftdecoders}.
Jordan et al.~\cite[Definition~15.1]{jordan2024optimization} also formulate
multivariate Reed--Muller optimal polynomial intersection on the full grid.
In their formulation, the allowed values at every point form a list of the
same size, whereas here we prescribe zeros off a random set and nonzeros on it. These earlier formulations do not give the
thresholds on $c$ for a random set $S$ that we study below.

\begin{definition}[Primal formulation]\label{def:rmprimal}
    Let $q$ be a fixed prime power and $\rho \in (0, q-1), c > 1$ be real constants.
    For a growing positive integer $v$, an instance of the problem $\rmcsp_{q, \rho, c}(v)$ is constructed as follows:
    \begin{itemize}
        \item Let $n = \cnt_q(v, \lfloor \rho v \rfloor)$.
        By Lemma~\ref{lem:count-asymptotics}, $n=q^{\Theta(v)}$.
        \item Let $m = \min(\lfloor n^c \rfloor, q^v)$.
        \item Sample a uniformly random subset $S \subseteq \FF_q^v$ of size $m$.
    \end{itemize}
    The algorithm's task is to find a polynomial $f \in \FF_q[X_1, \ldots, X_v]$ of total degree at most $v(q-1) - \lfloor \rho v \rfloor - 1$ such that for all $\xv \in S$ we have $f(\xv) \neq 0$, and for all $\xv \notin S$ we have $f(\xv) = 0$.
\end{definition}

\begin{definition}[Dual formulation]\label{def:rmdual}
    An instance of the problem $\rmsyn_{q, \rho, c}(v)$ is constructed as follows: compute $n, m$ and sample $S \subseteq \FF_q^v$ as in Definition~\ref{def:rmprimal}.
    Let $\matB \in \FF_q^{n \times m}$ have rows indexed by monomials in $\FF_q[X_1, \ldots, X_v]$ of individual degree $\leq q-1$ and total degree $\leq \lfloor \rho v \rfloor$, and columns indexed by the points of $S$.
    Each entry of $\matB$ is the row's monomial evaluated at the column's point.
    Then the algorithm's task is to find a vector $\yv \in \FF_q^m$ such that none of its entries are 0, and $\matB \yv = 0$.
    (In the generic lemmas of the subsequent sections, one should think of the code $C$ as referring to $\ker \matB$ and $C^\perp$ is therefore $\im \matB^\top$.)
\end{definition}

\begin{remark}\label{remark:rmtrivial}
    Note that when $m = q^v$ both of the above problems are completely trivial; the all 1's vector works for both tasks: since $\lfloor\rho v\rfloor<v(q-1)$, each monomial in a row of $\matB$ has some exponent below $q-1$, so it sums to $0$ over $\FF_q^v$ and $\matB\one=0$.
\end{remark}

The following equivalence is the standard consequence of Reed--Muller duality
and puncturing--shortening duality, recalled in Lemmas~\ref{lem:rmdual}
and~\ref{lem:puncshort}.

\begin{lemma}\label{lem:rmequivalence}
    The problems $\rmcsp_{q, \rho, c}(v)$ and $\rmsyn_{q, \rho, c}(v)$ are equivalent in difficulty (i.e., there exist $\poly(m)$ time reductions in both directions).
\end{lemma}
\begin{proof}
    For any $\yv \in \FF_q^m$, define $\yv' \in \FF_q^{q^v}$ to be the same as $\yv$ at entries in $S$ and 0 elsewhere.
    We will show that $\yv$ is a solution to $\rmsyn_{q, \rho, c}(v)$ if and only if $\yv'$ is the truth table of a solution to $\rmcsp_{q, \rho, c}(v)$.
    If we can show this, then we will be done (converting a truth table into a polynomial can be done in $\poly(q^v) = \poly(m)$ time via standard interpolation, since $m=q^{\Theta(v)}$).
    
    First, note that $\yv$ is nonzero in all positions if and only if the support of $\yv'$ is exactly $S$.
    So it remains to show that $\matB \yv = 0 \Leftrightarrow \yv' \in \RM_{\FF_q, v, v(q-1)-\lfloor \rho v \rfloor - 1}$.
    Indeed, we have $\matB \yv = 0 \Leftrightarrow \yv \in \ker \matB \Leftrightarrow \yv \in (\im \matB^\top)^\perp \Leftrightarrow \yv \in (\RM_{\FF_q, v, \lfloor \rho v \rfloor, S})^\perp$.
    By Lemma~\ref{lem:puncshort} applied to $\RM_{\FF_q, v, \lfloor \rho v \rfloor}^\perp$, followed by taking duals, and then Lemma~\ref{lem:rmdual}, we have:
    $$(\RM_{\FF_q, v, \lfloor \rho v \rfloor, S})^\perp = \RM_{\FF_q, v, \lfloor \rho v \rfloor}^\perp [\FF_q^v \backslash S] = \RM_{\FF_q, v, v(q-1) - \lfloor \rho v \rfloor - 1} [\FF_q^v \backslash S],$$
    so now it is clear by definition that $\yv \in (\RM_{\FF_q, v, \lfloor \rho v \rfloor, S})^\perp \Leftrightarrow \yv' \in \RM_{\FF_q, v, v(q-1)-\lfloor \rho v \rfloor - 1}$.
\end{proof}

\begin{remark}[Arbitrary target syndromes]\label{rem:arbitrary-syndromes}
The generic algorithms below also find full-support solutions to $\matB\yv=\sv$ for arbitrary compatible targets $\sv\in\im\matB$.
For the quantum algorithm, first find $\zv$ with $\matB\zv=\sv$ by Gaussian elimination and apply Theorem~\ref{thm:qalgomain} with forbidden values $\beta_i=-z_i$.
Its output $\xv\in\ker\matB$ satisfies $x_i\neq-z_i$, so
\[
    \yv=\xv+\zv,\qquad
    \matB\yv=\sv,\qquad y_i\neq0\quad\text{for every }i.
\]
The classical algorithm in Lemma~\ref{lem:classical-cube} already accepts a target syndrome.
If $\matB$ has dependent rows, first check compatibility and retain an equivalent system with independent rows.
When $\matB$ has full row rank, every $\sv\in\F_q^n$ is compatible.

When $m<q^v$, the good puncturing events in the RM proofs below are independent of the target, and their guarantees therefore hold for every target on the same good instance.
When $m=q^v$, the all-ones vector of Remark~\ref{remark:rmtrivial} solves only the zero-syndrome problem, and we do not claim the arbitrary-syndrome extension.
\end{remark}

\section{Quantum Algorithm via Two-Fold Multiplication}\label{sec:q2}

This section gives a quantum algorithm that, given a linear code $C$ and one forbidden value for each coordinate, finds a codeword of $C$ avoiding all of them. It succeeds whenever $\dist(C^\perp)$ is large compared with $\dim C^\perp$ and the dimension of the Schur square $(C^\perp)^{*2}$ (Theorem~\ref{thm:qalgomain}). Taking every forbidden value to be $0$ gives a full-support vector in $\ker\matB$, and Section~\ref{sec:q2-rm} applies the algorithm to $\rmsyn$ (Theorem~\ref{thm:rm-pair-dqi}).

In the generic results of this and the next section, $m$ denotes the code length and $n=\dim C^\perp$.
For $C=\ker\matB$, this is $\rank\matB$; a matrix with dependent rows can first be reduced to a basis of its row space.
In the RM specializations we retain the problem definition's $n=\cnt_q(v,\lfloor\rho v\rfloor)$, the dimension before puncturing.
In general $\rank\matB\leq n$. In the proofs of Theorems~\ref{thm:rm-pair-dqi} and~\ref{thm:cubecond}, restriction to $S$ is injective on the Reed--Muller code with probability $1-o(1)$, which gives $\rank\matB=n$.

\begin{definition}[Target sets avoiding one value per coordinate]\label{def:vandA}
    In the following, let $\beta_1, \ldots, \beta_m \in \F_q$ be arbitrary and let $A_i = \F_q \backslash \{\beta_i\}$.
    Let $A = A_1 \times \cdots \times A_m \subseteq \F_q^m$ and let $\ket{A}$ be the uniform superposition over $A$.
\end{definition}

\subsection{Single-Symbol Unambiguous Measurements}

Each coordinate of the decoding input has the form $X^c\QFT^\dag\ket{A_i}$.
The following measurement turns such a state into an unordered pair of symbols that always contains $c$, with a uniformly random partner; the next subsection recovers the codeword from these pairs.
The $q$ possible inputs have unit norm and pairwise inner products of modulus $1/(q-1)$ (see the proof), so this measurement is the tight, equal-overlap case of the Gram-matrix factorization of Dalai, Girardi, and Lami~\cite[Lemma~3]{dalai2026zeroerror}.
We include a direct proof using an isometry on the span of the input states; it needs no inverse of their Gram matrix, which is singular.

\begin{lemma}\label{lem:usd}
    Let $A_i$ be as in Definition~\ref{def:vandA}. For any $i \in [m]$, there is a quantum measurement on a single $\F_q$-qudit with the following input-output behavior: when its input has the form $X^c \QFT^\dag \ket{A_i}$ for some $c \in \F_q$, it outputs the unordered pair $\{c, c'\}$ for uniformly random $c' \gets \F_q \backslash \{c\}$.
    Since it acts on one qudit of dimension $q$, it can be implemented to operator-norm error $\delta$ in time $\poly(q,\log(1/\delta))$.
\end{lemma}
\begin{proof}
    Let $\ket{\phi_c} = X^c \QFT^\dag \ket{A_i}$.
    Then for any distinct $c, c'$, we have:
    \begin{align*}
        \braket{\phi_{c'} | \phi_c} &= \braket{A_i|\QFT X^{c-c'} \QFT^\dag | A_i} \\
        &= \frac{1}{|A_i|} \sum_{y, y' \in A_i} \braket{y' | \QFT X^{c-c'} \QFT^\dag | y} \\
        &= \frac{1}{|A_i|} \sum_{\substack{y, y' \in A_i \\ x, x' \in \F_q}} \braket{y' | \QFT | x'} \braket{x' | X^{c-c'} | x} \braket{x | \QFT^\dag | y} \\
        &= \frac{1}{q |A_i|} \sum_{\substack{y, y' \in A_i \\ x, x' \in \F_q}} \omega_p^{\Tr(x'y'-xy)} \one[x+c-c'=x'] \\
        &= \frac{1}{q |A_i|} \sum_{\substack{y, y' \in A_i \\ x \in \F_q}} \omega_p^{\Tr((x+c-c')y'-xy)} \\
        &= \frac{1}{q |A_i|} \sum_{\substack{y, y' \in A_i \\ x \in \F_q}} \omega_p^{\Tr(x(y'-y) + (c-c')y')} \\
        &= \frac{1}{|A_i|} \sum_{\substack{y \in A_i}} \omega_p^{\Tr((c-c')y)} \\
        &= \frac{1}{q-1} \sum_{\substack{y \in \F_q \backslash \{\beta_i\}}} \omega_p^{\Tr((c-c')y)} \\
        &= -\frac{1}{q-1} \omega_p^{\Tr((c-c')\beta_i)}.
    \end{align*}
    Now, for any $c \in \F_q$, define the state:
    \begin{align*}
        \ket{\psi_c} &:= \frac{1}{\sqrt{2(q-1)}} \omega_p^{\Tr(c\beta_i)} \sum_{a \in \F_q \backslash \{c\}} \left(\ket{c, a} - \ket{a, c}\right).
    \end{align*}
    Then for any $c' \neq c$, we have:
    \[
        \braket{\psi_{c'} | \psi_c} = -\frac{1}{q-1} \omega_p^{\Tr((c-c')\beta_i)} = \braket{\phi_{c'} | \phi_c}.
    \]
    Both families consist of unit vectors.
    Hence the ensembles of states $\{\ket{\phi_c}: c \in \F_q\}$ and $\{\ket{\psi_c}: c \in \F_q\}$ are isometric, so there exists an isometry mapping $\ket{\phi_c} \mapsto \ket{\psi_c}$ for all $c$.
    Measuring $\ket{\psi_c}$ in the standard basis gives $(c,c')$ or $(c',c)$ with $c'$ uniform on $\F_q\setminus\{c\}$; we output the unordered pair.

    The isometry and a unitary extension can be computed by linear algebra in dimension $\poly(q)$.
    Standard unitary synthesis then implements the extension and its inverse to operator-norm error $\delta$ in $\poly(q,\log(1/\delta))$ time~\cite[Section~4.5]{nielsen2010quantum}; see also~\cite[Propositions~14--15]{clz21}.
\end{proof}

\subsection{Relinearization Exploiting the Multiplication Property}

The algorithm below uses Arora--Ge relinearization~\cite{arorage}: encode the
allowed values by a polynomial equation, then treat its monomials as independent
linear unknowns. Earlier uses of this lifting idea for excluded-value
constraints appear in Friedl, Ivanyos, Magniez, Santha, and
Sen~\cite[Theorem~3.5]{friedl2014hidden} and
Ivanyos~\cite[Theorem~1]{ivanyos2008disequations}.
Its combination with quantum filtering and learning from a subset of field
values appears in Ivanyos, Prakash, and Santha~\cite[Proposition~5.5 and
Theorem~5.7]{ivanyos2018subset} and Chen, Liu, and Zhandry~\cite[Section~5.1]{clz21}.
Here the multiplication property restricts the quadratic unknown to the
Schur square of $C^\perp$. The additional analysis below bounds the ambiguity
of this relaxation using shortened Schur squares of $C^\perp$.

\begin{lemma}\label{lem:arorage}
    Let $C \subseteq \F_q^m$ be a code given by a generator or parity-check matrix.
    There is a classical algorithm using $\poly(m)$ field operations for the following task:
    \begin{itemize}
        \item There is a hidden $\cv \in C^\perp$, fixed but arbitrary.
        \item The input to the algorithm consists of the unordered pairs $\{c_1, c'_1\}, \ldots, \{c_m, c'_m\}$, where each $c'_i$ is sampled independently and uniformly from $\F_q \backslash \{c_i\}$.
        \item The algorithm's goal is to output $\cv$.
    \end{itemize}
    The algorithm succeeds with probability (over the randomness of the input pairs) at least:
    \begin{align*}
        1 - \sum_{\cv \in C^\perp \backslash \{0\}} \frac{q^{\dim ((C^\perp)^{*2}[Z(\cv)])}}{(q-1)^{\wt(\cv)}} \geq 1 - \frac{q^{\dim(C^\perp) + \dim((C^\perp)^{*2})}}{(q-1)^{\dist(C^\perp)}}.
    \end{align*}
\end{lemma}
\begin{proof}
    Order each pair arbitrarily as $(a_i,a'_i)$.
    Then we know that $\cv \in C^\perp$ and also that:
    \begin{align*}
        (\cv - \av) \odot (\cv-\av') &= 0 \\
        \Leftrightarrow \cv \odot \cv - (\av + \av') \odot \cv + (\av \odot \av') &= 0.
    \end{align*}
    The algorithm will find (a succinct description of) all solutions to the linear relaxation:
    \begin{align*}
        \wv - (\av+\av') \odot \vv + (\av \odot \av') &= 0,\text{ where} \\
        \wv &\in (C^\perp)^{*2},\text{ and }\\
        \vv &\in C^\perp.
    \end{align*}
    A basis of $C^\perp$ is obtained by Gaussian elimination, and a basis of $(C^\perp)^{*2}$ by reducing the pairwise products of its basis vectors.
    Gaussian elimination then solves the relaxation in $\poly(m)$ field operations, and the algorithm outputs $\vv$ if the solution is unique.
    Clearly $\vv = \cv, \wv = \cv \odot \cv$ gives a solution to this problem.
    We will argue that this is a unique solution with high probability.
    If not, then there must exist shifts $\Delta_\wv \in (C^\perp)^{*2}$ and $\Delta_{\vv} \in C^\perp \backslash \{0\}$ such that:
    \begin{equation}\label{eq:aroragenotunique}
        \Delta_{\wv} = (\av+\av') \odot \Delta_{\vv}.
    \end{equation}
    First, we note an immediate constraint implied by~\eqref{eq:aroragenotunique}: in any position where $\Delta_{\vv}$ is 0, $\Delta_{\wv}$ must also be 0 in that entry.
    In other words, we have $\Delta_{\wv} \in (C^\perp)^{*2}[Z(\Delta_\vv)]$.
    Bearing this in mind, we will take a union bound over all choices of $\Delta_\vv \in C^\perp \backslash \{0\}$ and $\Delta_\wv \in (C^\perp)^{*2}[Z(\Delta_\vv)]$ of the probability that $\av+\av'$ satisfies~\eqref{eq:aroragenotunique}.

    So first, let us fix such $\Delta_\vv$ and $\Delta_\wv$.
    Consider any index $i$ where the $i$th entry of $\Delta_\vv$ is nonzero.
    Then there is at most one value of $a_i+a'_i$ that would not contradict~\eqref{eq:aroragenotunique}.
    Since $a_i+a'_i=c_i+c'_i$, this sum is uniform on $\F_q\setminus\{2c_i\}$, so the probability that index $i$ is satisfied is $\leq 1/(q-1)$.
    Since the $c'_i$ are independent, multiplying over all such indices $i$ gives a probability of $\leq (q-1)^{-\wt(\Delta_\vv)}$.
    Now taking the aforementioned union bound gives:
    \begin{align*}
        \Pr[\text{algorithm fails}] &\leq \sum_{\Delta_\vv \in C^\perp \backslash \{0\}} \sum_{\Delta_\wv \in (C^\perp)^{*2}[Z(\Delta_\vv)]} (q-1)^{-\wt(\Delta_\vv)} \\
        &= \sum_{\Delta_\vv \in C^\perp \backslash \{0\}} q^{\dim ((C^\perp)^{*2}[Z(\Delta_\vv)])} (q-1)^{-\wt(\Delta_\vv)},
    \end{align*}
    as desired.
    Bounding each of the fewer than $q^{\dim C^\perp}$ terms by $q^{\dim((C^\perp)^{*2})}(q-1)^{-\dist(C^\perp)}$ gives the second bound.
\end{proof}

\subsection{Putting It Together}

\begin{theorem}\label{thm:qalgomain}
    Let $C\subseteq\F_q^m$ be given by a generator or parity-check matrix, let $\beta_1,\ldots,\beta_m\in\F_q$ be arbitrary, and define $A$ as in Definition~\ref{def:vandA}.
    Suppose first that $C^\perp\neq\{0\}$.
    For every $\epsilon\in(0,1/4)$ there is a quantum algorithm for $\CIP_{C,A}$ running in $\poly(m,q,\log(1/\epsilon))$ time whose failure probability, over its internal randomness, is at most
    \begin{align*}
        &2|C^\perp|\sum_{\cv\in C^\perp\setminus\{0\}}
          \frac{q^{\dim((C^\perp)^{*2}[Z(\cv)])}}
               {(q-1)^{\wt(\cv)}}+\epsilon\\
        &\hspace{2cm}\leq
          \frac{2q^{2\dim C^\perp+\dim (C^\perp)^{*2}}}
               {(q-1)^{\dist(C^\perp)}}+\epsilon.
    \end{align*}
    When $C^\perp=\{0\}$, the algorithm simply outputs any element of $A$ with no failure; the distance bound is needed only for nonzero $C^\perp$.
    For $q=2$ the bound exceeds $1$, so the theorem is of interest only for $q>2$.
\end{theorem}
\begin{proof}
    We will instantiate Theorem~\ref{thm:worstcasecip} with the uniform superposition over $A$ as our state $\ket{\psi_A}$.
    The $\QDP$ algorithm first applies Lemma~\ref{lem:usd} in every coordinate and then provides the result as input to Lemma~\ref{lem:arorage}, whose output is the $\QDP$ algorithm's answer.

    \smallskip\noindent
    \emph{Standard implementation and precision details.}
    Handle $C^\perp=\{0\}$ as stated, and otherwise put $n=\dim C^\perp$.
    Conditional on the hidden codeword, the input is a product state, so the coordinate pairs have the independent distribution required by Lemma~\ref{lem:arorage}.
    Let
    \[
        E=\sum_{\cv\in C^\perp\setminus\{0\}}
          q^{\dim((C^\perp)^{*2}[Z(\cv)])}(q-1)^{-\wt(\cv)}
    \]
    be the failure bound of Lemma~\ref{lem:arorage}.
    The algorithm of Lemma~\ref{lem:arorage} returns the unique solution of its system or a failure flag.
    Keeping the pair outcomes, random coins, and intermediate computations in work registers makes this a coherent decoder with a polynomial-size inverse.

    Implement the coordinate isometries to total operator-norm error at most
    $\delta=\epsilon/(16q^n)$ (so that the factor $2|C^\perp|=2q^n$ in Theorem~\ref{thm:worstcasecip} turns $\delta$ into $\epsilon/8$).
    A telescoping bound for the $m$ coordinate operations gives this accuracy by allocating $\delta/m$ to each.
    The resulting decoder's success is at least $1-E-\delta$.
    The necessary precision satisfies
    \[
        \log(m/\delta)=O(n\log q+\log m+\log(1/\epsilon)),
    \]
    so its gate cost remains polynomial in $m,q,\log(1/\epsilon)$.

    The state $\ket A=\bigotimes_i\ket{A_i}$ is prepared by preparing each uniform superposition over $q-1$ allowed symbols.
    Preparing each factor to trace-distance error at most $\epsilon/(64m)$ gives total preparation error at most $\epsilon/64$ (the error Theorem~\ref{thm:worstcasecip} allows for preparation at precision $\epsilon/16$).
    Its cost is $\poly(m,q,\log(1/\epsilon))$.

    Use the bounded-time implementation of Theorem~\ref{thm:worstcasecip} with precision parameter $\epsilon/16$.
    If $2q^nE+\epsilon<1$, then
    $1-E-\delta>1/2$, so its decoder-success hypothesis holds.
    The resulting failure bound is at most
    \[
        2q^n(E+\delta)+5\epsilon/16
        =2q^nE+7\epsilon/16
        \leq2q^nE+\epsilon.
    \]
    The same circuit, with the prescribed restart cap, is run on every input; when $2q^nE+\epsilon\geq1$ the claimed probability bound is vacuous.
    Thus the algorithm need not compute $E$.
    Finally, bounding the sum by the distance and the full Schur-square dimension gives the second inequality.
\end{proof}

\subsection{Specializing to Reed--Muller Codes}\label{sec:q2-rm}

We now apply Theorem~\ref{thm:qalgomain} to $\rmsyn$, using the counting estimate of Lemma~\ref{lem:count-asymptotics}.

\begin{theorem}[Quantum algorithm for randomly punctured Reed--Muller codes]
\label{thm:rm-pair-dqi}
    Fix a prime power $q>2$ and real constants
    $\rho\in(0,q-1)$ and $c > 1$ such that:
    \begin{align}
        c &> \frac{
            \max\left\{
                \cntexp_q(2\rho),
                \cntexp_q(\rho)+\rho/(q-1)
            \right\}
        }{\cntexp_q(\rho)}. \label{eq:cond2}
    \end{align}
    Here $\cntexp_q$ is the counting exponent of Definition~\ref{def:count-exponent}.
    Then there is
    a quantum algorithm for $\rmsyn_{q,\rho,c}(v)$, running in time $\poly(m)$ for $m=\min(\lfloor n^c\rfloor,q^v)$ as in Definition~\ref{def:rmprimal}, with the following guarantee:
    with probability $1-o(1)$ over the uniformly random
    $m$-element subset $S\subseteq\FF_q^v$, the algorithm
    outputs a solution with probability $1-o(1)$ over its
    internal measurements, as $v\to\infty$.

    Consequently, the same guarantee holds for
    $\rmcsp_{q,\rho,c}(v)$ by
    Lemma~\ref{lem:rmequivalence}.
\end{theorem}

Before we begin, note by Remark~\ref{remark:rmtrivial} that it suffices for us to address the case where $m = \lfloor n^c \rfloor < q^v$.
The idea of the proof is to apply Theorem~\ref{thm:qalgomain} with $\beta_1=\cdots=\beta_m=0$, $C=\ker\matB$, and $C^\perp =\im\matB^\top =\RM_{\FF_q,v,\lfloor\rho v\rfloor,S}.$
We must show that its failure bound is $o(1)$ for most $S$, which requires $\dim(W[Z(\cv)])$ to be small compared with $\wt(\cv)$ for every nonzero $\cv\in C^\perp$, where $W:=(C^\perp)^{*2}$.
We have $W \subseteq \RM_{\FF_q, v, 2 \lfloor \rho v \rfloor, S}$ since multiplying two polynomials of total degree $\leq r$ yields a polynomial of total degree $\leq 2r$.
For brevity, we denote:
\begin{equation}\label{eq:rm-abbrev}
    r:=\lfloor\rho v\rfloor,\qquad
    N:=q^v,\qquad
    L:=\cnt_q(v,2r) = \dim \RM_{\FF_q, v, 2r},\qquad \delta_0:=q^{-r/(q-1)}.
\end{equation}
If $c\,\cntexp_q(\rho)>1$, then Lemma~\ref{lem:count-asymptotics} gives $m=q^v$ for all large $v$, which is the trivial case.
Otherwise, condition~\eqref{eq:cond2} gives $\cntexp_q(2\rho)<c\,\cntexp_q(\rho)\leq1$, so $2\rho<(q-1)/2$ by the same lemma; in particular $2r\leq v(q-1)-1$, and $\RM_{\FF_q,v,2r}$ is well defined.
By Lemma~\ref{lem:rmdistance}, any nonzero codeword in $\RM_{\FF_q, v, r}$ has Hamming weight at least $N\delta_0$.
The key step is to bound the dimension of the square code after imposing zeros at the sampled zeros of a fixed low-degree word.
If the remaining space still has large dimension, translation symmetry forces it to be supported on many full-grid points.
A fresh sampled zero of the fixed word then has a substantial probability of imposing a new independent constraint.
Iterating this dimension decrease gives a tail strong enough to union-bound over all low-degree words.
We now make this argument precise.
\begin{lemma}\label{lm:rm-pair-growth}
    Assume $m=\lfloor n^c\rfloor<q^v$, and use the abbreviations in~\eqref{eq:rm-abbrev}.
    Then Equation~\eqref{eq:cond2} implies that:
    \begin{equation}\label{eq:rm-pair-growth}
        \frac{m}{L(1+\ln L)} = \omega(1),
        \qquad
        \frac{m\delta_0}{n} = \omega(1),
    \end{equation}
    as $v$ grows.
\end{lemma}
\begin{proof}
    By Lemma~\ref{lem:count-asymptotics}, we have:
    \[
        n=N^{\cntexp_q(\rho)+o(1)},
        \qquad
        L=N^{\cntexp_q(2\rho)+o(1)}.
    \]
    Put $\gamma=c\,\cntexp_q(\rho)$; since $m=\lfloor n^c\rfloor$, we have $m=N^{\gamma+o(1)}$.
    Multiplying~\eqref{eq:cond2} by $\cntexp_q(\rho)>0$ gives:
    \[
        \gamma>
        \max\left\{
            \cntexp_q(2\rho),
            \cntexp_q(\rho)+\frac{\rho}{q-1}
        \right\}.
    \]
    Consequently, we have:
    \begin{align*}
        \frac{m}{L} &= N^{\gamma-\cntexp_q(2\rho)+o(1)} = N^{\Omega(1)} = \omega(\log N) = \omega(1 + \ln L),\text{ and}\\
        \frac{m\delta_0}{n} &= N^{\gamma - \cntexp_q(\rho) - \rho/(q-1) + o(1)} = N^{\Omega(1)} = \omega(1),
    \end{align*}
    as desired.
\end{proof}

\begin{lemma}\label{lem:translationsymmetry}
    With $r$, $N$ and $L$ as in~\eqref{eq:rm-abbrev}, let $R \subseteq \RM_{\FF_q, v, 2r} \subseteq \FF_q^{q^v}$ be any subspace, and write $\supp(R) = \bigcup_{\cv \in R} \supp(\cv)$. Then
    $$\dim R \leq \frac{L}{N}|\supp(R)|.$$
\end{lemma}
\begin{proof}
    Let $I \subseteq \FF_q^v$ be an information set for $\RM_{\FF_q, v, 2r}$, i.e., a set of $L$ points at which no nonzero codeword in $\RM_{\FF_q, v, 2r}$ vanishes identically.
    Restriction to $I$ is injective on $R$, and words of $R$ vanish outside $\supp(R)$, so restriction to $I\cap\supp(R)$ is injective on $R$:
    $$|R| \leq q^{|I \cap \supp(R)|} \Leftrightarrow \dim R \leq |I \cap \supp(R)|.$$
    Now consider any vector $\yv \in \FF_q^v$. We claim that $I+\yv$ is also an information set for $\RM_{\FF_q, v, 2r}$; this is because if $f$ is a polynomial of degree $\leq 2r$ then $\xv \mapsto f(\xv-\yv)$ is also a polynomial of degree $\leq 2r$.
    Therefore, for every translation $\yv$ we have:
    $$\dim R \leq |(I+\yv) \cap \supp(R)|.$$
    Averaging over all translations $\yv \in \FF_q^v$ now implies the conclusion (each element of $\supp(R)$ is included in $I+\yv$ with probability exactly $|I|/q^v = L/N$).
\end{proof}

\begin{lemma}\label{lem:randproc}
    Under the growth conditions of Lemma~\ref{lm:rm-pair-growth}, with $W=(C^\perp)^{*2}$ and the abbreviations in~\eqref{eq:rm-abbrev}, the following holds for all sufficiently large $v$.
    Fix a nonzero $\gv\in\RM_{\FF_q,v,r}$ and put $\delta=\wt(\gv)/N\geq\delta_0$.
    Then
    \begin{equation}\label{eq:rm-pair-shortening-tail}
        \Pr_S\left[\dim(W[Z(\gv|_S)])\geq\frac{\delta m}{8}\right]
        \leq\exp\left(-\frac{\delta m^2}{32L}\right).
    \end{equation}
\end{lemma}
\begin{proof}
    Expose $S$ as an ordered sample $\xv_1,\ldots,\xv_m$ drawn uniformly without replacement from $\FF_q^v$.
    After $i$ points, define the full-grid subspace
    \[
        R_i=\{\uv\in\RM_{\FF_q,v,2r}:
          u_{\xv_j}=0\text{ whenever }j\leq i\text{ and }g_{\xv_j}=0\},
        \qquad d_i=\dim R_i.
    \]
    Then $R_0\supseteq\cdots\supseteq R_m$ and $d_0=L$.
    Every word of $W$ lifts to a full-grid word in $\RM_{\FF_q,v,2r}$.
    If such a word vanishes on $Z(\gv|_S)$, its lift satisfies all the constraints defining $R_m$.
    Applying restriction to $S$ therefore gives
    \[
        W[Z(\gv|_S)]\subseteq R_m|_S,
        \qquad
        \dim(W[Z(\gv|_S)])\leq\dim(R_m|_S)\leq d_m.
    \]

    Conditional on $\xv_1,\ldots,\xv_i$, the next point is uniform on the $N-i$ unexposed points.
    Adding its constraint decreases dimension by one exactly when
    $\xv_{i+1}\in\supp(R_i)\setminus\supp(\gv)$.
    No previously exposed point lies in this set: if $g_{\xv_j}=0$, every word of $R_i$ already vanishes there, and otherwise $\xv_j\in\supp(\gv)$.
    Hence
    \begin{align*}
        \Pr[d_{i+1}=d_i-1\mid\xv_1,\ldots,\xv_i]
        &=\frac{|\supp(R_i)\setminus\supp(\gv)|}{N-i}\\
        &\geq\max\left\{0,\frac{|\supp(R_i)|-\delta N}{N}\right\}\\
        &\geq\max\{0,d_i/L-\delta\},
    \end{align*}
    by Lemma~\ref{lem:translationsymmetry}.

    Put $k=\lceil\delta m/8\rceil$.
    If $k>L$, the event in~\eqref{eq:rm-pair-shortening-tail} is impossible because $d_m\leq L$.
    Otherwise $1\leq k\leq L$.
    Take $v$ sufficiently large that $m\geq16L$ and $m\geq4L\ln L$, which is possible by Lemma~\ref{lm:rm-pair-growth}.
    Then $k\geq2\delta L$, so whenever $d_i\geq k$, the conditional probability of a decrease is at least $d_i/(2L)$.
    With $\binom dk=0$ for $d<k$, and since $d_{i+1}\in\{d_i-1,d_i\}$ and $\binom dk$ is nondecreasing in $d$, so that a larger probability of a decrease only lowers the expectation, we obtain
    \begin{align*}
        \EE\left[\binom{d_{i+1}}k\mid\xv_1,\ldots,\xv_i\right]
        &\leq
          \left(1-\frac{d_i}{2L}\right)\binom{d_i}k
          +\frac{d_i}{2L}\binom{d_i-1}k\\
        &=\left(1-\frac{k}{2L}\right)\binom{d_i}k
        \qquad (d_i\geq k).
    \end{align*}
    If $d_i<k$, then $d_{i+1}<k$, so both sides of this inequality are zero.
    Taking expectations and iterating thus gives
    \begin{align*}
        \Pr[d_m\geq k]
        &\leq\EE\binom{d_m}k
         \leq\left(1-\frac{k}{2L}\right)^m\binom Lk\\
        &\leq\exp\left(k\ln L-\frac{mk}{2L}\right)
         \leq\exp\left(-\frac{mk}{4L}\right)
         \leq\exp\left(-\frac{\delta m^2}{32L}\right).
    \end{align*}
    Combining this with $\dim(W[Z(\gv|_S)])\leq d_m$ completes the proof.
\end{proof}

\begin{proof}[Proof of Theorem~\ref{thm:rm-pair-dqi}]
    Apply Theorem~\ref{thm:qalgomain} with precision $\epsilon=m^{-2}$.
    Its additive error is $o(1)$ and its runtime remains $\poly(m)$, since $q$ is fixed.
    It therefore suffices to show that, with probability $1-o(1)$ over $S$,
    \begin{equation}\label{eq:rm-pair-fine-failure}
        2|C^\perp|
        \sum_{\cv\in C^\perp\setminus\{0\}}
        q^{\dim(W[Z(\cv)])}(q-1)^{-\wt(\cv)}
        =o(1).
    \end{equation}
    We will do this by taking a union bound over all $\gv \in \RM_{\FF_q, v, r} \setminus \{0\}$ and considering the codeword $\cv = \gv|_S$.
    For a fixed such $\gv$:
    \begin{itemize}
        \item By Lemma~\ref{lem:randproc}, we have $\dim(W[Z(\gv|_S)]) \leq \frac{\wt(\gv) m}{8N}$, except with probability $\exp\left(-\frac{\delta_0 m^2}{32L}\right)$.

        \item By Lemma~\ref{lem:randompuncturedistance} applied to the one-dimensional code spanned by $\gv$, with $\eta=1/2$, we have $\wt(\gv|_S) \geq \frac{\wt(\gv)m}{2N}$, except with probability $\exp\left(-\frac{\delta_0 m}{8}\right)$.
    \end{itemize}
    When $m \geq 4L$ (as implied by Lemma~\ref{lm:rm-pair-growth}), the latter probability dominates.
    Therefore taking a union bound implies that both of these hold for all $\gv$, except with probability at most:
    $$O\left(q^n e^{-\delta_0 m/8}\right) = o(1),$$
    by Lemma~\ref{lm:rm-pair-growth}.
    On this simultaneous event, every nonzero full-grid word has a nonzero restriction, because $\wt(\gv|_S)\geq\wt(\gv)m/(2N)>0$.
    Thus restriction on $\RM_{\FF_q,v,r}$ is injective, $\rank\matB=n$, and $|C^\perp|=q^n$.
    The bound below uses only $|C^\perp|\leq q^n$, which holds for every $S$.
    We now derive Equation~\eqref{eq:rm-pair-fine-failure}:
    \begin{align*}
        2|C^\perp|
        \sum_{\cv\in C^\perp\setminus\{0\}}
        q^{\dim(W[Z(\cv)])}(q-1)^{-\wt(\cv)}
        &\leq 2q^n \sum_{\gv \in \RM_{\FF_q, v, r} \setminus \{0\}} \frac{q^{\wt(\gv)m/(8N)}}{(q-1)^{\wt(\gv)m/(2N)}} \\
        &\leq 2q^n \sum_{\gv \in \RM_{\FF_q, v, r} \setminus \{0\}} \frac{q^{\wt(\gv)m/(8N)}}{q^{\wt(\gv)m/(4N)}} \\
        &\leq 2q^{2n} \cdot q^{-\delta_0 m/8} \\
        &= o(1),
    \end{align*}
    again by Lemma~\ref{lm:rm-pair-growth}.
    The second inequality uses $(q-1)^2\geq q$, which holds since $q>2$; the third uses $\wt(\gv)\geq N\delta_0$ (Lemma~\ref{lem:rmdistance}) and $|\RM_{\FF_q,v,r}|=q^n$.
    Thus the failure bound of Theorem~\ref{thm:qalgomain} is $o(1)$ on this event, which completes the proof.
\end{proof}

\section{Classical Algorithm via Three-Fold Multiplication}\label{sec:classical}

For odd $q$, this section gives a classical algorithm that finds a full-support solution of $\matB\yv=\sv$ whenever $\dist(C^\perp)$ is large compared with $\dim C^\perp$ and the dimension of the Schur cube $(C^\perp)^{*3}$ (Lemma~\ref{lem:classical-cube}). Section~\ref{sec:classical-rm} applies it to $\rmsyn$ (Theorem~\ref{thm:cubecond}).

Two other classical algorithms find a vector $\yv$ with no zero entry such that $\matB\yv=0$. 
The first, from~\cite{ivanyos2012nil2,ivanyosSantha2015diagonal,imran2024zerosum,kothari2025sis}, works for any $m \geq \Omega_q(n^2)$ (for uniform $\matB$). 
Another well-known easy case is when $q$ is a power of a prime $p$ but is not itself prime; in this case, $\FF_q$ is an extension field of $\FF_p$ and $\matB \yv = 0$ can be ensured using an ansatz that is linear in the base field, since we can just select an $\FF_p$-affine subspace of $\FF_q$ that avoids 0 for each coordinate.
For $q=p^k$, taking these subspaces of dimension $k-1$ gives $m(k-1)$ unknowns and $nk$ linear equations over $\FF_p$, so a linear $m\geq nk/(k-1)$ suffices when the system has full rank. 
In regimes where $m \leq n^{2-\Omega(1)}$ and $q$ is an odd prime, neither of these algorithms is applicable as far as we know.

\subsection{A Quadratic Ansatz with an Auxiliary Code}\label{sec:classical-aux}

Recall from Section~\ref{sec:overview-quantum-decoder} that the high-level idea of the classical algorithm is to use a degree-2 ansatz to enforce nonvanishing in every coordinate, and to carefully set up the ansatz so that the degree-2 terms cancel, leaving behind a linear system. 

The cancellation of quadratic terms has precedents in multivariate cryptography.
In the Unbalanced Oil and Vinegar signature scheme~\cite[Section~2]{kipnis1999unbalanced}, there are no products of two oil variables, so fixing the vinegar variables leaves a linear system in the oil variables.
For quadratic hash functions, fixing the difference between two inputs makes the collision equations linear; Ding and Yang~\cite{ding2008multivariates} gave this attack, and Applebaum et al.~\cite[Section~5.2]{applebaum2017lowcomplexity} analyzed its success probability for random quadratic maps.
Zhandry~\cite[Section~5.1 of the 2017 version]{zhandry2019lightning} extends this idea to affine spaces of colliding inputs over $\FF_2$, choosing directions whose mixed quadratic terms vanish and then solving linear equations for a translate.
Here, orthogonality to products of codewords eliminates the quadratic part of the syndrome, while a nonsquare ensures that every output coordinate is nonzero.

\begin{algorithm}[label=alg:classical-aux]{Full-support solutions from a quadratic ansatz}
\textbf{Input:}
\begin{itemize}
    \item An odd-order field $\FF_q$ and a nonsquare $\xi\in\FF_q^\times$;
    \item A code $V \subseteq \FF_q^m$ and a matrix $\matB \in \FF_q^{n \times m}$ with linearly independent rows such that $V = \im \matB^\top$;
    \item A code $E \subseteq \FF_q^m$ generated by the rows of $\matA \in \FF_q^{\ell \times m}$;
    \item A vector $\hv \in \FF_q^m$ that is nonzero in every coordinate such that $\hv \perp V * E^{*2}$;
    \item A vector $\wv \in \FF_q^m$ such that $\rank (\matB \diag(\wv) \matA^\top) = n$;
    \item A target syndrome $\sv\in\FF_q^n$.
\end{itemize}
\textbf{Output:} a vector $\yv\in(\FF_q^\times)^m$ with $\matB\yv=\sv$.
\begin{enumerate}
\item Compute
\[
 \yv_0=\frac14\left(\hv^{-1}\odot\wv^{\odot2}-\xi\hv\right).
\]
\item Define $f:\FF_q^\ell\to\FF_q^m$ by
\[
 f(\zv)=\frac14\hv\odot
 \left[\left(\matA^\top\zv+\wv\odot\hv^{-1}\right)^{\odot2}
       -\xi\one\right].
\]
\item Use Gaussian elimination to find a solution $\zv\in\FF_q^\ell$ to
\[
 \bigl(\matB\diag(\wv)\matA^\top\bigr)\zv
       =2(\sv-\matB\yv_0).
\]
\item Return $\yv=f(\zv)$.
\end{enumerate}
\end{algorithm}

\begin{theorem}[Correctness of the quadratic ansatz]\label{thm:classical-aux}
On all inputs satisfying the conditions listed in it, Algorithm~\ref{alg:classical-aux}
returns a full-support solution to $\matB\yv=\sv$.
Given $\xi$, the algorithm is deterministic and uses polynomially many
field operations in $m,n,\ell$.
\end{theorem}
\begin{proof}
    First, note that:
    \begin{align*}
        f(\mathbf{0}) &= \frac{1}{4} \hv \odot \left[(\wv \odot \hv^{-1})^{\odot 2} - \xi \one\right] \\
        &= \frac{1}{4} (\hv^{-1} \odot \wv^{\odot 2} - \xi \hv) \\
        &= \yv_0.
    \end{align*}
    Now let us write $\matB = \begin{bmatrix} \bv_1^\top \\ \vdots \\ \bv_n^\top \end{bmatrix}$.
    For $i \in [n]$, the $i$th entry of $\matB f(\zv)$ is exactly:
    \begin{align*}
        \bv_i^\top f(\zv) &= \frac{1}{4} \bv_i^\top \left(\hv \odot (\matA^\top \zv)^{\odot 2} + 2\wv \odot (\matA^\top \zv) + \hv^{-1} \odot \wv^{\odot 2} - \xi \hv\right) \\
        &= \frac{1}{4} \bv_i^\top \left(\hv \odot (\matA^\top \zv)^{\odot 2} + 2\wv \odot (\matA^\top \zv)\right) + \bv_i^\top \yv_0 \\
        &= \frac{1}{2} \bv_i^\top (\wv \odot (\matA^\top \zv)) + \bv_i^\top \yv_0.
    \end{align*}
    The last step uses the fact that $\hv \perp V * E^{*2}$ and $\bv_i \in V$ and $\matA^\top \zv \in E$, so the quadratic term in $\zv$ vanishes.
    So the equation we wish to solve is:
    $$(\matB \diag(\wv) \matA^\top) \zv= 2\sv - 2\matB \yv_0.$$
    This is linear in $\zv$. By assumption,
    $\matB\diag(\wv)\matA^\top$ has rank $n$, so its image is
    $\FF_q^n$ and the system has a solution for every target $\sv$.

    For every coordinate $i$ and every $\zv$, we also have
    \[
      f_i(\zv)=\frac{h_i}{4}
      \left[\left((\matA^\top\zv)_i+\frac{w_i}{h_i}\right)^2-\xi\right]\ne0,
    \]
    because $h_i\ne0$ and $\xi$ is a nonsquare. Computing $\yv_0$, Gaussian
    elimination, and evaluating the output all take polynomially many
    field operations.
\end{proof}

\subsection{Classical Algorithm}

Taking the auxiliary code to be $C^\perp$ itself gives a classical
algorithm using the three-fold multiplication property. The algorithm
samples a vector $\hv$ orthogonal to every product of three codewords of
$C^\perp$, and applies the quadratic ansatz on the support of $\hv$.

The product orthogonality used here has a precedent in error-correcting pairs:
a pair of codes whose Schur product is orthogonal to the code being decoded turns error location into a linear system~\cite[Remark~2.4]{pellikaan1992decoding};
see also Márquez-Corbella and Pellikaan~\cite[Definition~1]{marquez2012pairs}.
Those constructions decode bounded-weight errors under additional dimension
and distance conditions. Here we combine product orthogonality with a nonsquare
quadratic to construct a full-support solution for every target syndrome.

\begin{algorithm}[label=alg:classical-cube,unbreakable]{Classical algorithm via three-fold multiplication}
\textbf{Input:} a full-row-rank matrix $\matB\in\FF_q^{n\times m}$,
where $q$ is odd, a nonsquare $\xi\in\FF_q^\times$, and a target
$\sv\in\FF_q^n$.
Write $C=\ker\matB$.

\textbf{Output:} a vector $\yv\in(\FF_q^\times)^m$ with $\matB\yv=\sv$,
or \textsc{Failure}.
\begin{enumerate}
    \item Generate $U=(C^\perp)^{*3}$ from the coordinatewise
    products of three rows of $\matB$, and sample $\hv$ uniformly from
    $U^\perp$. Put $T=\supp(\hv)$, and let $\matB'$ and
    $\matB''$ consist of the columns of $\matB$ in $T$ and $[m]\setminus T$,
    respectively. If $\rank\matB'<n$, return \textsc{Failure}.
    \item Choose $I\subseteq T$ of size $n$ such that the submatrix $\matB_I$ of columns indexed by $I$ is
    invertible. On $T$, let $\wv$ be the $0/1$ indicator of $I$.
    \item Run Algorithm~\ref{alg:classical-aux} with $\matB'$ in place of $\matB$,
    $\matA=\matB'$, $V=E=C^\perp|_T$, the vectors $\hv|_T$ and $\wv$,
    the nonsquare $\xi$, and target $\sv-\matB''\one$.
    Denote the returned vector on $T$ by $\yv'$.
    \item Return $\yv$ with $\yv|_T=\yv'$ and
    $\yv|_{[m]\setminus T}=\one$, in the original coordinate order.
\end{enumerate}
After \textsc{Failure}, the algorithm may be retried with an independent sample of $\hv$.
\end{algorithm}

\begin{lemma}[Three-fold multiplication guarantee]\label{lem:classical-cube}
Let $\matB\in\FF_q^{n\times m}$ have linearly independent rows, where
$q$ is odd, and put $C=\ker\matB$ and $U=(C^\perp)^{*3}$.
Suppose that, for some $\epsilon\in(0,1]$,
\begin{equation}\label{eq:cubecondition}
    \frac1{q-1}\sum_{\cv\in C^\perp\setminus\{\mathbf0\}}
    q^{\dim(U[Z(\cv)])-\wt(\cv)}\le1-\epsilon.
\end{equation}
Given a nonsquare $\xi$ and any target $\sv\in\FF_q^n$,
Algorithm~\ref{alg:classical-cube} returns a full-support solution to
$\matB\yv=\sv$ with probability at least $\epsilon$ over the sample of $\hv$, and uses
polynomially many field operations in $m,n$.
Moreover, if $\dist(C^\perp)\ge n+\dim((C^\perp)^{*3})$, then
Equation~\eqref{eq:cubecondition} holds for $\epsilon=1/2$.
\end{lemma}

\begin{proof}
    First suppose that the sampled $\hv$ passes the rank test. The restriction
    $\hv|_T$ is nonzero in every coordinate. Since $\hv$ vanishes outside
    $T$ and is orthogonal to $(C^\perp)^{*3}$, its restriction is
    orthogonal to $V*E^{*2}$. The choice of $\wv$ gives
    \[
        \matB'\diag(\wv)(\matB')^\top
        =\matB_I\matB_I^\top,
    \]
    which is invertible.
    Thus all hypotheses of Algorithm~\ref{alg:classical-aux} hold.

    By Theorem~\ref{thm:classical-aux}, the auxiliary algorithm returns
    a full-support $\yv'$ with
    $\matB'\yv'=\sv-\matB''\one$. Setting the remaining coordinates to $1$ gives
    \[
        \matB\yv=\matB'\yv'+\matB''\one=\sv,
    \]
    with every coordinate nonzero. Every $\hv$ that passes the rank test
    works for every target. Generating the Schur cube, sampling
    $\hv$ from $U^\perp$, selecting independent columns, and applying the
    auxiliary algorithm all use polynomially many field operations.

    It remains to bound the probability that the algorithm returns \textsc{Failure}.
    Since $\matB$ has linearly independent rows, $\rank\matB'<n$ if
    and only if there exists a nonzero $\cv\in C^\perp$ that vanishes
    on $T=\supp(\hv)$. This event is unaffected by multiplying $\cv$
    by a nonzero scalar.

    Fix a nonzero $\cv\in C^\perp$. A vector supported on $Z(\cv)$ lies in
    $U^\perp$ exactly when its restriction to $Z(\cv)$ is orthogonal to
    $U|_{Z(\cv)}$. Consequently,
    \begin{align*}
        \Pr[\supp(\hv)\subseteq Z(\cv)]
        &=\frac{q^{m-\wt(\cv)-\dim(U|_{Z(\cv)})}}{q^{m-\dim U}}\\
        &=q^{\dim U-\dim(U|_{Z(\cv)})-\wt(\cv)}\\
        &=q^{\dim(U[Z(\cv)])-\wt(\cv)}.
    \end{align*}
    The last equality follows from rank--nullity for the restriction
    map $U\to U|_{Z(\cv)}$, whose kernel is $U[Z(\cv)]$.
    Taking a union bound over one representative from each nonzero
    scalar class therefore gives
    \[
        \Pr[\text{the algorithm returns \textsc{Failure}}]
        \le\frac1{q-1}\sum_{\cv\in C^\perp\setminus\{\mathbf0\}}
        q^{\dim(U[Z(\cv)])-\wt(\cv)}
        \le1-\epsilon.
    \]

    Finally, suppose that $\dist(C^\perp)\ge n+\dim((C^\perp)^{*3})$.
    Since $\dim C^\perp=n$, we have
    \begin{align*}
        \frac1{q-1}\sum_{\cv\in C^\perp\setminus\{\mathbf0\}}
        q^{\dim(U[Z(\cv)])-\wt(\cv)}
        &\le\frac{q^n-1}{q-1}\,
        q^{\dim U-\dist(C^\perp)}\\
        &\le\frac{q^n-1}{q-1}\,q^{-n}
        \le\frac1{q-1}\le\frac12,
    \end{align*}
    as claimed.
\end{proof}

Compare~\eqref{eq:cubecondition} with the quantum failure bound in Theorem~\ref{thm:qalgomain}.
Both involve a product space shortened at the zeros of each nonzero dual word.
Here that space is the Schur cube, and the weight factor is $q^{-\wt(\cv)}$ with prefactor $1/(q-1)$; the quantum bound uses the Schur square, $(q-1)^{-\wt(\cv)}$, and the additional factor $2|C^\perp|$ from Regev's reduction.
So the two conditions differ in the product degree (square versus cube) and in these factors.

\subsection{Specializing to Reed--Muller Codes}\label{sec:classical-rm}

Specializing Lemma~\ref{lem:classical-cube} to Reed--Muller codes implies the following theorem:
\begin{theorem}[Classical algorithm for randomly punctured Reed--Muller codes]\label{thm:cubecond}
    Fix an odd prime power $q$ and real constants
    $\rho\in(0,q-1)$ and $c > 1$ such that:
    \begin{align}
        c &> \frac{
            \max\left\{
                \cntexp_q(3\rho),
                \cntexp_q(\rho)+\rho/(q-1)
            \right\}
        }{\cntexp_q(\rho)}. \label{eq:cubecond2}
    \end{align}
    Here $\cntexp_q$ is the counting exponent of Definition~\ref{def:count-exponent}.
    Then there is
    a classical algorithm for $\rmsyn_{q,\rho,c}(v)$, running in time $\poly(m)$ for $m=\min(\lfloor n^c\rfloor,q^v)$ as in Definition~\ref{def:rmprimal}, with the following guarantee:
    with probability $1-o(1)$ over the uniformly random
    $m$-element subset $S\subseteq\FF_q^v$, the algorithm
    outputs a solution with probability $1-o(1)$ over its
    internal randomness, as $v\to\infty$.

    Consequently, the same guarantee holds for
    $\rmcsp_{q,\rho,c}(v)$ by
    Lemma~\ref{lem:rmequivalence}.
\end{theorem}

\begin{proof}
The full-grid case follows from Remark~\ref{remark:rmtrivial}.
Otherwise repeat the proof of Theorem~\ref{thm:rm-pair-dqi} with $U=(C^\perp)^{*3}\subseteq\RM_{\FF_q,v,3r,S}$ in place of $W$ and $L=\cnt_q(v,3r)$.
The same argument with condition~\eqref{eq:cubecond2} gives $3\rho<(q-1)/2$, so $\RM_{\FF_q,v,3r}$ is well defined.
Condition~\eqref{eq:cubecond2} gives the same growth inequalities.
Lemmas~\ref{lem:translationsymmetry} and~\ref{lem:randproc} use only the dimension and translation invariance of the full-grid code $\RM_{\FF_q,v,3r}$, so they apply unchanged.
With probability $1-o(1)$, simultaneously for every nonzero full-grid word $\gv$, we have
\[
    \wt(\gv|_S)\geq\frac{\wt(\gv)m}{2N},
    \qquad
    \dim(U[Z(\gv|_S)])\leq\frac{\wt(\gv)m}{8N}.
\]
The first inequality makes restriction injective, so $\rank\matB=n$.
On this event the failure bound of Lemma~\ref{lem:classical-cube} is at most
\[
    \frac{q^n-1}{q-1}\,q^{-3\delta_0 m/8}=o(1),
\]
because $m\delta_0/n\to\infty$, so Lemma~\ref{lem:classical-cube} applies with $\epsilon=1-o(1)$.
This also gives the arbitrary-syndrome claim of Remark~\ref{rem:arbitrary-syndromes} when $m<q^v$.
\end{proof}

\clearpage

\part{Reed--Solomon Codes and Histogram-Local Constraints}\label{part:rs}

\section{Fourier Spectrum of Set States}

For a nonempty set $A\subseteq\F_q^M$, we study the uniform state
$\ket A=|A|^{-1/2}\sum_{\xv\in A}\ket{\xv}$.
To use a classical decoder in Regev's reduction, we need its Fourier transform to place noticeable mass on errors within the decoding radius.
We relate this spectral question to a classical experiment: start at a uniform element of $A$, replace one randomly chosen coordinate or block, and ask whether the result remains in $A$.
The expected relative Fourier weight is exactly the probability of leaving $A$.
This is the product-space Fourier relation between influence and coordinate resampling~\cite[Chapter~8]{odonnell2014analysis}; we prove the needed form below.
We then show that balanced coordinate maps preserve the exit probability, allowing us to analyze histogram constraints on a smaller alphabet.
These results feed into Corollary~\ref{cor:efwfoldedstate}, which gives the Fourier mass needed in Theorems~\ref{thm:yzpositive} and~\ref{thm:plainrs}.

\subsection{Fourier Sparsity $\Leftrightarrow$ Stability}

Throughout this section, let $M$ be a dimension, $s$ a block length dividing $M$, and $m = M/s$.
We take the convention of segmenting vectors $\xv \in \F_q^M$ as a collection of vectors $\xv_1, \ldots, \xv_{m} \in \F_q^s$.
We use the Fourier convention in Equation~\eqref{eq:fourier-convention}.
We write $\wt(\yv)$ for ordinary Hamming weight and $\wt_s(\yv)$ for the number of nonzero blocks among $\yv_1,\ldots,\yv_m$.

\begin{definition}[Coordinate replacement walk]\label{def:coordinatereplacement}
    For a block $i\in[m]$, define the averaging operator $\calK_i$ by:
    \begin{equation}
        \calK_i := \frac{1}{q^s} \sum_{\xv \in \F_q^M} \sum_{\tv \in \F_q^s} \ket{\xv_1, \ldots, \xv_{i-1}, \tv, \xv_{i+1}, \ldots, \xv_{m}} \bra{\xv_1, \ldots, \xv_{m}}.
    \end{equation}
    Also define $\calK = \frac{1}{m} \sum_{i \in [m]} \calK_i$.
    The matrix entries of $\calK$ are the transition probabilities of the coordinate replacement walk: sample $i\in[m]$ uniformly and replace the $i$th block by a uniform element of $\F_q^s$.
\end{definition}

Here $\calK$ acts linearly on amplitude vectors; we use it as a stochastic matrix, not as a quantum operation to implement.
In particular, the coefficient $q^{-s}$ in $\calK_i$ is a transition probability, and $\calK_i\ket{\xv}$ has norm $q^{-s/2}$, not $1$.
On block $i$, $\calK_i$ is the orthogonal projector onto the uniform state, and it is the identity on all other blocks.
Thus each $\calK_i$, and hence $\calK$, is Hermitian.
The Fourier basis diagonalizes these projectors.

\begin{lemma}\label{lemma:perturbdiag}
    For any $\yv \in \F_q^M$, we have:
    \begin{equation}
        \calK \QFT \ket{\yv} = \left(1 - \frac{\wt_s(\yv)}{m}\right) \QFT \ket{\yv}.
    \end{equation}
    The same holds with $\QFT^\dag$ in the place of $\QFT$ everywhere (one way to see this is that $\QFT^\dag \ket{\yv} = \QFT \ket{-\yv}$ and $\wt_s(\yv) = \wt_s(-\yv)$).
\end{lemma}
\begin{proof}
    By linearity, it suffices to show for every block $i^* \in [m]$ that
    \begin{equation}
        \calK_{i^*} \QFT \ket{\yv} = \begin{cases}
            \QFT \ket{\yv},\text{ if $\yv_{i^*} = 0^{s}$},\\
            0,\text{ otherwise.}
        \end{cases}
    \end{equation}
    Indeed, we have:
    \begin{align*}
        \calK_{i^*} \QFT \ket{\yv} &= q^{-M/2} \sum_{\xv \in \F_q^M} \omega_p^{\Tr(\langle \xv, \yv \rangle)} \calK_{i^*} \ket{\xv} \\
        &= q^{-M/2-s} \sum_{\xv \in \F_q^M} \sum_{\tv \in \F_q^s} \omega_p^{\Tr(\langle \xv, \yv \rangle)} \ket{\xv_1, \ldots, \xv_{i^*-1}, \tv, \xv_{i^*+1}, \ldots, \xv_m} \\
        &= q^{-M/2-s} \sum_{\xv \in \F_q^M} \sum_{\tv \in \F_q^s} \omega_p^{\Tr(\langle (\xv_1, \ldots, \xv_{i^*-1}, \tv, \xv_{i^*+1}, \ldots, \xv_m), \yv \rangle)} \ket{\xv} \text{ (swapping $\tv, \xv_{i^*}$)} \\
        &= q^{-M/2-s} \sum_{\xv \in \F_q^M} \left(\prod_{i \neq i^*} \omega_p^{\Tr(\langle \xv_i, \yv_i \rangle)} \cdot \left( \sum_{\tv \in \F_q^s} \omega_p^{\Tr(\langle \tv, \yv_{i^*}\rangle )}\right)\right) \ket{\xv}.
    \end{align*}
    If $\yv_{i^*} \neq 0^s$ then orthogonality of additive characters makes the sum over $\tv$ vanish.
    Otherwise, this vector is:
    \begin{align*}
        q^{-M/2} \sum_{\xv \in \F_q^M} \left(\prod_{i \neq i^*} \omega_p^{\Tr(\langle \xv_i, \yv_i \rangle)}\right) \ket{\xv} &= q^{-M/2} \sum_{\xv \in \F_q^M} \omega_p^{\Tr(\langle \xv, \yv \rangle)} \ket{\xv} \\
        &= \QFT\ket{\yv}.
    \end{align*}
    The conclusion follows.
\end{proof}

\begin{definition}[Expected Fourier weight]
    Let $\ket{\psi}$ be any normalized state on $\F_q^M$.
    The expected Fourier weight $\EFW_s(\ket{\psi})$ is defined as the expected outcome of the following experiment: measure $\QFT\ket{\psi}$ in the standard basis to obtain $\xv\in\F_q^M$ and output $\wt_s(\xv)/m$.
    Thus this quantity is always in $[0, 1]$.
    We abbreviate $\EFW_1$ as $\EFW$.
\end{definition}

\begin{remark}\label{rem:fourier-convention}
Since the matrix of $\QFT$ in~\eqref{eq:fourier-convention} is symmetric and a set state $\ket{A}$ has real amplitudes,
$\langle\yv|\QFT^\dag|A\rangle=\overline{\langle\yv|\QFT|A\rangle}$, and therefore
$|\langle\yv|\QFT|A\rangle|^2=|\langle\yv|\QFT^\dag|A\rangle|^2$ for every $\yv\in\F_q^M$.
Consequently, all Fourier-weight probabilities defined using $\QFT$ are unchanged if $\QFT^\dag$ is used instead.
\end{remark}

\begin{claim}[Monotonicity in $s$]\label{clm:monotonicity}
    For every nonempty $A\subseteq\F_q^M$ and every divisor $s$ of $M$, we have $\EFW_1(\ket{A}) \leq \EFW_s(\ket{A})$.
\end{claim}
\begin{proof}
    It suffices to check for every string $\xv \in \F_q^M$ that $\frac{1}{M} \wt(\xv) \leq \frac{1}{m} \wt_s(\xv) \Leftrightarrow \wt(\xv) \leq s \cdot \wt_s(\xv)$.
    This is in turn apparent by definition; $\xv$ has only $\wt_s(\xv)$ blocks where it can have nonzero entries, and each of these blocks contributes at most $s$ nonzero symbols in $\F_q$.
\end{proof}

\begin{corollary}[Fourier weight as one-step exit probability]
\label{cor:fourier-exit}
Let \(A\subseteq \F_q^M\) be nonempty.
Let $X$ be a uniform random variable over $A$, and define $X'$ to be the result of taking one step in the coordinate replacement walk from $X$.
(In more detail: choose a coordinate $i \in [m]$ uniformly at random, and replace $X_i$ by an independent uniform element of $\F_q^s$.)
Then:
\[
    \EFW_s(\ket{A})
    =
    \Pr[X'\notin A].
\]
\end{corollary}
\begin{proof}
    By the transition-matrix interpretation of $\calK$,
    \[
        \braket{A|\calK|A}
        =\frac1{|A|}\sum_{\xv,\xv'\in A}
          \bra{\xv'}\calK\ket{\xv}
        =\Pr[X'\in A].
    \]
    It therefore suffices to show for any normalized state $\ket{\psi}$ that
    $\EFW_s(\ket{\psi})=1-\braket{\psi|\calK|\psi}$.
    Let us write $\ket{\psi} = \sum_{\yv} \alpha_{\yv} \QFT^\dag \ket{\yv}$ in the Fourier basis.
    Then we have:
    \begin{align*}
        1 - \braket{\psi|\calK|\psi} &= 1 - \sum_{\yv} |\alpha_\yv|^2 \left(1 - \frac{\wt_s(\yv)}{m}\right) \text{ (Lemma~\ref{lemma:perturbdiag})} \\
        &= 1 - \sum_{\yv} |\alpha_\yv|^2 + \frac{1}{m} \sum_{\yv} |\alpha_\yv|^2 \wt_s(\yv) \\
        &= \frac{1}{m} \sum_{\yv} |\alpha_\yv|^2 \wt_s(\yv) \\
        &= \EFW_s(\ket{\psi}),
    \end{align*}
    as claimed.
\end{proof}

\begin{corollary}[Mass within a decoding radius]\label{cor:fourier-mass}
Let $A\subseteq\F_q^M$ be nonempty and let $\zeta\in(0,1]$.
Measuring $\QFT^\dag\ket A$ gives a vector $\ev$ satisfying
\[
    \Pr[\wt_s(\ev)\leq\zeta m]
    \geq 1-\frac{\EFW_s(\ket A)}{\zeta}.
\]
In particular, a constant gap $\EFW_s(\ket A)\leq\zeta-\delta$ gives probability at least $\delta/\zeta$ of block weight at most $\zeta m$.
\end{corollary}
\begin{proof}
The Fourier-weight distribution is unchanged by replacing $\QFT$ with $\QFT^\dag$ (Remark~\ref{rem:fourier-convention}).
Apply Markov's inequality to the nonnegative random variable $\wt_s(\ev)/m$.
\end{proof}

\subsection{Preimage Sets and States}

Histogram constraints will first be specified on an alphabet $[n]$.
The following construction pulls them back to blocks in $\F_q^s$.
When each block map is balanced, the uniform distribution on the preimage and its block replacement walk map to the uniform distribution on $\calT$ and its coordinate replacement walk (Lemma~\ref{lem:foldpreservesexit}).

\begin{definition}\label{def:foldedsets}
    Let $H_1, \ldots, H_{m}: \F_q^s \to [n]$ be arbitrary functions and let $\calT \subseteq [n]^{m}$.
    We define the \emph{blockwise preimage} of $\calT$ under $H_1, \ldots, H_m$ to be the set $A_{H_1, \ldots, H_{m}} \subseteq \F_q^M$ given by:
    $$A_{H_1, \ldots, H_{m}} := \{\xv \in \F_q^M: (H_1(\xv_1), H_2(\xv_2), \ldots, H_{m}(\xv_{m})) \in \calT\}.$$
    When this set is nonempty, we define $\ket{A_{H_1, \ldots, H_{m}}}$ to be its uniform superposition.
\end{definition}

\begin{claim}[Cardinality for balanced functions]
    If every $H_i$ is balanced, then
    $$|A_{H_1, \ldots, H_{m}}| = |\calT| \cdot (q^s/n)^{m}.$$
\end{claim}
\begin{proof}
    Every vector $\xv'\in \calT$ has exactly $(q^s/n)^{m}$ preimages under $\xv \mapsto (H_1(\xv_1), \ldots, H_{m}(\xv_{m}))$.
\end{proof}

\begin{lemma}[Exit probabilities under balanced block maps]\label{lem:foldpreservesexit}
    Let $H_1,\ldots,H_m:\F_q^s\to[n]$ be balanced functions, and let $\calT\subseteq[n]^m$ be nonempty.
    Sample $X$ uniformly from $A_{H_1,\ldots,H_m}$ and obtain $X'$ by replacing one uniformly chosen block by an independent uniform element of $\F_q^s$.
    Similarly, sample $Y$ uniformly from $\calT$ and obtain $Y'$ by replacing one uniformly chosen coordinate by an independent uniform symbol in $[n]$.
    Then
    \[
        \Pr[X'\notin A_{H_1,\ldots,H_m}]=\Pr[Y'\notin \calT].
    \]
\end{lemma}
\begin{proof}
    Every word in $\calT$ has exactly $(q^s/n)^m$ preimages, so $(H_1(X_1),\ldots,H_m(X_m))$ is uniform on $\calT$.
    Conditional on the chosen block, its fresh uniform value maps under $H_i$ to an independent uniform symbol in $[n]$.
    Thus the image of $X'$ has the law of $Y'$, and membership in the preimage set is exactly membership of this image in $\calT$.
\end{proof}

\section{Histograms and Load States}

For $\xv\in[n]^m$, the \emph{load} of a symbol $a\in[n]$ is its number of occurrences $c_a(\xv)$ in $\xv$. A \emph{load set} $\calT$ consists of the words in which the load of every symbol lies in a prescribed set. This section shows that $\ket{\calT}$ and the corresponding preimage states can be prepared efficiently (Theorem~\ref{thm:Aprep} and Corollary~\ref{cor:prepfolded}). It then computes the probability that replacing one coordinate of a uniform word of $\calT$ keeps it in $\calT$ (Theorem~\ref{thm:maininfinite}). By Corollary~\ref{cor:fourier-exit}, this probability gives the Fourier weight bound (Corollary~\ref{cor:efwfoldedstate}) used by the quantum algorithms in Sections~\ref{sec:yzoracle} and~\ref{sec:plainrs}.

\subsection{Definitions}

\begin{definition}[Histogram]\label{def:histogram}
    For $\xv \in [n]^m$, we define its \emph{histogram} $\cv(\xv) \in \ZZ_{\geq 0}^n$ by:
    \begin{equation}
        c_a(\xv) = |\{i: x_i = a\}|.
    \end{equation}
\end{definition}

\begin{definition}[Load sets and states]\label{def:loadstates}
    For each symbol $a\in[n]$, let $\calB_a\subseteq\{0,1,\ldots,m\}$ be given by its indicator table or by a membership circuit of size $\poly(m,n)$.
    In the circuit representation, the circuits are public input or are uniformly constructible in $\poly(m,n)$ time, and decide $k\in \calB_a$ for every $a\in[n]$ and $0\leq k\leq m$.
    Define the set:
    $$\calT = \left\{\xv \in [n]^m: c_a(\xv) \in \calB_a\ \forall a \in [n]\right\}.$$
    If $\calT \neq \emptyset$, we may also define the quantum state:
    \begin{equation}
        \ket{\calT} := \frac{1}{\sqrt{|\calT|}} \sum_{\xv \in \calT} \ket{\xv}.
    \end{equation}
    Given functions $H_1, \ldots, H_m: \F_q^s \to [n]$, we may also define the preimage set $A_{H_1, \ldots, H_m}$, and its state when it is nonempty, as in Definition~\ref{def:foldedsets}.
\end{definition}

Evaluating the membership circuits on all $n(m+1)$ possible pairs $(a,k)$ constructs the indicator tables in polynomial time; classical evaluations can also be made reversible.
Both state preparation and verification use only these finite tables.
When a set of allowed loads, such as $B_r$ in Definition~\ref{def:specialloadstates}, is an infinite subset of $\ZZ_{\geq0}$, only its intersection with $\{0,\ldots,m\}$ is needed.

\begin{lemma}\label{lem:histogram-counting}
    Let $\cv \in \ZZ_{\geq 0}^n$ be a vector with entry sum $m$.
    Then the number of strings $\xv \in [n]^m$ such that $\cv(\xv) = \cv$ is exactly:
    \[
        \binom{m}{c_1, \ldots, c_n} = \frac{m!}{\prod_{a\in [n]} c_a!}.
    \]
    Here, we are using multinomial coefficient notation, which we abbreviate as $\binom{m}{\cv}$.
\end{lemma}
\begin{proof}
    This quantity is exactly the number of ways of ordering $c_1$ 1's, $c_2$ 2's, \ldots, $c_{n}$ $n$'s.
    It is well-known that this is exactly the multinomial coefficient $\binom{m}{c_1, \ldots, c_n}$.
\end{proof}

\begin{corollary}\label{cor:loadset-size}
    We have:
    \[
    |\calT| = \sum_{\substack{\cv \in \ZZ_{\geq 0}^n \\ c_a \in \calB_a\ \forall a \in [n] \\ \sum_a c_a = m}} \binom{m}{\cv}.
    \]
\end{corollary}

\subsection{Efficient Preparation}

The goal of this subsection is to prove the following theorem:
\begin{theorem}\label{thm:Aprep}
    For nonempty $\calT$ given as in Definition~\ref{def:loadstates}, the state $\ket{\calT}$ can be prepared to trace-distance error $\epsilon$ with $\poly(m,n,\log(1/\epsilon))$ gates, including the construction of the indicator tables.
\end{theorem}
\begin{proof}
    Apply Lemma~\ref{lemma:histsup}, followed coherently by Lemma~\ref{lemma:dicke}, to obtain
    \[
        \frac{1}{\sqrt{|\calT|}}
        \sum_{\cv}\sqrt{\binom m\cv}\ket{\cv}\ket{\psi_{\cv}}
        =
        \frac{1}{\sqrt{|\calT|}}\sum_{\xv\in \calT}\ket{\cv(\xv)}\ket{\xv}.
    \]
    Compute $\cv(\xv)$ reversibly from $\xv$ and subtract it from the first register.
    This clears the histogram register and leaves $\ket{\calT}$.
    Taking the error of each preparation step to be at most $\epsilon/2$ gives the claimed bound.
\end{proof}

\begin{corollary}\label{cor:prepfolded}
    Assume $n\geq 2$.
    For any balanced functions $H_1, \ldots, H_{m}: \F_q^s \to [n]$, the preimage state $\ket{A_{H_1, \ldots, H_{m}}}$ of Definition~\ref{def:foldedsets} can be prepared up to error $\epsilon$ with $\poly(m, n, s, \log(1/\epsilon), \log q)$ gates and $O(m \sqrt{n})$ oracle queries to the $H_i$'s.
\end{corollary}
\begin{proof}
    Prepare $\ket{\calT}$ to error $\epsilon/2$ using Theorem~\ref{thm:Aprep}, and then apply Corollary~\ref{cor:decompress} with error $\epsilon/2$, taking the parameters $(M,N,\ell)$ of that corollary to be $(n,q^s/n,m)$ and identifying $\F_q^s$ with $[q^s]$.
\end{proof}

\begin{lemma}[Histogram synthesis]\label{lemma:histsup}
    For nonempty $\calT$ as in Definition~\ref{def:loadstates}, the following state can be prepared to within error $\epsilon$ with $\poly(m, n, \log(1/\epsilon))$ gates:
    \begin{align*}
        \ket{\psi_{\mathsf{hist}}} &= \frac{1}{\sqrt{|\calT|}} \sum_{\substack{\cv \in \ZZ_{\geq 0}^n \\ c_a \in \calB_a\ \forall a \in [n] \\ \sum_a c_a = m}} \sqrt{\binom{m}{\cv}} \ket{\cv}.
    \end{align*}
\end{lemma}
\begin{proof}
    We prepare the state using the Grover--Rudolph procedure~\cite{groverRudolph2002}. The required
    input is the ability to compute the total squared amplitude mass assigned
    to any prefix of the histogram $\cv$.
    For a prefix $c_1^*\in\calB_1,\ldots,c_k^*\in\calB_k$
    define $t=\sum_{a=1}^k c_a^*$.
    If \(t>m\), the prefix has zero mass. Otherwise, the total mass of all
    completions of this prefix is
    \begin{align*}
        \mathsf{Count}(c^*_1, \ldots, c^*_k) &:= \sum_{\substack{c_{k+1}, \ldots, c_n \geq 0 \\ c_a \in \calB_a\ \forall a > k\\ \sum_{a > k} c_a = m - t}} \binom{m}{\cv} \\
        &= \binom{m}{c^*_1, \ldots, c^*_k, m - t} \sum_{\substack{c_{k+1}, \ldots, c_n \geq 0 \\ c_a \in \calB_a\ \forall a > k\\ \sum_{a > k} c_a = m - t}} \binom{m-t}{c_{k+1}, \ldots, c_n} \\
        &= \binom{m}{c_1^*,\ldots,c_k^*,m-t}
        \mathsf{Table}_{k,t},
    \end{align*}
    where the second equality uses $\binom{m}{\cv}=\binom{m}{c_1^*,\ldots,c_k^*,m-t}\binom{m-t}{c_{k+1},\ldots,c_n}$, and we define
    \[
        \mathsf{Table}_{k,t}
        :=
        \sum_{\substack{
            c_{k+1},\ldots,c_n\ge 0\\
            c_a\in \calB_a\ \forall a>k\\
            \sum_{a>k}c_a=m-t
        }}
        \binom{m-t}{c_{k+1},\ldots,c_n}.
    \]

    We can compute $\mathsf{Table}_{k, t}$ for $k \in [0, n]$ and $t \in [0, m]$ efficiently with dynamic programming. The base case is $\mathsf{Table}_{n, t} = 1$ if $t = m$ and 0 otherwise.
    For \(k<n\), the recurrence is
    \[
        \mathsf{Table}_{k,t}
        =
        \sum_{\substack{
            c\in \calB_{k+1}\\
            0\le c\le m-t
        }}
        \binom{m-t}{c}
        \mathsf{Table}_{k+1,t+c},
    \]
    which follows by choosing $c_{k+1}=c$ and using $\binom{m-t}{c,c_{k+2},\ldots,c_n}=\binom{m-t}{c}\binom{m-t-c}{c_{k+2},\ldots,c_n}$.
    By the multinomial theorem, $\mathsf{Table}_{k,t}\leq(n-k)^{m-t}\leq n^m$, so each entry has $O(m\log n)$ bits.
    The $O(mn)$ entries can therefore be computed exactly, and reversibly computed and uncomputed, using polynomially many gates.

    Having computed this table, we prepare the histogram register
    coordinate-by-coordinate. Suppose we have already chosen
    \(c_1,\ldots,c_k\), with $t=\sum_{a=1}^k c_a.$
    The conditional probability of choosing \(c_{k+1}=c\) is
    \[
        \Pr[c_{k+1}=c\mid c_1,\ldots,c_k]
        =
        \frac{
            \mathbf 1[c\in \calB_{k+1}]
            \binom{m-t}{c}
            \mathsf{Table}_{k+1,t+c}
        }{
            \mathsf{Table}_{k,t}
        },
    \]
    where the numerator is interpreted as \(0\) if \(c>m-t\).
    On a zero-mass prefix, the controlled rotation may be defined arbitrarily because that branch has zero amplitude.
    Cumulative sums of these masses give the probabilities at each node of a binary preparation tree for $c_{k+1}$.
    The relevant ratios, square roots, and rotation angles can be computed to $\poly(\log(1/\epsilon),m,n)$ bits of precision by reversible arithmetic.
    There are polynomially many rotations, so approximating each to error $\epsilon/\poly(m,n)$ and uncomputing all work registers prepares the claimed state within error $\epsilon$.
\end{proof}

\begin{lemma}[Generalized Dicke state synthesis]\label{lemma:dicke}
    There exists a unitary that, for every histogram $\cv$ with entry sum $m$, maps $\ket{\cv}\ket{0}$ to $\ket{\cv}\ket{\psi_{\cv}}$, where
    \begin{align*}
        \ket{\psi_{\cv}} &= \frac{1}{\sqrt{\binom{m}{\cv}}} \sum_{\xv \in [n]^m: c_a(\xv) = c_a\ \forall a \in [n]} \ket{\xv},
    \end{align*}
    to within error $\epsilon$ with $\poly(m, n, \log(1/\epsilon))$ gates.
\end{lemma}
\begin{proof}
    Use the same binary-tree preparation procedure.
    Given a prefix $x^*_1, \ldots, x^*_k$ with prefix counts $c^*_a \leq c_a$, the conditional probability that the next symbol is $b$ is
    \[
        \frac{c_b-c^*_b}{m-k}.
    \]
    These probabilities and their cumulative sums have polynomial bit length and can be computed reversibly from $\cv$ and the prefix counts.
    Preparing each successive symbol to error $\epsilon/m$, updating the prefix counts, and uncomputing all work registers gives total error at most $\epsilon$ and polynomial gate complexity.
\end{proof}

\subsection{Stability}

This subsection computes the probability that replacing one coordinate of a uniform word of $\calT$ keeps it in $\calT$. The histogram of a uniform word of $\calT$ has the law of independent tilted-Poisson loads conditioned on their sum being $m$ (Lemma~\ref{lem:hard-wall-histogram-law}). Admissibility (Definition~\ref{def:admissible}) asks for a tilt $\nu$ under which the unconditioned sum has mean $m+O(1)$, together with a lattice condition, so that a local central limit theorem shows that the conditioning changes the relevant probabilities by only $O(n^{-1/2})$. In Theorem~\ref{thm:maininfinite}, the source load is size-biased because a uniformly chosen coordinate holds a symbol with probability proportional to its load. We prove the theorem in Appendix~\ref{app:maininfinite}.

\begin{definition}[Admissible load data]\label{def:admissible}
    Let $\eta > 0$ be a constant parameter; it will be the limit of $m/n$.
    Fix nonempty sets $B_1, \ldots, B_R \subseteq \ZZ_{\geq 0}$ and constants $\alpha_1, \ldots, \alpha_R \in \RR_{> 0}$ summing to $1$.
    We say that these are admissible if all of the following hold:
    \begin{itemize}
        \item Given $m$, membership circuits for $B_r\cap\{0,\ldots,m\}$ are constructible in $\poly(m)$ time and decide every input $0\leq k\leq m$ in $\poly(m)$ time, uniformly in $r$.
        \item There exists $\nu > 0$ such that the second coordinate $K$ of $(I,K)\sim\tpoiss_{\nuba}$ has mean $\eta$:
        $$\EE_{(I, K) \gets \tpoiss_{\nuba}} [K] = \eta.$$
        \item The additive subgroup generated by $\bigcup_{r\in[R]}(B_r-B_r)$ is all of $\ZZ$.
    \end{itemize}
    Whenever these data are admissible, we fix one such witnessing value of $\nu$; every tilted-Poisson quantity below uses this fixed witness.
\end{definition}

The membership requirement gives the computational access from Definition~\ref{def:loadstates}.
The witness $\nu$ and the infinite tilted-Poisson distributions are used only to analyze stability; the algorithms neither compute $\nu$ nor test admissibility.
The partition $S_1,\ldots,S_R$ of Definition~\ref{def:specialloadstates} below is given explicitly in each instance.

\begin{definition}[Special-form load sets]\label{def:specialloadstates}
    Fix load data $B_1,\ldots,B_R$, $\alpha_1,\ldots,\alpha_R$ and $\eta$ as in Definition~\ref{def:admissible}; admissibility itself is not required.
    A special-form load set is the following family of instances of Definition~\ref{def:loadstates}, indexed by $n$: fix a partition of $[n]$ into $S_1 \cup \ldots \cup S_R$ such that $|S_r| = \alpha_r n + O(1)$ for all $r$, let $m = \eta n + O(1)$, and for every $r \in [R]$ and $a \in S_r$, set $\calB_a=B_r\cap\{0,\ldots,m\}$.
\end{definition}

\begin{theorem}[One-step stability of load sets]\label{thm:maininfinite}
    Assume $B_1,\ldots,B_R$ and $\alpha_1,\ldots,\alpha_R$ are admissible (Definition~\ref{def:admissible}) with mean load $\eta$ and tilt $\nu$, and define $\calT\subseteq[n]^m$ as in Definition~\ref{def:specialloadstates}.
    For all sufficiently large $n$, the set $\calT$ is nonempty.
    Let $X$ be uniform on $\calT$, and obtain $X'$ by replacing a uniformly chosen coordinate of $X$ by an independent uniform symbol in $[n]$.
    Define independent random variables $(I_\src,K_\src)$ and $(I_\tar,K_\tar)$ by
    \[
        \Pr[I_\src=r\land K_\src=k]
        =\frac{\alpha_r}{\eta}\,k\Pr_{K\sim\tpoiss_{\nu,B_r}}[K=k],
        \qquad
        (I_\tar,K_\tar)\sim\tpoiss_{\nuba}.
    \]
    Then
    \[
        \Pr[X'\in \calT]
        =
        \Pr[K_\src-1\in B_{I_\src}]
        \Pr[K_\tar+1\in B_{I_\tar}]
        +O(n^{-1/2}).
    \]
    The implicit constant depends only on the fixed load data, the chosen tilt $\nu$, and the bounds in $|S_r|=\alpha_r n+O(1)$ and $m=\eta n+O(1)$.
\end{theorem}

\begin{corollary}\label{cor:efwfoldedstate}
    Let $H_1, \ldots, H_m: \F_q^s \to [n]$ be balanced functions, and define $A_{H_1, \ldots, H_m}$ from $\calT$ as in Definition~\ref{def:foldedsets}.
    Then, retaining notation from Theorem~\ref{thm:maininfinite}, we have:
    $$\EFW(\ket{A_{H_1, \ldots, H_m}}) \leq \EFW_s(\ket{A_{H_1, \ldots, H_m}}) = 1 - \Pr[K_{\src} - 1 \in B_{I_\src}] \cdot \Pr[K_{\tar} + 1 \in B_{I_\tar}] + O\left(\frac{1}{\sqrt{n}}\right).$$
\end{corollary}
\begin{proof}
    The first inequality follows from Claim~\ref{clm:monotonicity}.
    The equality follows by chaining:
    \begin{itemize}
        \item Corollary~\ref{cor:fourier-exit}: $\EFW_s(\ket{A_{H_1, \ldots, H_m}})$ is equal to the one-step exit probability of the coordinate replacement walk on $A_{H_1, \ldots , H_m}$.
        \item Lemma~\ref{lem:foldpreservesexit}: this is in turn equal to the one-step exit probability of the same walk on $\calT$.
        \item Theorem~\ref{thm:maininfinite}: this is equal to $1 - \Pr[K_{\src} - 1 \in B_{I_\src}] \cdot \Pr[K_{\tar} + 1 \in B_{I_\tar}] + O\left(\frac{1}{\sqrt{n}}\right)$, as desired.
    \end{itemize}
\end{proof}

\section{Oracle Separation for Load States}\label{sec:yzoracle}

This section adapts the random-oracle separation of Yamakawa and Zhandry~\cite{yamakawa2024verifiable} to load constraints. The code is a Reed--Solomon code whose coordinates are grouped into blocks of length $s$, and the load constraint applies after each block is relabeled by an independent random function $H_i:\F_q^s\to[n]$. We give an efficient quantum algorithm (Theorem~\ref{thm:yzpositive}), show that classical algorithms making polynomially many queries succeed only with probability $2^{-\Omega(\lambda)}$ (Lemma~\ref{lem:yzclassicalsoundness}), and derive proofs of quantumness (Theorems~\ref{thm:yzpoq} and~\ref{thm:yzworstcasepoq}).

Let $\lambda\geq 2$ be the integer security parameter, and take all logarithms in this section to base $2$.
Whenever the load data are admissible, $\nu$ is the tilt fixed in Definition~\ref{def:admissible}.

\begin{definition}[Parameters and the code $C_\lambda$]\label{def:clambda}
    Let $q = 2^{2 \lfloor \log \lambda \rfloor}$, $M = q-1$, $s = 2^{\lfloor \log \lambda \rfloor} + 1$, $n = 2^{\lfloor \log \lambda \rfloor}$, $k=\lfloor\rho M\rfloor$ for a constant rate parameter $\rho \in (0, 1)$, and let $\gamma$ be an arbitrary generator of $\F_q^*$.
    We define
    \[
        C_\lambda=\RS_{\F_q,\gamma,k}\subseteq\F_q^M.
    \]
    We group its $M$ coordinates into $m:=M/s=n-1$ consecutive blocks of size $s$, using the fixed identification $\F_q^M=(\F_q^s)^m$.
    The code in the intersection problem and its dual remain the $\F_q$-linear codes $C_\lambda$ and $C_\lambda^\perp$.
    The quantum algorithm decodes $C_\lambda^\perp$ in the block Hamming metric using Lemma~\ref{lem:frscapacity}; the classical lower bound uses folded list recovery for $C_\lambda$.
\end{definition}

\begin{definition}[The oracle distribution $\calD_\lambda$]\label{def:Dlambda}
    Fix sets $B_1, \ldots, B_R \subseteq \ZZ_{\geq 0}$ and constants $\alpha_1, \ldots, \alpha_R \in \RR_{> 0}$ summing to $1$.
    Then define $S_1$ to consist of the first $\lfloor \alpha_1 n \rfloor$ elements of $[n]$, $S_2$ the next $\lfloor \alpha_2 n \rfloor$ elements, and so on until $S_{R-1}$, and then let $S_R$ consist of the remaining elements of $[n]$.
    Let $\calT \subseteq [n]^m$ be the resulting special-form load set; since $m=n-1$, it has $\eta=1$. The load data need not be admissible here.
    A sample from $\calD_\lambda$ is a tuple of independent uniformly random functions $H_1,\ldots,H_m:\F_q^s\to[n]$.
    The corresponding oracle-represented code-intersection instance consists of the public descriptions of $C_\lambda$, the partition, and the truncated-load membership circuits from Definition~\ref{def:loadstates}, together with quantum oracle access to the $H_i$'s.
    An output succeeds precisely when it belongs to $C_\lambda\cap A_{H_1,\ldots,H_m}$.
    If the latter set is empty, the instance has no successful output.
\end{definition}

\subsection{Quantum Algorithm}

\begin{theorem}\label{thm:yzpositive}
    Assume the load data of $\calD_\lambda$ are admissible (Definition~\ref{def:admissible}) with $\eta=1$.
    With the notation from Theorem~\ref{thm:maininfinite}, suppose that the rate parameter $\rho$ of $C_\lambda$ (Definition~\ref{def:clambda}) satisfies
    $$\rho > 1 - \Pr[K_{\src} - 1 \in B_{I_\src}] \cdot \Pr[K_{\tar} + 1 \in B_{I_\tar}].$$
    There is a uniform quantum algorithm which, for all sufficiently large $\lambda$ and given an oracle-represented instance sampled from $\calD_\lambda$, runs in time $\poly(\lambda)$ and makes $\poly(\lambda)$ queries to the $H_i$'s.
    Its success probability over the randomness of $H_1, \ldots, H_m$ and the algorithm's internal randomness is $\Omega(1)$.
\end{theorem}
\begin{proof}
    By Theorem~\ref{thm:maininfinite}, $\calT$ is nonempty for all sufficiently large $\lambda$.

    We next prove the theorem in the setting where the $H_i$'s are sampled as uniformly random balanced functions from $\F_q^s \to [n]$.
    Such functions exist because $n$ divides $q^s$.
    Balancedness and nonemptiness of $\calT$ imply that every preimage set is nonempty.
    In this setting, we will simply instantiate Corollary~\ref{cor:ctlistdecode}.
    For brevity, write
    \[
        p_{\mathsf{stay}}
        =
        \Pr[K_{\src} - 1 \in B_{I_\src}]
        \Pr[K_{\tar} + 1 \in B_{I_\tar}].
    \]
    Fix a constant $\epsilon>0$ such that
    $1-p_{\mathsf{stay}}<\zeta:=\rho-\epsilon$.
    Such a choice exists by the hypothesis.
    We check each requirement:
    \begin{itemize}
        \item By Corollary~\ref{cor:prepfolded}, $\ket{A_{H_1, \ldots, H_m}}$ can be prepared to precision $\epsilon'$ with $\poly(\lambda,\log(1/\epsilon'))$ gates and oracle queries.

        \item The distribution of represented instances, including the full tuple of balanced function oracles, is translation-invariant under every $\vv=(\vv_1,\ldots,\vv_m)\in(\F_q^s)^m$.
        Indeed, define $H_i^{(\vv)}(\xv)=H_i(\xv-\vv_i)$.
        Precomposition by a translation preserves the uniform distribution on balanced functions, and
        \[
            A_{H_1^{(\vv)},\ldots,H_m^{(\vv)}}
            =
            A_{H_1,\ldots,H_m}+\vv.
        \]
        Oracle access to each $H_i^{(\vv)}$ is simulated with one query to $H_i$.

        \item Fix any balanced $H_1,\ldots,H_m$.
        By Corollary~\ref{cor:efwfoldedstate}, $\EFW_s(\ket{A_{H_1,\ldots,H_m}})$ is at most $1-p_{\mathsf{stay}}+O(n^{-1/2})$, where the error term does not depend on the $H_i$.
        By Corollary~\ref{cor:fourier-mass}, measuring $\QFT^\dag\ket{A_{H_1,\ldots,H_m}}$ therefore gives block weight at most $\zeta m$ with probability at least
        \[
            1-\frac{1-p_{\mathsf{stay}}+O(n^{-1/2})}{\zeta}
            \geq
            \frac{\zeta-(1-p_{\mathsf{stay}})}{2\zeta}
            =
            \Omega(1)
        \]
        for all sufficiently large $n$, because $\zeta-(1-p_{\mathsf{stay}})>0$ is a constant.
        Since this holds for every balanced tuple, $P_\sparse=\Omega(1)$.

        \item Lemma~\ref{lem:grsduality} gives
        $C_\lambda^\perp=\GRS_{\F_q,\gamma,M-k,\uv}$
        for explicit nonzero multipliers $\uv$.
        Its actual rate is $R^\perp=(M-k)/M=1-\rho+O(M^{-1})$.
        For all sufficiently large $M$, $R^\perp\leq1-\rho+\epsilon/2$, and hence
        \[
            \zeta\leq1-R^\perp-\epsilon/2.
        \]
        Apply Lemma~\ref{lem:frscapacity} with slack $\epsilon/2$.
        Its folding requirement holds for all sufficiently large $\lambda$, since $s=\Theta(\lambda)$ and the slack is fixed.
        Thus the dual is list-decodable from $\zeta m$ block errors with a final list of size $L=O(1/\epsilon)$.
        The time to enumerate and filter the affine candidate space is $q^{O(1/\epsilon)}\poly(M)=\poly(\lambda)$, because $q,M=\Theta(\lambda^2)$.
    \end{itemize}
    The bounded-time version of Corollary~\ref{cor:ctlistdecode} now gives an algorithm whose success probability is at least $P_\sparse/L-2^{-\Omega(q+M)}=\Omega(1)$ and whose running time and oracle-query complexity are $\poly(\lambda)$.

    Append an efficient verifier to this algorithm.
    On a candidate $\xv\in\F_q^M$, it checks $\xv\in C_\lambda$ using a parity-check matrix, queries $H_i(\xv_i)$ for every block, computes the resulting histogram, and checks the relevant $B_r$-membership conditions.
    It accepts exactly when $\xv\in C_\lambda\cap A_{H_1,\ldots,H_m}$ and uses only $m=\poly(\lambda)$ additional queries.
    If $Q=\poly(\lambda)$ bounds the total queries of the composed algorithm, apply Lemma~\ref{lem:distinguishbalanced} at each step of a hybrid over the $m$ oracles, conditioning on the other independent oracles.
    The composed algorithm's acceptance probability changes by at most
    \[
        O\left(\frac{mQ^3}{q^s}\right)
        =
        \negl(\lambda)
    \]
    when the balanced functions are replaced by independent uniformly random functions.
    For either oracle distribution, acceptance is exactly the code-intersection success event, so the success probability for $\calD_\lambda$ remains $\Omega(1)$.

    Finally, instances with $A_{H_1,\ldots,H_m}=\emptyset$ have no successful output (Definition~\ref{def:Dlambda}), and they occur with negligible probability.
    Fix any $\yv\in \calT$.
    If $A_{H_1,\ldots,H_m}$ is empty, then $y_i$ is absent from the range of some $H_i$, and hence
    \[
        \Pr[A_{H_1,\ldots,H_m}=\emptyset]
        \leq
        m(1-1/n)^{q^s}.
    \]
    The right-hand side is $\negl(\lambda)$, consistently with the preceding success bound.
\end{proof}

\subsection{Classical Hardness}

This subsection proves the two facts used in the classical soundness proof (Lemma~\ref{lem:yzclassicalsoundness}). A classical algorithm learns $H_i(x_i)$ only by querying it. List recovery of the folded code (Lemma~\ref{lem:frslistrecovery}) bounds the number of codewords most of whose blocks have been queried. Resistance to random completion (Lemma~\ref{lem:prangeresistance}) shows that each such codeword is unlikely to satisfy the load constraint, because the oracle values on its remaining blocks are still uniformly random.
\begin{lemma}[List recovery of the folded code]\label{lem:frslistrecovery}
    Let $\rho$ be the rate parameter of $C_\lambda$, and fix constants $0<c<c'<1$ and $0<\zeta<1-\rho$.
    For all sufficiently large $\lambda$, the folded code $C_\lambda^{(s)}$ is $(\zeta,\ell,L)$-list-recoverable for $\ell=\lceil 2^{\lambda^c}\rceil$ and $L=2^{\tilde{O}(\lambda^{c'})}$.
\end{lemma}
\begin{proof}
    Apply Theorem~\ref{thm:rudra} with $d=\lceil\lambda^{c'}\rceil$ and $r=\lceil\lambda^{c''}\rceil$, where $c'<c''<1$.
    We have $d\leq s$ for all sufficiently large $\lambda$, and
    \[
        \ell^{1/(d+1)}=1+o(1),\qquad
        \left(\frac{q-1}{k-1}\right)^{1/(d+1)}=1+o(1),\qquad
        \frac{d}{r}=o(1),\qquad
        \frac{s}{s-d+1}=1+o(1).
    \]
    After multiplying~\eqref{eq:gr1} by $s$, its right-hand side is therefore
    \[
        \left(1+\frac{d}{r}\right)
        \sqrt[d+1]{(q-1)\ell(k-1)^d}\,
        \frac{s}{s-d+1}
        =(1+o(1))(k-1).
    \]
    Since $k/(q-1)=\rho+o(1)$ and $1-\zeta>\rho$, inequality~\eqref{eq:gr1} holds for all sufficiently large $\lambda$.
    Similarly,
    \[
        (r+d)\sqrt[d+1]{\frac{(q-1)\ell}{k-1}}
        =(1+o(1))\lambda^{c''}<q,
    \]
    so~\eqref{eq:gr2} also holds.
    The resulting list size is at most
    $q^d=2^{O(d\log\lambda)}=2^{\tilde{O}(\lambda^{c'})}$.
\end{proof}

\begin{definition}[Tail thresholds]\label{def:tailthreshold}
    For load data $B_1,\ldots,B_R$ and $\alpha_1,\ldots,\alpha_R$, the \emph{tail threshold} of $B_r$ is
    \[
        \tau_r
        :=
        \min\{t\in\ZZ_{\geq0}:\{t,t+1,\ldots\}\subseteq B_r\},
    \]
    with $\tau_r=\infty$ if the set on the right is empty.
    The \emph{weighted tail threshold} is $\tau=\sum_{r\in[R]}\alpha_r\tau_r\in[0,\infty]$.
\end{definition}
The following lemma gives the property of the load constraint used in the proof of soundness.
Call a symbol of class $r$ \emph{safe} once its count is at least $\tau_r$: it then stays allowed however the remaining coordinates are filled.
The hypothesis $\tau>(1-\epsilon)\eta$ says that $(1-\epsilon)m$ fixed coordinates cannot make every symbol safe.
So linearly many symbols are not yet safe, each with a forbidden count a bounded distance above its current count, and each reaches that count with constant probability.

\begin{lemma}[Resistance to random completion]\label{lem:prangeresistance}
    Let $\calT\subseteq[n]^m$ be a special-form load set (Definition~\ref{def:specialloadstates}), so that $m=\eta n+O(1)$, and let $\tau$ be the weighted tail threshold of its load data (Definition~\ref{def:tailthreshold}).
    Fix a constant $0<\epsilon\leq1$ and suppose
    $\tau>(1-\epsilon)\eta$.
    Let $\xv\in([n]\cup\{\bot\})^m$ be arbitrary with at least $\epsilon m$ coordinates equal to $\bot$.
    Independently replace every $\bot$ by a uniformly random element of $[n]$ to obtain $\zv\in[n]^m$.
    Then
    \[
        \Pr[\zv\in \calT]\leq 2^{-\Omega(m)}.
    \]
    The constant implicit in $\Omega(m)$ may depend on the fixed load data and on $\epsilon$, but the bound is uniform over $\xv$.
\end{lemma}
\begin{proof}
    Let $h$ be the number of $\bot$ coordinates and set $t=\lceil\epsilon m\rceil$.
    If $h>t$, first fill any $h-t$ of the missing coordinates at random and condition on their values.
    The resulting partial word has exactly $t$ missing coordinates, and a bound uniform over that partial word remains valid after averaging.
    It therefore suffices to consider $h=t$.

    By the hypothesis on $\tau$, there are finite integers $u_1,\ldots,u_R\geq0$ such that $u_r\leq\tau_r$ whenever $\tau_r<\infty$ and
    \[
        \sum_{r\in[R]}\alpha_ru_r>(1-\epsilon)\eta.
    \]
    Fix $\Delta>0$ such that the left-hand side equals $(1-\epsilon)\eta+\Delta$, and let $U=\max_r u_r$.
    For $a\in S_r$, let $c_a$ be the number of occurrences of $a$ among the nonmissing coordinates of $\xv$, and call $a$ vulnerable if $c_a<u_r$.
    If $v_r$ is the number of vulnerable symbols in $S_r$ and $v=\sum_rv_r$, then
    \begin{align*}
        m-t
        =\sum_{a\in[n]}c_a
        &\geq\sum_{r\in[R]}u_r(|S_r|-v_r)\\
        &\geq n((1-\epsilon)\eta+\Delta)-Uv+O(1).
    \end{align*}
    On the other hand, $m-t\leq(1-\epsilon)m=(1-\epsilon)\eta n+O(1)$.
    Thus $v\geq\Delta n/(2U)$ for all sufficiently large $n$; note that $U>0$ because $\sum_r\alpha_ru_r>(1-\epsilon)\eta\geq0$.

    For every $r$ and every integer $c<u_r$, the definition of $\tau_r$ gives an integer $T_{r,c}\geq c$ outside $B_r$.
    There are only finitely many relevant pairs $(r,c)$, so all the $T_{r,c}$ are bounded by a constant $D$.
    For each vulnerable $a\in S_r$, set $d_a=T_{r,c_a}-c_a\in\{0,\ldots,D\}$.
    If the random completion adds exactly $d_a$ occurrences of $a$, then the final count of $a$ is outside $B_r$.

    Let $Y_a$ be the number of missing coordinates completed with $a$. 
    The vector $(Y_a)_{a\in[n]}$ has the multinomial law of $t$ independent uniform draws from $[n]$; we compare it with independent Poisson variables.
    If $X_1,\ldots,X_n$ are independent Poisson random variables of mean $\theta=t/n$, then the standard Poisson--multinomial identity gives
    \[
        (Y_1,\ldots,Y_n)
        \mathrel{\stackrel{\mathrm{d}}{=}}
        (X_1,\ldots,X_n)
        \text{ conditioned on }\sum_{a\in[n]}X_a=t.
    \]
    Consequently,
    \begin{align*}
        \Pr[\zv\in \calT]
        &\leq
        \Pr[Y_a\neq d_a\text{ for every vulnerable }a]\\
        &\leq
        \frac{\prod_{a\text{ vulnerable}}\Pr[X_a\neq d_a]}
        {\Pr[\sum_{a\in[n]}X_a=t]}.
    \end{align*}
    The second inequality uses $\Pr[E\mid F]\leq\Pr[E]/\Pr[F]$ and the independence of the $X_a$.
    Since $\theta=\epsilon\eta+O(1/n)$ and $d_a\in\{0,\ldots,D\}$, there is a constant $p_*>0$ such that $\Pr[X_a=d_a]\geq p_*$ for every vulnerable $a$ and all sufficiently large $n$.
    The numerator is therefore at most $(1-p_*)^v=2^{-\Omega(n)}$.
    Moreover, $\sum_aX_a$ is Poisson with mean $t$, and Stirling's formula gives
    \[
        \Pr\left[\sum_{a\in[n]}X_a=t\right]
        =e^{-t}\frac{t^t}{t!}
        =\Theta(t^{-1/2}).
    \]
    Since $m=\eta n+O(1)$ and $t=\Theta(n)$, dividing by this polynomially small probability preserves the bound $2^{-\Omega(m)}$.
\end{proof}

In particular, if $\ZZ_{\geq0}\setminus B_r$ is infinite for some $r$, then $\tau_r=\infty$ and the hypothesis of Lemma~\ref{lem:prangeresistance} holds for every constant $\epsilon>0$.

\subsection{Proof of Quantumness}

In a proof of quantumness, a classical verifier accepts an efficient quantum prover with noticeable probability, and accepts any classical prover making polynomially many oracle queries with negligible probability. Our protocols are noninteractive, publicly verifiable, and keyless. The protocol $\Pi_{\mathsf{base}}$ below has completeness on average over the oracle (Theorem~\ref{thm:yzpoq}); Section~\ref{sec:yzworstcase} gives a protocol whose completeness holds for every fixed oracle.

\begin{definition}[The protocol $\Pi_{\mathsf{base}}$]\label{def:pibase}
    In the keyless noninteractive protocol $\Pi_{\mathsf{base}}$, the quantum prover runs the algorithm of Theorem~\ref{thm:yzpositive} and sends its output $\xv\in(\F_q^s)^m$.
    On failure, it sends a fixed dummy string.
    The classical verifier checks that the unfolded vector $\xv$ belongs to $C_\lambda$, queries $H_i(x_i)$ for every $i\in[m]$, and accepts exactly when
    \[
        (H_1(x_1),\ldots,H_m(x_m))\in \calT.
    \]
\end{definition}
Thus the verifier accepts exactly the successful outputs of Definition~\ref{def:Dlambda}.

\begin{lemma}[Classical soundness]\label{lem:yzclassicalsoundness}
    Let $\rho$ be the rate parameter of $C_\lambda$ and $\tau$ the weighted tail threshold of the load data of $\calD_\lambda$ (Definition~\ref{def:tailthreshold}).
    Suppose $\eta=1$ and there are constants
    \[
        0<\epsilon<\zeta<1-\rho
        \qquad\text{such that}\qquad
        \tau>1-\epsilon.
    \]
    Every computationally unbounded classical prover making polynomially many classical queries to the $H_i$'s convinces the verifier of $\Pi_{\mathsf{base}}$ (Definition~\ref{def:pibase}) with probability at most $2^{-\Omega(\lambda)}$ over the uniformly random functions $H_1,\ldots,H_m$.
\end{lemma}
\begin{proof}
    This follows the online list-recovery argument of~\cite[Lemma~6.9]{yamakawa2024verifiable}.
    Let $\mathcal A$ make $Q=\poly(\lambda)$ classical queries.
    We may append the verifier's $m$ queries to $\mathcal A$: after $\mathcal A$ outputs $\xv$, query every $H_i(x_i)$ that has not yet been queried.
    This changes the query bound by only $m=\poly(\lambda)$ and ensures that an accepted output has eventually been queried in every block.

    Fix constants $0<c<c'<1$ such that $Q+m\leq2^{\lambda^c}$ for all sufficiently large $\lambda$.
    For each $i\in[m]$ and each time $j$, let $Q_i^{(j)}\subseteq\F_q^s$ be the set of inputs queried to $H_i$ during the first $j$ total oracle queries.
    A folded codeword $\xv\in C_\lambda^{(s)}$ is $N$-queried at time $j$ if $x_i\in Q_i^{(j)}$ for exactly $N$ indices $i$.
    A single query changes this number by at most one.
    Hence any codeword that is eventually queried in every block must at some point become $N_0$-queried, where $N_0:=\lceil(1-\zeta)m\rceil$.

    Since the sets $Q_i^{(j)}$ only grow, every codeword that ever becomes $N_0$-queried has $x_i$ in the final set $Q_i$ for at least $(1-\zeta)m$ indices.
    Each final set has size at most $Q+m\leq\ell$, where $\ell=\lceil2^{\lambda^c}\rceil$, so Lemma~\ref{lem:frslistrecovery} implies that at most
    $L=2^{\tilde{O}(\lambda^{c'})}$ codewords ever become $N_0$-queried.
    Enumerate them in the order in which they first reach this threshold, breaking ties lexicographically.

    Simulate the random functions by lazy sampling and condition on the complete transcript at the time one such candidate first becomes $N_0$-queried.
    The candidate and the already sampled values $H_i(x_i)$ are then fixed.
    For the remaining $m-N_0=\lfloor\zeta m\rfloor$ blocks, the values $H_i(x_i)$ are still independent and uniform in $[n]$.
    Since $\epsilon<\zeta$, there are at least $\epsilon m$ such blocks for all sufficiently large $\lambda$.
    Lemma~\ref{lem:prangeresistance} shows, uniformly over the conditioned transcript, that the candidate becomes a valid proof with probability at most $2^{-\Omega(m)}$.
    The event that this candidate is eventually output and accepted is a subset of that event.

    Union bounding over the at most $L$ candidates gives
    \[
        \Pr[\mathcal A\text{ convinces the verifier}]
        \leq
        L\cdot2^{-\Omega(m)}
        =
        2^{-\Omega(\lambda)},
    \]
    because $m=\Theta(\lambda)$ and $\tilde{O}(\lambda^{c'})=o(\lambda)$.
\end{proof}

\begin{theorem}[Proof of quantumness]\label{thm:yzpoq}
    Assume that $B_1,\ldots,B_R$ and $\alpha_1,\ldots,\alpha_R$ are admissible (Definition~\ref{def:admissible}) with $\eta=1$.
    Let
    \[
        p_{\mathsf{stay}}
        =
        \Pr[K_{\src}-1\in B_{I_{\src}}]\,
        \Pr[K_{\tar}+1\in B_{I_{\tar}}],
    \]
    using the distributions from Theorem~\ref{thm:maininfinite}, and let $\tau=\sum_r\alpha_r\tau_r$ be the weighted tail threshold (Definition~\ref{def:tailthreshold}).
    Suppose
    \[
        1-p_{\mathsf{stay}}<\min\{\tau,1\}.
    \]
    Choose the rate parameter $\rho$ of $C_\lambda$ (Definition~\ref{def:clambda}) to satisfy
    \[
        1-p_{\mathsf{stay}}<\rho<\min\{\tau,1\}.
    \]
    Then $\Pi_{\mathsf{base}}$ (Definition~\ref{def:pibase}) is a keyless noninteractive publicly verifiable proof of quantumness with constant completeness and negligible classical soundness: the honest quantum prover is accepted with probability $\Omega(1)$, whereas every computationally unbounded classical prover making polynomially many classical queries is accepted with probability $2^{-\Omega(\lambda)}$.
\end{theorem}
\begin{proof}
    The rate choice and Theorem~\ref{thm:yzpositive} give quantum completeness $\Omega(1)$.
    Since $\rho<\min\{\tau,1\}$, we can choose $\epsilon$ with $\max\{0,1-\tau\}<\epsilon<1-\rho$, and then $\zeta\in(\epsilon,1-\rho)$; Lemma~\ref{lem:yzclassicalsoundness} then gives classical soundness.
\end{proof}

If some $\ZZ_{\geq0}\setminus B_r$ is infinite, then $\tau=\infty$.
In this useful special case, the hypothesis $1-p_{\mathsf{stay}}<\min\{\tau,1\}$ of Theorem~\ref{thm:yzpoq} reduces to $p_{\mathsf{stay}}>0$, and one may choose any fixed rate $\rho$ strictly between $1-p_{\mathsf{stay}}$ and $1$.

\subsection{Worst-Case Completeness}\label{sec:yzworstcase}

The completeness in Theorem~\ref{thm:yzpoq} is averaged over the random oracle.
We now follow the final construction of~\cite[Section~6]{yamakawa2024verifiable} to obtain a protocol whose completeness holds for every fixed oracle.
The prover masks the oracle by adding a function $f_{\mathsf{key}}$ from a family with limited independence, choosing $\mathsf{key}$ itself and sending it with the proof.
The independence of the family is at least twice the number of queries in one trial, counting the honest prover and its verification, so in each trial the masked oracle is indistinguishable from a uniformly random one, whatever the original oracle is.
Since the base completeness is a constant rather than $1-\negl(\lambda)$ as in~\cite{yamakawa2024verifiable}, each trial amplifies it using $u$ independent copies of the masked oracle.

For soundness, a cheating prover may choose $\mathsf{key}$ after querying the oracle.
We therefore run $t$ trials with one key, and show that for each fixed key a classical prover succeeds with probability $2^{-\Omega(t\lambda)}$.
The key length does not depend on $t$, so a union bound over all keys gives negligible soundness once $t$ is a large enough polynomial.

\begin{theorem}[Proof of quantumness with worst-case completeness]\label{thm:yzworstcasepoq}
    Under the hypotheses of Theorem~\ref{thm:yzpoq}, there is a keyless noninteractive publicly verifiable protocol relative to a random oracle with the following properties.
    For every fixed oracle, the honest quantum prover is accepted with probability $1-\negl(\lambda)$.
    For a uniformly random oracle, every computationally unbounded classical prover making polynomially many classical queries is accepted with probability at most $\negl(\lambda)$.
\end{theorem}
\begin{proof}
    Fix a constant $p_0>0$ that lower bounds the completeness of $\Pi_{\mathsf{base}}$ for all sufficiently large $\lambda$, and set $u=\lambda$.
    An \emph{amplified trial} takes $u$ independent tuples of oracle functions, runs the base prover independently on each, and checks each answer with the base verifier.
    It outputs $(a,\xv)$ for the first index $a\in[u]$ whose answer $\xv$ is accepted, or a fixed pair that the verifier rejects if there is none.
    For a uniformly random oracle, the probability that all $u$ attempts fail is at most
    $(1-p_0)^u=2^{-\Omega(\lambda)}$.

    Let $Q_0=\poly(\lambda)$ upper bound the total number of queries made by one amplified trial together with its verification.
    We use the limited-independence lemma of~\cite[Lemma~2.5]{yamakawa2024verifiable}: the acceptance probability of a quantum algorithm making at most $Q_0$ queries is the same whether its oracle is a uniformly random function or a uniformly random member of a $2Q_0$-wise independent family.
    Identify $[n]$ with the additive group $\ZZ/n\ZZ$, and let
    \[
        \mathcal F_\lambda
        =
        \{f_{\mathsf{key}}:[u]\times[m]\times\F_q^s\to[n]\}_{\mathsf{key}\in\mathsf{Keys}_\lambda}
    \]
    be an efficiently computable $2Q_0$-wise independent family.
    Such a family has key length $\poly(\lambda)$: encode the domain injectively in a sufficiently large binary extension field, evaluate a random polynomial of degree $2Q_0-1$, and project the output onto $\log n$ fixed binary coordinates (recall that $n$ is a power of $2$).
    The family does not depend on the number $t$ of trials, introduced next: its domain has no trial index, and $Q_0$ counts the queries of a single trial.
    Hence the key length does not depend on $t$, as the final union bound over keys requires.

    Let $t=\poly(\lambda)$ be chosen below, and regard the random oracle as a function
    \[
        H:[t]\times[u]\times[m]\times\F_q^s\to[n].
    \]
    For $j\in[t]$, $a\in[u]$, and $i\in[m]$, define the \emph{masked function}
    \[
        G_i^{j,a,\mathsf{key}}(x)
        =
        H(j,a,i,x)+f_{\mathsf{key}}(a,i,x)\pmod n,
    \]
    and call the tuple $(G_i^{j,a,\mathsf{key}})_{i\in[m]}$ the \emph{slice} $(j,a)$.
    The prover samples $\mathsf{key}$ uniformly.
    For every $j\in[t]$, using a fresh workspace and fresh internal randomness, it runs an amplified trial on the $u$ slices $(j,a)$ for $a\in[u]$, and obtains a pair $(a_j,\xv^{(j)})$.
    It sends $\mathsf{key}$ and all $t$ pairs to the verifier.
    The verifier accepts if and only if, for every $j$, the base verifier accepts $\xv^{(j)}$ relative to the slice $(j,a_j)$.
    Both algorithms are uniform and use polynomially many gates and oracle queries for polynomial $t$.

    We first prove completeness for an arbitrary fixed $H$.
    Fix $j$.
    With $H$ fixed, the $j$th amplified trial and its verification form a $Q_0$-query algorithm whose oracle is $f_{\mathsf{key}}$: each query to a masked function $G_i^{j,a,\mathsf{key}}$ is one query to $f_{\mathsf{key}}$ plus a fixed value of $H$.
    By the limited-independence lemma, its acceptance probability is unchanged if $f_{\mathsf{key}}$ is replaced by a uniformly random function.
    After this replacement, the slices $(j,a)$ for $a\in[u]$ are independent tuples of uniformly random functions, so the $j$th trial fails with probability at most $(1-p_0)^u=2^{-\Omega(\lambda)}$.
    A union bound over the $t=\poly(\lambda)$ trials shows that all of them pass with probability $1-\negl(\lambda)$.

    It remains to prove soundness for uniformly random $H$.
    First fix $\mathsf{key}\in\mathsf{Keys}_\lambda$, and let $E_{\mathsf{key}}$ be the event that the prover's proof contains this key and the verifier accepts; we show that $\Pr[E_{\mathsf{key}}]\leq2^{-\Omega(t\lambda)}$.
    Translation by the fixed function $f_{\mathsf{key}}$ maps the random oracle $H$ to a collection of independent uniformly random functions $G_i^{j,a,\mathsf{key}}$, and each query to $H$ reveals one value of one masked function.
    Let a classical prover make $Q=\poly(\lambda)$ queries, including the $tm$ verifier queries that we append after it outputs.
    Fix constants $0<\epsilon<\zeta<1-\rho$ with $1-\epsilon<\tau$, as in the proof of Theorem~\ref{thm:yzpoq}.
    Fix any constants $0<c<c'<1$; then $Q\leq2^{\lambda^c}$ for all sufficiently large $\lambda$, and we use the list size
    $L=2^{\tilde{O}(\lambda^{c'})}$ from Lemma~\ref{lem:frslistrecovery}.

    Apply the threshold argument from the proof of Lemma~\ref{lem:yzclassicalsoundness} separately to every slice $(j,a)$: a codeword reaches the threshold in slice $(j,a)$ when $N_0=\lceil(1-\zeta)m\rceil$ of its blocks $x_i$ have been queried as $(j,a,i,x_i)$.
    For each slice, at most $L$ codewords ever reach the threshold.
    If the verifier accepts, then for each $j$ the output $\xv^{(j)}$ is a codeword all of whose blocks have been queried in slice $(j,a_j)$, so it reached the threshold there.
    Hence each pair $(a_j,\xv^{(j)})$ is among at most $uL$ candidates.
    Index the candidates within each slice by the order in which they reach the threshold, breaking ties lexicographically, and fix a tuple of indices $(a_j,b_j)_{j\in[t]}\in([u]\times[L])^t$.
    Let $T_j$ be the stopping time at which the $b_j$th candidate in slice $(j,a_j)$ reaches the threshold, with $T_j=\infty$ if it never does.
    Call this candidate \emph{valid} if the base verifier accepts it relative to slice $(j,a_j)$.
    We bound the probability that every $T_j$ is finite and every indexed candidate is valid.
    We do not condition on the candidates existing: that event depends on later queries, and conditioning on it could bias oracle values that are still unqueried.

    Run the lazy-sampling experiment and handle the indices $j$ one at a time, starting with none handled.
    Stop at the smallest finite $T_j$ among the unhandled indices, breaking ties by $j$.
    At that time the candidate is determined by the transcript.
    Since a query raises the number of queried blocks of a codeword by at most one, the candidate has exactly $N_0$ queried blocks in slice $(j,a_j)$.
    Its remaining $\lfloor\zeta m\rfloor\geq\epsilon m$ values there are unqueried, hence independent and uniform given everything seen so far.
    Lemma~\ref{lem:prangeresistance} therefore bounds the conditional probability that the candidate is valid by $\delta_\lambda=2^{-\Omega(m)}$, uniformly over the transcript.
    Now reveal these remaining values, add them to the information we condition on, and mark $j$ as handled; later queries to these points are answered with the revealed values.
    The revealed values belong to slice $(j,a_j)$.
    Every unhandled $j'$ has its candidate in slice $(j',a_{j'})$, and since $j'\neq j$, the two slices use disjoint points of $H$.
    So the unqueried values in the slices of the unhandled candidates remain independent and uniform.
    If no unhandled candidate ever reaches the threshold, the event fails; otherwise we repeat.
    Induction and the tower property give probability at most $\delta_\lambda^t=2^{-\Omega(tm)}$ that all $t$ indexed candidates exist and are valid.

    If $E_{\mathsf{key}}$ occurs, then for some tuple $(a_j,b_j)_{j\in[t]}$, each output $\xv^{(j)}$ is the $b_j$th candidate of slice $(j,a_j)$ and is valid.
    A union bound over the at most $(uL)^t$ tuples gives
    \[
        \Pr[E_{\mathsf{key}}]
        \leq
        (uL)^t2^{-\Omega(tm)}
        =
        2^{-\Omega(t\lambda)},
    \]
    where the equality holds because $\log(uL)=o(\lambda)$ and $m=\Theta(\lambda)$.

    Finally, a cheating prover may choose the key appearing in its proof.
    Since the acceptance event is the union of the events $E_{\mathsf{key}}$ over $\mathsf{key}\in\mathsf{Keys}_\lambda$, another union bound gives soundness at most
    $|\mathsf{Keys}_\lambda|2^{-\Omega(t\lambda)}$.
    The constant in $2^{-\Omega(t\lambda)}$ does not depend on $t$, because $\delta_\lambda$, $u$, $L$ and $m$ do not; $t$ affects only how large $\lambda$ must be.
    Since $\log|\mathsf{Keys}_\lambda|=\poly(\lambda)$ independently of $t$, we may choose a sufficiently large polynomial $t$ so that this quantity is negligible.
\end{proof}

\section{Ordinary Reed--Solomon Codes in the Plain Model}\label{sec:plainrs}

The large coordinate domains and oracle access in Section~\ref{sec:yzoracle} are used to prove classical soundness, not to obtain the quantum algorithm.
We now give a plain-model variant in which there is no nontrivial folding, every function is represented by its complete truth table, and the code is an ordinary Reed--Solomon code.

\begin{definition}[Plain-model ordinary-RS instances]\label{def:plainrs}
    Fix admissible load data $B_1,\ldots,B_R$ and $\alpha_1,\ldots,\alpha_R$ (Definition~\ref{def:admissible}) with $\eta=1$.
    For a prime power $q$ tending to infinity, let $m=q-1$, $n=q$, and $k=\lfloor\rho m\rfloor$ for a fixed rational rate parameter $\rho\in(0,1)$.
    Fix an explicitly represented field $\F_q$ and an explicitly specified generator $\gamma$ of $\F_q^*$.
    Define
    \[
        C_q=\RS_{\F_q,\gamma,k}\subseteq\F_q^m.
    \]
    Partition $[n]$ into sets $S_1,\ldots,S_R$ with $|S_r|=\alpha_rn+O(1)$, and let $\calT\subseteq[n]^m$ be the resulting special-form load set from Definition~\ref{def:specialloadstates}.
    The field and code parameters, the complete partition, and the truncated-load membership circuits from Definition~\ref{def:loadstates} are public parts of the instance.

    Sample independent uniformly random balanced functions $\pi_1,\ldots,\pi_m:\F_q\to[n]$, and provide their complete truth tables as the input.
    Since both the domain and codomain have size $q$, these functions are precisely uniformly random bijections.
    The search problem is to output an element of $C_q\cap A_{\pi_1,\ldots,\pi_m}$, where Definition~\ref{def:foldedsets} is used with $s=1$ and $H_i=\pi_i$.
\end{definition}

The truth tables in Definition~\ref{def:plainrs} have $O(q^2\log q)$ bits in total, and the remaining public data have size $\poly(q)$.
Thus the complete instance has polynomial length, the functions are ordinary public inputs, and there is no random oracle.

\begin{theorem}[Plain-model ordinary-RS algorithms]\label{thm:plainrs}
    Use the random variables from Theorem~\ref{thm:maininfinite}, and define
    \[
        p_{\mathsf{stay}}
        =\Pr[K_\src-1\in B_{I_\src}]
         \Pr[K_\tar+1\in B_{I_\tar}].
    \]
    For the instances in Definition~\ref{def:plainrs}, with rate parameter $\rho$, the following hold for all sufficiently large $q$:
    \begin{enumerate}[(i)]
        \item If $q$ is prime and $\rho>1-p_{\mathsf{stay}}$, a uniform quantum algorithm runs in time $\poly(q)$ and outputs an element of $C_q\cap A_{\pi_1,\ldots,\pi_m}$ with probability at least $q^{-O(1)}$.
        \item For every prime power $q$, if $\rho>1-p_{\mathsf{stay}}^2$, a uniform quantum algorithm runs in time $\poly(q)$ and outputs such an element with probability $\Omega(1)$.
    \end{enumerate}
    Both probabilities are over the random public truth tables and the algorithm's internal randomness.
    The polynomial exponents and success bounds may depend on the fixed rate and load data.
    A purported output can be verified deterministically in time $\poly(q)$.
\end{theorem}
\begin{proof}
    By Theorem~\ref{thm:maininfinite}, $\calT$ is nonempty for all sufficiently large $q$.

    Let $\calD_q$ be the distribution over the represented instances of Definition~\ref{def:plainrs}, including the complete public truth tables.
    Every target set in its support is nonempty because $\calT$ is nonempty and all the $\pi_i$'s are bijections.
    This distribution of representations is translation-invariant.
    Indeed, for $\vv\in\F_q^m$, define $\pi_i^{(\vv)}(x)=\pi_i(x-v_i)$.
    Precomposition by a translation preserves the uniform distribution on bijections, and
    \[
        A_{\pi_1^{(\vv)},\ldots,\pi_m^{(\vv)}}
        =
        A_{\pi_1,\ldots,\pi_m}+\vv.
    \]
    Given the original tables and $\vv$, a reversible lookup for $\pi_i^{(\vv)}$ is implemented by subtracting $v_i$ from its input before looking up $\pi_i$.
    Thus the translated-set representations used in Corollary~\ref{cor:ctlistdecode} are also available in time $\poly(q)$.

    Next, for every tuple in the support of $\calD_q$, the corresponding set state is efficiently preparable from the public tables.
    Prepare $\ket{\calT}$ using Theorem~\ref{thm:Aprep}.
    From the truth table of a bijection $\pi_i$, its inverse table can be computed in time $\poly(q)$.
    Encode $[n]$ and $\F_q$ by the same $q$ basis states of $\lceil\log q\rceil$ qubits, so that $\pi_i^{-1}$ is a permutation of these states, and extend it by fixing the unused basis states.
    Decomposing it into at most $q-1$ transpositions gives a reversible circuit of size $\poly(q)$.
    Applying the inverse permutation in every coordinate maps
    \[
        \frac{1}{\sqrt{|\calT|}}\sum_{\yv\in \calT}\ket{\yv}
        \quad\longmapsto\quad
        \frac{1}{\sqrt{|\calT|}}
        \sum_{\yv\in \calT}
        \ket{\pi_1^{-1}(y_1),\ldots,\pi_m^{-1}(y_m)}
        =
        \ket{A_{\pi_1,\ldots,\pi_m}}.
    \]
    The total preprocessing and circuit size are polynomial in $q$, and the preparation error can be made exponentially small in polynomial time.

    The remaining requirements are Fourier sparsity and dual decoding.
    Let $W$ be the normalized Hamming weight obtained by measuring
    $\QFT^\dag\ket{A_{\pi_1,\ldots,\pi_m}}$.
    Corollary~\ref{cor:efwfoldedstate} with $s=1$, together with Remark~\ref{rem:fourier-convention}, gives uniformly over the public bijections
    \[
        \EE[W]=1-p_{\mathsf{stay}}+O(q^{-1/2}).
    \]
    By Lemma~\ref{lem:grsduality}, the dual is a generalized Reed--Solomon code, and its rate is
    \[
        R_q^\perp=\frac{m-k}{m}=1-\rho+O(q^{-1}).
    \]
    We now choose a decoding radius separately for the two claims.

    For (i), since $\rho>1-p_{\mathsf{stay}}$, fix a constant $\delta>0$ such that
    \[
        1-p_{\mathsf{stay}}<\zeta:=\rho-\delta.
    \]
    Markov's inequality gives
    \[
        \Pr[W\leq\zeta]
        \geq1-\frac{\EE[W]}{\zeta}=\Omega(1).
    \]
    Thus a fixed constant $P_\sparse>0$ lower-bounds this probability uniformly over the public bijections.
    For all sufficiently large $q$, we have $R_q^\perp\leq1-\rho+\delta/2$, so
    $\zeta\leq1-R_q^\perp-\delta/2$.
    Since $q$ is prime and $m=q-1$, Theorem~\ref{thm:rscapacity} with slack $\delta/2$ applies to $C_q^\perp$; its dimension condition $m-k\leq(1-\delta/2)m$ holds for large $q$ because $\rho-\delta=\zeta>0$.
    It lists every codeword in that radius in time $q^{O_\delta(1)}$, with list size $L\leq m^{O_\delta(1)}$.
    The bounded-time version of Corollary~\ref{cor:ctlistdecode} therefore gives success at least
    \[
        \frac{P_\sparse}{L}-2^{-\Omega(q+m)}\geq q^{-O(1)}.
    \]
    The running time, including the factor $L/P_\sparse$ and the restart cap, is polynomial.

    For (ii), since $\rho>1-p_{\mathsf{stay}}^2$ gives $1-\sqrt{1-\rho}>1-p_{\mathsf{stay}}$, fix a constant $\delta>0$ such that
    \[
        1-p_{\mathsf{stay}}<\zeta:=1-\sqrt{1-\rho}-\delta.
    \]
    Again Markov's inequality gives a fixed lower bound $P_\sparse>0$ on $\Pr[W\leq\zeta]$, uniformly over the public bijections.
    For sufficiently large $q$,
    $\sqrt{R_q^\perp}\leq\sqrt{1-\rho}+\delta/2$, whence
    $\zeta\leq1-\sqrt{R_q^\perp}-\delta/2$.
    Theorem~\ref{thm:guruswamisudan} (using the decoder of~\cite{guruswamiSudan1999}) gives polynomial-time list decoding at this radius with $L=O(1/\delta)$, since $R_q^\perp$ is bounded away from zero.
    The bounded-time version of Corollary~\ref{cor:ctlistdecode} consequently gives success $\Omega(1)$ in polynomial time.

    In both cases every use of a $\pi_i$ is a reversible lookup in its public table, so the algorithm is in the plain model.
    Finally, after producing a candidate, the algorithm checks membership in $C_q$, evaluates the public tables in all $m$ coordinates, and checks the resulting histogram.
    It outputs the candidate only if every check passes.
    These checks take time $\poly(q)$ and accept exactly the elements of $C_q\cap A_{\pi_1,\ldots,\pi_m}$.
\end{proof}

The first part of Theorem~\ref{thm:plainrs} uses ordinary prime-field capacity decoding~\cite{brakensiek2026capacity,jeronimo2026capacity} and requires no folding.
Its polynomial list bound yields inverse-polynomial average success; it does not give the constant success of the second part.
Repetition on a fixed input alone does not upgrade this average guarantee, since the successful inputs could occupy only an inverse-polynomial fraction of the instance distribution.

We also emphasize that Theorem~\ref{thm:plainrs} makes no classical hardness claim; the online soundness argument from Section~\ref{sec:yzoracle} relies on leaving many values of each exponentially large random function unqueried.
Here each entire truth table has polynomial size and can be read by a polynomial-time classical algorithm, so that argument does not apply.

\section{A Heterogeneous Code-Intersection Example}
\label{sec:heterogeneous-example}

We specialize the oracle separation of Section~\ref{sec:yzoracle} and the plain-model algorithms of
Section~\ref{sec:plainrs} to the three-class example highlighted in the introduction.
For every alphabet size $n$, set $m=n-1$, and let $S_1$
consist of the first $\lfloor n/2\rfloor$ symbols, let $S_2$ consist of the next
$\lfloor n/3\rfloor$ symbols, and let $S_3$ contain the remaining symbols.  Thus
\[
    |S_1|=\left\lfloor\frac n2\right\rfloor,\qquad
    |S_2|=\left\lfloor\frac n3\right\rfloor,\qquad
    |S_3|=n-|S_1|-|S_2|.
\]
Let
\[
    B_1=\{0\},\qquad B_2=\{1,2\},\qquad B_3=\{3,4\},
\]
and let $c_a(\yv)$ be the histogram of Definition~\ref{def:histogram}.  Define
\begin{equation}\label{eq:three-class-set}
    \calT_{\mathsf{het}}
    =
    \{\yv\in[n]^m:c_a(\yv)\in B_r
      \text{ for every }r\in[3]\text{ and }a\in S_r\}.
\end{equation}
For $r\in[3]$, let $\calT^{(r)}$ be the relaxation retaining only the class-$r$ constraints:
\begin{equation}\label{eq:three-class-relaxations}
    \calT^{(r)}
    =
    \{\yv\in[n]^m:c_a(\yv)\in B_r\text{ for every }a\in S_r\}.
\end{equation}

\subsection{Stability Calculation}
Here we compute the $p_\mathsf{stay}$ associated to $\calT_\mathsf{het}$ using Theorem~\ref{thm:maininfinite}.
With
\[
    (\alpha_1,\alpha_2,\alpha_3)
    =
    \left(\frac12,\frac13,\frac16\right),
\]
the first class has mean $0$ since $B_1=\{0\}$, and the tilted means in the second and third classes are
\[
    \mu_2(\nu)=\EE_{X \gets \tpoiss_{\nu, B_2}} [X] = \frac{2(1+\nu)}{2+\nu},
    \qquad
    \mu_3(\nu)=\EE_{X \gets \tpoiss_{\nu, B_3}} [X] =\frac{4(3+\nu)}{4+\nu}.
\]
For $\tpoiss_{\nu, (B_1, B_2, B_3), (\alpha_1, \alpha_2, \alpha_3)}$ to have mean 1, $\nu$ must satisfy:
\[
    \frac13\mu_2(\nu)+\frac16\mu_3(\nu)=1
    \quad\Longleftrightarrow\quad
    \frac1{\nu+2}+\frac1{\nu+4}=\frac12.
\]
Solving, and using $\nu>0$, gives $\nu=\sqrt5-1$.  The load family is admissible: its fixed finite allowed-load
sets have polynomial-time membership tests,
this $\nu$ witnesses the required mean, and $1\in B_2-B_2$, so the generated additive subgroup is
$\ZZ$.
Hence we may apply Theorem~\ref{thm:maininfinite}.

Under the target-load distribution from Theorem~\ref{thm:maininfinite}, an increment is legal only
from $1$ to $2$ or from $3$ to $4$.  Hence, by the mean equation,
\[
    \Pr[K_{\tar}+1\in B_{I_{\tar}}]
    =
    \frac13\frac2{\nu+2}+\frac16\frac4{\nu+4}
    =
    \frac13.
\]
Under the size-biased source-load distribution, a decrement is legal only from $2$ to $1$ or from $4$
to $3$.  Since the mean load is one, and by the mean equation,
\[
    \Pr[K_{\src}-1\in B_{I_{\src}}]
    =
    \frac13\frac{2\nu}{\nu+2}+\frac16\frac{4\nu}{\nu+4}
    =
    \frac{\nu}{3}.
\]
Therefore
\[
    p_{\mathsf{stay}}
    =
    \frac{\nu}{9}
    =
    \frac{\sqrt5-1}{9},
    \qquad
    1-p_{\mathsf{stay}}
    =
    \frac{10-\sqrt5}{9}.
\]

\subsection{Oracle Setting}

Use the random-oracle Reed--Solomon parameters $q,M,s,n,m,C_\lambda$ and the independent functions
$H_i:\F_q^s\to[n]$ from Section~\ref{sec:yzoracle}, with rate
\[
    \rho=\frac{7}{8}.
\]
For any $\calT'\subseteq[n]^m$, let its associated verifier accept
$\xv=(\xv_1,\ldots,\xv_m)\in(\F_q^s)^m$ precisely when
\begin{equation}\label{eq:relaxed-verifier}
    \xv\in C_\lambda
    \quad\text{and}\quad
    (H_1(\xv_1),\ldots,H_m(\xv_m))\in \calT'.
\end{equation}

\begin{corollary}[Three-way heterogeneous separation]
\label{cor:heterogeneous-separation}
    In the oracle setting above, with rate parameter $\rho=7/8$, there is a uniform quantum algorithm running in time $\poly(\lambda)$ and making
    $\poly(\lambda)$ quantum queries which, with probability $\Omega(1)$ over the random functions and its
    internal randomness, outputs a word accepted by the verifier~\eqref{eq:relaxed-verifier} for $\calT_{\mathsf{het}}$~\eqref{eq:three-class-set}.  The same output
    is therefore accepted by all three relaxed verifiers.

    On the other hand, for every fixed $r\in[3]$, every computationally unbounded classical algorithm
    making $\poly(\lambda)$ classical queries outputs a word accepted by the verifier for the relaxed set $\calT^{(r)}$~\eqref{eq:three-class-relaxations} with
    probability at most $2^{-\Omega(\lambda)}$.
\end{corollary}
\begin{proof}
    The chosen rate satisfies
    \[
        \frac{7}{8}
        -
        \frac{10-\sqrt5}{9}
        =
        \frac{8\sqrt5-17}{72}>0.
    \]
    Theorem~\ref{thm:yzpositive} consequently gives the claimed quantum algorithm for
    $\calT_{\mathsf{het}}$.  Since $\calT_{\mathsf{het}}\subseteq \calT^{(r)}$ for every $r$, its successful output is
    accepted by each relaxed verifier.

    Fix $r\in[3]$.  View $\calT^{(r)}$ as a three-class load set by replacing $B_j$ with
    $\ZZ_{\geq0}$ for $j\neq r$.  The retained $B_r$ is finite, so its tail threshold is $\infty$; the two
    unrestricted classes have tail threshold $0$.  Thus the weighted tail threshold of this relaxed
    family is $\tau=\infty$.  For example, the constants
    \[
        \epsilon=\frac1{32},
        \qquad
        \zeta=\frac1{16}
    \]
    satisfy $0<\epsilon<\zeta<1-\rho$ and $\tau>1-\epsilon$.
    Lemma~\ref{lem:yzclassicalsoundness}, applied to these relaxed load data, gives
    success probability $2^{-\Omega(\lambda)}$ for every polynomial-query classical algorithm.
\end{proof}

\subsection{Plain Model}

Instantiate Definition~\ref{def:plainrs} with $n=q$, $m=q-1$, the partition above, and a fixed rational
rate $\rho\in(0,1)$.  The code is the ordinary Reed--Solomon code
\[
    C_q=\RS_{\F_q,\gamma,\lfloor\rho(q-1)\rfloor}\subseteq\F_q^{q-1}.
\]
The functions $\pi_i:\F_q\to[q]$ are independent uniformly random bijections, given as complete public
truth tables.  The search problem is to find
\[
    \xv\in C_q
    \quad\text{such that}\quad
    (\pi_1(x_1),\ldots,\pi_m(x_m))\in \calT_{\mathsf{het}}.
\]
Thus the histogram constraints are imposed after a separate public relabeling in each coordinate.

\begin{corollary}[Plain-model heterogeneous instances]
\label{cor:heterogeneous-plain}
    For the instances above (Definition~\ref{def:plainrs} with the three-class load data and rate parameter $\rho$), the following hold for all sufficiently large $q$:
    \begin{enumerate}[(i)]
        \item If $q$ is prime and $\rho>(10-\sqrt5)/9$, a uniform quantum algorithm runs in time
        $\poly(q)$ and finds a solution with probability at least $q^{-O(1)}$.
        In particular, this holds at rate $\rho=7/8$.
        \item For every prime power $q$, if $\rho>(75+2\sqrt5)/81$, a uniform quantum algorithm runs in
        time $\poly(q)$ and finds a solution with probability $\Omega(1)$.
        In particular, this holds at rate $\rho=63/64$.
    \end{enumerate}
    Both probabilities are over the random public truth tables and the algorithm's internal randomness.
\end{corollary}
\begin{proof}
    The load family is admissible and has $p_{\mathsf{stay}}=(\sqrt5-1)/9$, as calculated above.
    Part (i) follows from Theorem~\ref{thm:plainrs}(i), using
    $1-p_{\mathsf{stay}}=(10-\sqrt5)/9<7/8$.
    For part (ii),
    \[
        1-p_{\mathsf{stay}}^2
        =
        1-\left(\frac{\sqrt5-1}{9}\right)^2
        =
        \frac{75+2\sqrt5}{81},
        \qquad
        \frac{63}{64}-\frac{75+2\sqrt5}{81}
        =
        \frac{303-128\sqrt5}{5184}>0.
    \]
    Theorem~\ref{thm:plainrs}(ii) therefore gives the claimed algorithm with constant success probability.
\end{proof}

As emphasized in Section~\ref{sec:plainrs}, these plain-model guarantees make no classical hardness claim;
the lower bound in Corollary~\ref{cor:heterogeneous-separation} relies on an oracle model with exponentially large input domains.

\section{Statement on AI Use}

\begin{flushleft}
\renewcommand{\arraystretch}{1.3}
\begin{tabular}{@{}l@{\qquad}c@{}}
\textbf{Activity} & \textbf{AI used} \\
Asking the research question & No \\
Coming up with the approach & Yes \\
Proof development & Yes \\
Writing and exposition & Yes \\
Checking for bugs & Yes \\
Other supporting tasks & Yes \\
\end{tabular}
\end{flushleft}

\begin{sloppypar}
\paragraph{Details (elaborating on the points above).} 
This project was built around two central goals. Our first goal was to extend existing code-intersection problems for Reed--Solomon codes beyond local nonlinear constraints, in the plain model. 
To arrive at this goal, we interacted heavily with GPT-5.5 Pro to ascertain intermediate technical lemmas, arrive at the ``stability'' formalism for reasoning about expected Fourier weight, and obtain a formal version of Theorem~\ref{thm:maininfinite} with proof. The authors completely rewrote the proof and then noticed a way to adapt the Yamakawa--Zhandry oracle separation to this setting. Details were filled in with the help of GPT-5.5 Pro, GPT-5.6 Sol Pro, and GPT-5.6 Sol Ultra.

Our second goal was motivated by the repeated dequantizations~\cite{ivanyos2012nil2,ivanyosSantha2015diagonal,imran2024zerosum,kothari2025sis,shutty2026optimizationusinglocallyquantumdecoders} of algorithms~\cite{clz21,shutty2026optimizationusinglocallyquantumdecoders} using Regev's reduction with local quantum decoding followed by a linear solver. This includes earlier unpublished work by the authors which we had dequantized ourselves. 
We wanted to show a generic dequantization theorem for any algorithm following this template, and gave this task to GPT-5.6 Sol Ultra with adversarial sub-agents, each tasked with either proving or disproving this conjecture (but not both). 
To our surprise, after some iterations the sub-agents landed on quantum algorithms leveraging a two-fold multiplication property as a candidate that resisted their initial dequantization efforts, which we have presented in Section~\ref{sec:q2}.

We then made an extended human-in-the-loop effort over several weeks---using GPT-5.6 Sol Pro/Ultra and GPT-6 Astra Pro/Ultra---to find a classical algorithm that relied only on a two-fold multiplication property. Ultimately, we did not achieve this but instead found an algorithm relying on a stronger three-fold multiplication property, which we present in Section~\ref{sec:classical}. GPT-6 Astra did find other slightly more performant classical algorithms specializing to Reed--Muller codes that were significantly more complicated but nevertheless failed to fully cover our quantum algorithm; we have not digested or verified those algorithms so we do not present them here. We are happy to privately share LLM-generated details of these other classical algorithms on request.

The aforementioned GPT models were used to assist with writing throughout. These models and Claude Fable 5 were also used for proofreading, and Claude Opus 5.5 was used for figure generation. The authors are fully accountable for the contents and correctness of this paper.

\end{sloppypar}

\section{Acknowledgements}

We thank Ryan Babbush, Aparna Gupte, Stephen P. Jordan, Sascha Schmidhuber, Kimi Sun, Vinod Vaikuntanathan, and Mary Wootters for useful discussions. We also thank the members of the Google Quantum AI Algorithms team who attended internal talks by the authors on this work and provided stimulating questions and feedback. 
This work was done primarily while SR was a student researcher at Google Quantum AI. SR was also supported in part by NSF CNS-2534400 and by Jane Street. SR's access to ChatGPT Pro was supported by the UK AISI alignment project and the Simons Institute for the Theory of Computing, and his access to Claude Max was supported by Jane Street.

\clearpage

\bibliographystyle{alpha}
\bibliography{refs}

\appendix

\section{Proof of the One-Dimensional Local Central Limit Theorem}\label{app:1dlocalclt}

\begin{proof}[Proof of Theorem~\ref{thm:1dlocalclt}]
    For $r\in[R]$, let
    \[
        \phi_r(t)=\EE_{X\sim\calD_r}[e^{itX}]
        \qquad\text{and}\qquad
        \psi_r(t)=e^{-it\mu_r}\phi_r(t).
    \]
    Also write
    \[
        \mu_n=\sum_{r\in[R]}n_r\mu_r,
        \qquad
        V_n=\sum_{r\in[R]}n_r\sigma_r^2,
        \qquad
        \psi_n(t)=\prod_{r\in[R]}\psi_r(t)^{n_r}.
    \]
    By independence, $\psi_n(t)=\EE[e^{it(Z-\mu_n)}]$.
    Because some $\sigma_r^2$ is positive, every $\alpha_r$ is positive, and $n_r=\alpha_rn+O(1)$, there is a constant $v>0$ such that $V_n\geq vn$ for all sufficiently large $n$.

    We first control the characteristic function near the origin.
    The finite third absolute moment assumption gives, for each $r$,
    \[
        \psi_r(t)
        =
        1-\frac{\sigma_r^2t^2}{2}+O(|t|^3)
        \qquad (t\to 0).
    \]
    Since there are only finitely many distributions, we may choose constants $\delta_0,C>0$ such that every $\psi_r(t)$ is nonzero for $|t|\leq\delta_0$ and, using the branch of the logarithm that vanishes at $1$,
    \[
        \log\psi_r(t)
        =
        -\frac{\sigma_r^2t^2}{2}+R_r(t),
        \qquad
        |R_r(t)|\leq C|t|^3
    \]
    for every $r$ and $|t|\leq\delta_0$.
    Hence, with $R_n(t)=\sum_r n_rR_r(t)$,
    \[
        \psi_n(t)
        =
        \exp\left(-\frac{V_nt^2}{2}+R_n(t)\right),
        \qquad
        |R_n(t)|\leq Cn|t|^3.
    \]
    Shrinking $\delta_0$ if necessary, we may assume $C\delta_0\leq v/4$.
    Then $|R_n(t)|\leq Cn|t|^3\leq vnt^2/4$, so the inequality $|e^w-1|\leq |w|e^{|w|}$ and $V_n\geq vn$ give, with $C_1=C$ and $c_1=v/4$,
    \begin{align}
        \left|\psi_n(t)-e^{-V_nt^2/2}\right|
        &=
        e^{-V_nt^2/2}|e^{R_n(t)}-1| \nonumber\\
        &\leq
        C_1n|t|^3e^{-c_1nt^2}
        \label{eq:lclt-small-frequency}
    \end{align}
    whenever $|t|\leq\delta_0$.

    We next control frequencies bounded away from the origin.
    If $|\phi_r(t)|=1$, then the unit-modulus random variable $e^{itX}$ is constant almost surely for $X\sim\calD_r$.
    Consequently, $e^{it(x-y)}=1$ for all $x,y\in\operatorname{supp}(\calD_r)$.
    If $|\phi_r(t)|=1$ for every $r$, then the integers $d$ with $e^{itd}=1$ form a subgroup of $\ZZ$ containing every difference of support points, so by the lattice assumption this subgroup is $\ZZ$; in particular $e^{it}=1$.
    Thus, for every $t\in[-\pi,\pi]\setminus\{0\}$, there is an $r$ such that $|\phi_r(t)|<1$.
    For $t$ in the compact set
    \[
        K=[-\pi,-\delta_0]\cup[\delta_0,\pi],
    \]
    define $g(t)=\prod_r|\phi_r(t)|^{\alpha_r/2}$.
    The preceding observation shows that $g(t)<1$ throughout $K$.
    Continuity and compactness give a constant $\theta<1$ such that $g(t)\leq\theta$ on $K$.
    Since $n_r\geq\alpha_rn/2$ for all sufficiently large $n$ and $|\phi_r(t)|\leq1$,
    \begin{equation}
        |\psi_n(t)|
        =
        \prod_r|\phi_r(t)|^{n_r}
        \leq
        g(t)^n
        \leq
        \theta^n
        \label{eq:lclt-large-frequency}
    \end{equation}
    for every $t\in K$.

    We now compare the two Fourier inversion formulas.
    Since $Z$ is integer-valued,
    \[
        \Pr[Z=z]
        =
        \frac{1}{2\pi}
        \int_{-\pi}^{\pi}
        e^{-it(z-\mu_n)}\psi_n(t)\,dt.
    \]
    The Gaussian Fourier transform gives
    \[
        \frac{1}{\sqrt{2\pi V_n}}
        \exp\left(-\frac{(z-\mu_n)^2}{2V_n}\right)
        =
        \frac{1}{2\pi}
        \int_{\RR}
        e^{-it(z-\mu_n)}e^{-V_nt^2/2}\,dt.
    \]
    Taking the difference, splitting the integrals at $\delta_0$, and using that the oscillatory factors have modulus $1$, we obtain
    \begin{align*}
        2\pi
        \left|
        \Pr[Z=z]
        -
        \frac{e^{-(z-\mu_n)^2/(2V_n)}}{\sqrt{2\pi V_n}}
        \right|
        \leq{}&
        \int_{|t|\leq\delta_0}
        \left|\psi_n(t)-e^{-V_nt^2/2}\right|\,dt\\
        &+
        \int_{\delta_0\leq|t|\leq\pi}|\psi_n(t)|\,dt
        +
        \int_{|t|>\delta_0}e^{-V_nt^2/2}\,dt.
    \end{align*}
    By~\eqref{eq:lclt-small-frequency}, the first integral is at most
    \[
        C_1n\int_{\RR}|t|^3e^{-c_1nt^2}\,dt
        =
        O(1/n),
    \]
    where the equality follows from the substitution $u=\sqrt n\,t$.
    By~\eqref{eq:lclt-large-frequency}, the second integral is at most $2\pi\theta^n$.
    Finally, $V_n\geq vn$ implies
    \[
        \int_{|t|>\delta_0}e^{-V_nt^2/2}\,dt
        \leq
        e^{-vn\delta_0^2/4}
        \int_{\RR}e^{-vnt^2/4}\,dt
        =
        O(e^{-c_2n})
    \]
    for some $c_2>0$.
    The three bounds are independent of $z$, so the claimed $O(1/n)$ estimate holds uniformly over $z\in\ZZ$.
\end{proof}

\section{Proof of Theorem~\ref{thm:maininfinite}}\label{app:maininfinite}

The proof has three steps. First, the histogram of a uniform word of $\calT$ has the law of independent tilted-Poisson loads $K_a$ conditioned on $\sum_aK_a=m$ (Lemma~\ref{lem:hard-wall-histogram-law}). Second, the corollaries of the local central limit theorem in Section~\ref{app:lclt-cor} show that, for distinct symbols $a$ and $b$, this conditioning changes the expectation of a product $g(K_a)h(K_b)$ by only $O(n^{-1/2})$ (Lemma~\ref{lem:infinitedecoupling}). Third, replacing an occurrence of $a$ by $b\neq a$ keeps the word in $\calT$ exactly when the load of $a$ can decrease by one and the load of $b$ can increase by one; applying the second step to these two events gives the product formula of the theorem.

Assume the hypotheses of the theorem; in particular, $\nu$ is the tilt fixed with the admissible load data.
Recall that:
\[
    \calT = \{\xv \in [n]^m: c_a(\xv) \in B_r\ \forall r \in [R], a \in S_r\}.
\]
Since every count is at most $m$, this agrees with Definition~\ref{def:specialloadstates}, which uses $B_r\cap\{0,\ldots,m\}$.
For a value $a \in [n]$, we also let $\class(a)$ be the unique $r \in [R]$ such that $a \in S_r$.

\subsection{Corollaries of the Local Central Limit Theorem}\label{app:lclt-cor}

\begin{definition}
    Let $n_1, \ldots, n_R \in \ZZ_{> 0}$ and $z \in \ZZ$ be integers.
    We define the quantity $b_{n_1, \ldots, n_R}(z)$ as follows:
    let $\{Z_{r, i}: r \in [R], i \in [n_r]\}$ be independent random variables with $Z_{r, i} \sim \tpoiss_{\nu, B_r}$.
    Then:
    $$b_{n_1, \ldots, n_R}(z) = \Pr\left[\sum_{r \in [R]} \sum_{i \in [n_r]} Z_{r, i} = z\right].$$
\end{definition}

For use throughout this subsection, let $\mu_r$ and $\sigma_r^2$ be the mean and variance of $\tpoiss_{\nu,B_r}$.
Lemma~\ref{lem:tpoisstailbound} gives finite third moments.
Moreover, admissibility gives the lattice condition in Theorem~\ref{thm:1dlocalclt}; it also implies that some $\sigma_r^2$ is positive, since otherwise every $B_r$ would be a singleton and the union of the difference sets could not generate $\ZZ$.
Thus the local central limit theorem applies whenever every class has $\alpha_rn+O(1)$ summands, including after deleting $O(1)$ summands.

\begin{lemma}[Approximate independence after conditioning]\label{lem:approxindep}
    Let $K_1,\ldots,K_n$ be independent random variables with $K_a\sim\tpoiss_{\nu,B_{\class(a)}}$.
    Let $J\subseteq[n]$ have size bounded by a constant, and let $\{k_a:a\in J\}$ satisfy $0\leq k_a\leq O(n^{1/4})$ and $\Pr[K_a=k_a]>0$ for all $a\in J$.
    Then
    \[
        \frac{\Pr[K_a=k_a\ \forall a\in J\mid K_1+\cdots+K_n=m]}{\Pr[K_a=k_a\ \forall a\in J]}=1+O(1/\sqrt{n}).
    \]
    The implicit constant depends only on the fixed load data and on the constants bounding $|J|$ and $k_an^{-1/4}$; in particular, it is uniform over all such $J$ and $k_a$.
\end{lemma}
\begin{proof}
    For each $r \in [R]$, let $n_r = |S_r|$ and $n'_r = |S_r \backslash (S_r \cap J)|$.
    Note that we have $n_r, n'_r = \alpha_r n + O(1)$.
    By independence of the $K_a$, the conditional probability of interest equals:
    \begin{align*}
        \left(\prod_{a \in J} \Pr[K_a = k_a]\right) \cdot \frac{b_{n'_1, \ldots, n'_R}(m-\sum_{a \in J} k_a)}{b_{n_1, \ldots, n_R}(m)},
    \end{align*}
    so it remains to show that:
    \[
        \frac{b_{n'_1, \ldots, n'_R}(m-\sum_{a \in J} k_a)}{b_{n_1, \ldots, n_R}(m)} = 1 + O(1/\sqrt{n}).
    \]
    Write $V=\sum_r n_r\sigma_r^2$ and $V'=\sum_r n'_r\sigma_r^2$.
    Since $\sum_r\alpha_r\mu_r=\eta$ and $m=\eta n+O(1)$,
    \[
        \left|m-\sum_rn_r\mu_r\right|=O(1).
    \]
    Moreover,
    \[
        \left|
        m-\sum_{a\in J}k_a-\sum_rn'_r\mu_r
        \right|
        =
        O(n^{1/4}).
    \]
    Theorem~\ref{thm:1dlocalclt}, as in the proof of Corollary~\ref{cor:lclt-central-ratio}, therefore gives
    \[
        b_{n_1,\ldots,n_R}(m)
        =\frac{1}{\sqrt{2\pi V}}\left(1+O(n^{-1/2})\right)
    \]
    and
    \[
        b_{n'_1,\ldots,n'_R}\left(m-\sum_{a\in J}k_a\right)
        =\frac{1}{\sqrt{2\pi V'}}\left(1+O(n^{-1/2})\right).
    \]
    Finally, $V-V'=O(1)$ and $V,V'=\Theta(n)$, so $\sqrt{V/V'}=1+O(1/n)$.
    Taking the ratio proves the claim.
\end{proof}

\begin{lemma}[Tail bound after conditioning]\label{lem:tailboundafterconditioning}
    Let $K_1,\ldots,K_n$ be independent random variables with $K_a\sim\tpoiss_{\nu,B_{\class(a)}}$.
    Let $J\subseteq[n]$ have size bounded by a constant, and let $\{k_a:a\in J\}$ satisfy $\Pr[K_a=k_a]>0$ for all $a\in J$.
    Then
    \[
        \frac{\Pr[K_a=k_a\ \forall a\in J\mid K_1+\cdots+K_n=m]}{\Pr[K_a=k_a\ \forall a\in J]}\leq O(\sqrt{n}).
    \]
\end{lemma}
\begin{proof}
    By following the proof of Lemma~\ref{lem:approxindep}, it suffices to prove that:
    \[
        \frac{b_{n'_1, \ldots, n'_R}(m-\sum_{a \in J} k_a)}{b_{n_1, \ldots, n_R}(m)} \leq O(\sqrt{n}).
    \]
    This follows by crudely upper bounding the numerator by 1 (since it is a probability) and lower bounding the denominator by $\Omega(1/\sqrt{n})$ (by Corollary~\ref{cor:lclt-central-ratio}).
\end{proof}

\subsection{Main Step: Source-Target Decoupling}

\begin{lemma}[Product-law representation of weighted histograms]\label{lem:product-law-histograms}
    Consider a disjoint partition $S_1 \cup \ldots \cup S_R$ of $[n]$.
    Let $\{K_a: a \in [n]\}$ be independent random variables such that $K_a \sim \tpoiss_{\nu, B_r}$ for all $r \in [R], a \in S_r$.
    Assume $\Pr[K_1 + \ldots + K_n = m] > 0$.
    Then for any $c_1, \ldots, c_{n}$ summing to $m$, we have:
    \begin{equation}\label{eq:productlaw}
        \Pr[
            K_1=c_1,\dots,K_{n}=c_{n}
            \mid
            K_1+\ldots+K_n=m
        ]
        =
        \frac{
            \prod_{r \in [R]} \prod_{a \in S_r}\frac{\one[c_a \in B_r]}{c_a!}
        }{
            \sum_{\substack{d_1,\dots,d_{n}\ge 0\\ d_1+\cdots+d_{n}=m}}
            \prod_{r \in [R]} \prod_{a \in S_r}\frac{\one[d_a \in B_r]}{d_a!}
        }.
    \end{equation}
\end{lemma}

\begin{proof}
    Ignoring the conditioning for a moment, we have:
    \begin{align*}
        \Pr[K_a = c_a\ \forall a] &= \frac{1}{\prod_{r \in [R]} (Z(\nu, B_r))^{|S_r|}} \prod_{r \in [R]} \prod_{a \in S_r} \frac{\one[c_a \in B_r] \nu^{c_a}}{c_a!}\\
        &= \frac{\nu^m}{\prod_{r \in [R]} (Z(\nu, B_r))^{|S_r|}} \prod_{r \in [R]} \prod_{a \in S_r} \frac{\one[c_a \in B_r]}{c_a!}.
    \end{align*}
    Summing over all $(c_1,\ldots,c_n)$ with $\sum_ac_a=m$ gives:
    \begin{align*}
        \Pr[K_1+\ldots+K_{n} = m] &= 
        \frac{\nu^m}{\prod_{r \in [R]} (Z(\nu, B_r))^{|S_r|}} \sum_{\substack{d_1,\dots,d_{n}\ge 0\\ d_1+\cdots+d_{n}=m}} \prod_{r \in [R]} \prod_{a \in S_r} \frac{\one[d_a \in B_r]}{d_a!}
    \end{align*}
    The conclusion follows by dividing these quantities.
\end{proof}

\begin{lemma}[Hard-wall histogram law]
\label{lem:hard-wall-histogram-law}
The following two distributions are identical:
\begin{itemize}
    \item Sample $\xv \gets \calT$ uniformly at random.
    Output the histogram $\cv(\xv)$.
    \item Sample $K_1, \ldots, K_n$ independently such that $K_a \sim \tpoiss_{\nu, B_r}$ for all $r \in [R], a \in S_r$.
    Output $(K_1, \ldots, K_n)$ conditioned on $\sum_a K_a = m$.
\end{itemize}
\end{lemma}
\begin{proof}
    We apply Lemma~\ref{lem:product-law-histograms}.
    Its hypothesis $\Pr[\sum_aK_a=m]>0$ holds because $\calT$ is nonempty: the histogram of any element of $\calT$ is a vector of allowed loads summing to $m$, and each $\tpoiss_{\nu,B_r}$ gives positive probability to every element of $B_r$.
    The lemma tells us that for every $(c_1,\ldots,c_n)$ with $c_a \in B_r$ for all $r \in [R]$ and $a \in S_r$, and with $\sum_ac_a=m$, we have:
    \begin{align*}
        \Pr[K_1=c_1,\dots,K_n=c_{n}
            \mid
            K_1+\cdots+K_{n}=m]
        &=
        \frac{
            \prod_{a\in [n]}\frac{1}{c_a!}
        }{
            \sum_{\substack{d_1,\dots,d_n\\ d_a \in B_r\ \forall r \in [R], a \in S_r \\ d_1 + \ldots + d_{n} = m}}
            \prod_{a\in [n]}\frac{1}{d_a!}
        } \\
        &= \frac{
            \frac{m!}{\prod_{a\in[n]} c_a!}
        }{
            \sum_{\substack{d_1,\dots,d_{n}\\ d_a \in B_r\ \forall r \in [R], a \in S_r \\ d_1 + \ldots + d_{n} = m}}
            \frac{m!}{\prod_{a\in [n]} d_a!}
        } \\
        &= \Pr_{\xv \gets \calT} [\cv(\xv) = (c_1, \ldots, c_n)],
    \end{align*}
    as desired. For other histograms both probabilities are zero. The final step above uses Lemma~\ref{lem:histogram-counting} and Corollary~\ref{cor:loadset-size}.
\end{proof}

The next estimate removes the conditioning for the two loads touched by one replacement.
The source load is weighted by its size, because choosing a uniform coordinate selects a symbol with probability proportional to its load.

\begin{lemma}[Source-target decoupling]\label{lem:infinitedecoupling}
    Let the $K_a$ be independent, with $K_a\sim\tpoiss_{\nu,B_{\class(a)}}$.
    For distinct $a,b\in[n]$ and functions $g,h:\ZZ_{\geq0}\to\RR$ satisfying
    $0\leq g(k)\leq k$ and $0\leq h(k)\leq1$, we have
    \[
        \EE\left[g(K_a)h(K_b)\,\middle|\,\sum_{c\in[n]}K_c=m\right]
        =\EE[g(K_a)]\EE[h(K_b)]+O(n^{-1/2}).
    \]
    The error is uniform in $a,b,g,h$, with constants depending only on the fixed parameters in Theorem~\ref{thm:maininfinite}.
\end{lemma}
\begin{proof}
    Expanding over the two loads gives
    \begin{align*}
        &\EE\left[g(K_a)h(K_b)\,\middle|\,\sum_cK_c=m\right]\\
        &\qquad=\sum_{k,\ell\geq0}g(k)h(\ell)
            \Pr\left[K_a=k,K_b=\ell\,\middle|\,\sum_cK_c=m\right],
    \end{align*}
    whereas independence gives
    \[
        \EE[g(K_a)]\EE[h(K_b)]
        =\sum_{k,\ell\geq0}g(k)h(\ell)\Pr[K_a=k]\Pr[K_b=\ell].
    \]
    We compare these sums first when both loads are at most $\lceil n^{1/4}\rceil$, and then bound the two remaining tails.

    For $0\leq k,\ell\leq\lceil n^{1/4}\rceil$, Lemma~\ref{lem:approxindep}, whose error bound is uniform over such $k$ and $\ell$, gives
    \begin{align*}
        &\left|
            \Pr\left[K_a=k,K_b=\ell\,\middle|\,\sum_cK_c=m\right]
            -\Pr[K_a=k]\Pr[K_b=\ell]
        \right|\\
        &\qquad\leq O(n^{-1/2})\Pr[K_a=k]\Pr[K_b=\ell].
    \end{align*}
    If either marginal probability is zero, both joint probabilities are zero, so the same inequality holds.
    Since the weights $g(k)h(\ell)$ are nonnegative, the triangle inequality bounds the difference of the restricted sums by
    \begin{align*}
        &O(n^{-1/2})
            \sum_{0\leq k,\ell\leq\lceil n^{1/4}\rceil}
            g(k)h(\ell)\Pr[K_a=k]\Pr[K_b=\ell]\\
        &\qquad\leq O(n^{-1/2})
            \sum_{k,\ell\geq0}k\Pr[K_a=k]\Pr[K_b=\ell]\\
        &\qquad=O(n^{-1/2})
            \left(\sum_{k\geq0}k\Pr[K_a=k]\right)
            \left(\sum_{\ell\geq0}\Pr[K_b=\ell]\right)\\
        &\qquad=O(n^{-1/2})\EE[K_a]
            =O(n^{-1/2}).
    \end{align*}
    Here we used $g(k)h(\ell)\leq k$ and the bounded means of the finitely many tilted-Poisson laws.

    For the unconditioned tail, the same pointwise bound gives
    \begin{align*}
        &\sum_{\substack{k,\ell\geq0\\\max\{k,\ell\}>\lceil n^{1/4}\rceil}}
            g(k)h(\ell)\Pr[K_a=k]\Pr[K_b=\ell]\\
        &\qquad\leq
            \sum_{\substack{k>\lceil n^{1/4}\rceil\\\ell\geq0}}
                k\Pr[K_a=k]\Pr[K_b=\ell]
            +\sum_{\substack{k\geq0\\\ell>\lceil n^{1/4}\rceil}}
                k\Pr[K_a=k]\Pr[K_b=\ell]\\
        &\qquad=
            \EE[K_a\one[K_a>\lceil n^{1/4}\rceil]]
            +\EE[K_a]\Pr[K_b>\lceil n^{1/4}\rceil]\\
        &\qquad=\negl(n),
    \end{align*}
    by Lemma~\ref{lem:tpoisstailbound}.
    Applying Lemma~\ref{lem:tailboundafterconditioning} term by term then bounds the conditional tail:
    \begin{align*}
        &\sum_{\substack{k,\ell\geq0\\\max\{k,\ell\}>\lceil n^{1/4}\rceil}}
            g(k)h(\ell)
            \Pr\left[K_a=k,K_b=\ell\,\middle|\,\sum_cK_c=m\right]\\
        &\qquad\leq O(\sqrt n)
            \sum_{\substack{k,\ell\geq0\\\max\{k,\ell\}>\lceil n^{1/4}\rceil}}
                g(k)h(\ell)\Pr[K_a=k]\Pr[K_b=\ell]\\
        &\qquad=O(\sqrt n)\negl(n)=\negl(n).
    \end{align*}
    The difference of the full expectations is therefore bounded by the error on the restricted sums plus the two nonnegative tail sums:
    \begin{align*}
        &\left|
            \EE\left[g(K_a)h(K_b)\,\middle|\,\sum_cK_c=m\right]
            -\EE[g(K_a)]\EE[h(K_b)]
        \right|\\
        &\qquad\leq O(n^{-1/2})+\negl(n)+\negl(n)
            =O(n^{-1/2}).
    \end{align*}
    All bounds are uniform in $a,b,g,h$: the probability estimates are uniform, the weights satisfy the same pointwise bounds, and there are only finitely many load classes.
\end{proof}

\subsection{Completing the Proof}

\begin{proof}[Proof of Theorem~\ref{thm:maininfinite}]
    Take independent loads $K_a\sim\tpoiss_{\nu,B_{\class(a)}}$ for $a\in[n]$.
    First, admissibility and Corollary~\ref{cor:lclt-central-ratio} give
    $\Pr[\sum_aK_a=m]=\Theta(n^{-1/2})>0$ for all sufficiently large $n$.
    Thus an allowed histogram of total load $m$ exists, and $\calT$ is nonempty.

    Conditional on a word $X$, the source symbol is $a$ with probability $c_a(X)/m$, and the replacement symbol is an independent uniform $b\in[n]$.
    If $a=b$, the word is unchanged; this event has probability exactly $1/n$.
    If $a\neq b$, the replacement is legal precisely when the source load can decrease by one and the target load can increase by one.
    Since the loads sum to $m$, the new loads $K_a-1$ (when $K_a\geq1$) and $K_b+1$ (for $b\neq a$) still lie in $\{0,\ldots,m\}$, so the truncation in Definition~\ref{def:specialloadstates} does not affect this condition.
    Lemma~\ref{lem:hard-wall-histogram-law} therefore gives the exact identity
    \begin{align*}
        \Pr[X'\in \calT]
        =\frac1n+\frac1{mn}\sum_{\substack{a,b\in[n]\\a\neq b}}
        \EE\left[
            K_a\one[K_a-1\in B_{\class(a)}]
            \one[K_b+1\in B_{\class(b)}]
            \,\middle|\,\sum_cK_c=m
        \right].
    \end{align*}
    For each $a\neq b$, apply Lemma~\ref{lem:infinitedecoupling} with
    \[
        g(k)=k\one[k-1\in B_{\class(a)}],
        \qquad
        h(\ell)=\one[\ell+1\in B_{\class(b)}].
    \]
    These functions satisfy $0\leq g(k)\leq k$ and $0\leq h(\ell)\leq1$, so
    \begin{align*}
        &\EE\left[
            K_a\one[K_a-1\in B_{\class(a)}]
            \one[K_b+1\in B_{\class(b)}]
            \,\middle|\,\sum_cK_c=m
        \right]\\
        &\qquad=
            \EE[K_a\one[K_a-1\in B_{\class(a)}]]
            \Pr[K_b+1\in B_{\class(b)}]+O(n^{-1/2}).
    \end{align*}
    The error is uniform in $a,b$.  Since $m=\Theta(n)$, summing these errors and including the normalization gives
    \[
        \frac1{mn}\sum_{\substack{a,b\in[n]\\a\neq b}}O(n^{-1/2})
        =O\left(\frac{n(n-1)}{mn}\,n^{-1/2}\right)
        =O(n^{-1/2}).
    \]

    To factor the resulting sum, we insert the missing terms with $a=b$.
    Their total contribution satisfies
    \begin{align*}
        &0\leq\frac1{mn}\sum_{a\in[n]}
            \EE[K_a\one[K_a-1\in B_{\class(a)}]]
            \Pr[K_a+1\in B_{\class(a)}]\\
        &\qquad\leq\frac1{mn}\sum_{a\in[n]}\EE[K_a]\\
        &\qquad=\frac1{mn}\sum_{r\in[R]}|S_r|
            \EE_{K\sim\tpoiss_{\nu,B_r}}[K]
            =\frac{O(n)}{mn}=O(1/n).
    \end{align*}
    Thus the exact identity above becomes
    \begin{align*}
        \Pr[X'\in \calT]
        &=\frac1{mn}\sum_{a,b\in[n]}
            \EE[K_a\one[K_a-1\in B_{\class(a)}]]
            \Pr[K_b+1\in B_{\class(b)}]\\*
        &\qquad{}+\frac1n+O(1/n)+O(n^{-1/2})\\
        &=
        \left(\frac1m\sum_{a\in[n]}\EE[K_a\one[K_a-1\in B_{\class(a)}]]\right)\\*
        &\qquad{}\cdot\left(\frac1n\sum_{b\in[n]}\Pr[K_b+1\in B_{\class(b)}]\right)
        +O(n^{-1/2}).
    \end{align*}

    It remains to identify the two averages with the source and target probabilities.
    First, the class-size and length assumptions give, for each $r\in[R]$,
    \begin{align*}
        \frac{|S_r|}{m}
        &=\frac{\alpha_rn+O(1)}{\eta n+O(1)}
            =\frac{\alpha_r}{\eta}+O(1/n),\\
        \frac{|S_r|}{n}
        &=\frac{\alpha_rn+O(1)}{n}
            =\alpha_r+O(1/n).
    \end{align*}
    For the source average, group symbols by class and then use the definition of the size-biased source law:
    \begin{align*}
        &\frac1m\sum_{a\in[n]}\EE[K_a\one[K_a-1\in B_{\class(a)}]]\\
        &\qquad=\sum_{r\in[R]}\frac{|S_r|}{m}
            \EE_{K\sim\tpoiss_{\nu,B_r}}[K\one[K-1\in B_r]]\\
        &\qquad=\sum_{r\in[R]}\frac{\alpha_r}{\eta}
            \EE_{K\sim\tpoiss_{\nu,B_r}}[K\one[K-1\in B_r]]+O(1/n)\\
        &\qquad=\sum_{r\in[R]}\sum_{k\geq0}\frac{\alpha_r}{\eta}\,k
            \Pr_{K\sim\tpoiss_{\nu,B_r}}[K=k]\one[k-1\in B_r]+O(1/n)\\
        &\qquad=\Pr[K_\src-1\in B_{I_\src}]+O(1/n).
    \end{align*}
    The $O(1/n)$ error follows because there are finitely many classes and
    $\EE[K\one[K-1\in B_r]]\leq\EE[K]=O(1)$ in each class.
    Similarly, the target average is
    \begin{align*}
        &\frac1n\sum_{b\in[n]}\Pr[K_b+1\in B_{\class(b)}]\\
        &\qquad=\sum_{r\in[R]}\frac{|S_r|}{n}
            \Pr_{K\sim\tpoiss_{\nu,B_r}}[K+1\in B_r]\\
        &\qquad=\sum_{r\in[R]}\alpha_r
            \Pr_{K\sim\tpoiss_{\nu,B_r}}[K+1\in B_r]+O(1/n)\\
        &\qquad=\sum_{r\in[R]}\sum_{k\geq0}\alpha_r
            \Pr_{K\sim\tpoiss_{\nu,B_r}}[K=k]\one[k+1\in B_r]+O(1/n)\\
        &\qquad=\Pr[K_\tar+1\in B_{I_\tar}]+O(1/n).
    \end{align*}
    Here each class probability is at most one, so the finitely many class-size errors again sum to $O(1/n)$.

    Substituting these two calculations into the factored expression gives
    \begin{align*}
        \Pr[X'\in \calT]
        &=\left(\Pr[K_\src-1\in B_{I_\src}]+O(1/n)\right)\\*
        &\qquad{}\cdot\left(\Pr[K_\tar+1\in B_{I_\tar}]+O(1/n)\right)
            +O(n^{-1/2})\\
        &=\Pr[K_\src-1\in B_{I_\src}]
            \Pr[K_\tar+1\in B_{I_\tar}]
            +O(1/n)+O(n^{-1/2})\\
        &=\Pr[K_\src-1\in B_{I_\src}]
            \Pr[K_\tar+1\in B_{I_\tar}]+O(n^{-1/2}),
    \end{align*}
    where the second equality uses that both probabilities lie in $[0,1]$.
\end{proof}

\end{document}

%% file: figures/figure-style.tex
\usetikzlibrary{arrows.meta,calc,decorations.pathreplacing,patterns,patterns.meta}
\definecolor{FigBlue}{HTML}{2A78D6}    
\definecolor{FigAqua}{HTML}{1BAF7A}
\definecolor{FigOrange}{HTML}{EB6834}
\definecolor{FigInk}{HTML}{1F1E1C}
\definecolor{FigMuted}{HTML}{6F6D67}
\definecolor{FigLine}{HTML}{C3C2B7}
\definecolor{FigSoft}{HTML}{ECEBE6}

%% file: figures/rm-support-size.tex
\begin{figure}[t]
\centering
\begin{tikzpicture}[x=0.40cm, y=0.40cm, line cap=round, line join=round, font=\footnotesize,
    zeropt/.style={circle, fill=FigLine, inner sep=0pt, minimum size=2.8pt},
    given/.style={circle, draw=FigInk, line width=0.6pt, fill=white, inner sep=0pt,
      minimum size=10.5pt},
    val/.style={given, text=FigBlue, font=\footnotesize\bfseries},
    hl/.style={draw=FigOrange, line width=1.1pt},
    hlnote/.style={text=FigOrange!85!black, font=\scriptsize},
    note/.style={text=FigMuted, font=\scriptsize},
    head/.style={text=FigInk, font=\small\bfseries, anchor=base west, inner sep=0pt},
    sub/.style={text=FigInk, font=\footnotesize\itshape, anchor=base west, inner sep=0pt},
    stmt/.style={text=FigInk, font=\footnotesize, anchor=north west, inner sep=0pt},
    goal/.style={text=FigInk, font=\footnotesize, fill=FigBlue!7, draw=FigBlue!45,
      line width=0.5pt, rounded corners=2pt, inner xsep=5pt, inner ysep=4pt, align=flush left},
    link/.style={-{Latex[length=2.2mm]}, draw=FigInk, line width=0.7pt},
  ]
  \node[text=FigInk, font=\small, anchor=base] at (16.7, 13.4)
    {Example over $\FF_3$: the set $S$ of $11$ points in $\FF_3^4$ is given.};

  \node[head] at (-3.3, 11.75) {Dual formulation};
  \node[sub] at (-3.3, 10.7) {Find a full-support kernel vector};
  \node[stmt] at (-3.3, 10.05) {%
    \begin{tabular}{@{}l@{\ \ }l@{}}
      \textbf{Given} & the evaluation matrix $\matB$, determined by $S$\\
      \textcolor{FigBlue}{\textbf{Find}} & $\yv$ with every entry nonzero and $\matB\yv=\mathbf 0$
    \end{tabular}};

  \node[head] at (18.25, 11.75) {Primal formulation};
  \node[sub] at (18.25, 10.7) {Find a polynomial with prescribed support};
  \node[stmt] at (18.25, 10.05) {%
    \begin{tabular}{@{}l@{\ \ }l@{}}
      \textbf{Given} & the positions $S$ (outlined)\\
      \textcolor{FigBlue}{\textbf{Find}} & $f$ with $\deg f\le 6$, nonzero exactly on $S$
    \end{tabular}};

  \path[fill=FigBlue!7, rounded corners=2pt] (0.6, 5.1) rectangle (11.4, 6.1);
  \path[fill=FigOrange!12, draw=FigOrange, line width=0.8pt, rounded corners=3pt]
    (1.58, -0.42) rectangle (2.42, 6.1);
  \node[hlnote, anchor=north] at (2, -0.65) {column of $(2,2,0,1)$};

  \node[text=FigBlue, anchor=base west, inner sep=0pt] at (0.5, 6.55) {One possible solution $\yv$};
  \foreach \y [count=\j] in {1,1,2,2,2,1,1,2,1,1,1}
    \node[text=FigBlue, font=\footnotesize\bfseries] at (\j, 5.6) {$\y$};

  \node[text=FigInk, anchor=east] at (-1.25, 2) {$\matB={}$};
  \foreach \lab [count=\r from 0] in {1, X_1, X_2, X_3, X_4}
    \node[note, font=\footnotesize, anchor=east] at (0.2, 4-\r) {$\lab$};
  \foreach \e [count=\j] in {1,1,1,1,1,1,1,1,1,1,1}
    \node[text=FigInk] at (\j, 4) {$\e$};
  \foreach \e [count=\j] in {1,2,0,0,2,0,1,2,0,1,2}
    \node[text=FigInk] at (\j, 3) {$\e$};
  \foreach \e [count=\j] in {1,2,0,2,1,1,2,1,0,2,2}
    \node[text=FigInk] at (\j, 2) {$\e$};
  \foreach \e [count=\j] in {0,0,0,0,1,1,1,1,2,2,2}
    \node[text=FigInk] at (\j, 1) {$\e$};
  \foreach \e [count=\j] in {0,1,2,2,0,1,2,2,0,1,1}
    \node[text=FigInk] at (\j, 0) {$\e$};
  \draw[FigInk, line width=0.5pt] (0.75, 4.55) -- (0.5, 4.55) -- (0.5, -0.55) -- (0.75, -0.55);
  \draw[FigInk, line width=0.5pt] (11.25, 4.55) -- (11.5, 4.55) -- (11.5, -0.55) -- (11.25, -0.55);

  \draw[link] (13.0, 3.05) -- (16.4, 3.05);
  \node[note, anchor=south, align=center] at (14.7, 3.2) {extend by zero\\outside $S$};
  \draw[link] (16.4, 2.35) -- (13.0, 2.35);
  \node[note, anchor=north, align=center] at (14.7, 2.2) {restrict\\to $S$};

  \begin{scope}[shift={(19.5, -1.7)}]
    \foreach \t in {2.5, 5.5} {
      \draw[FigLine, line width=0.3pt, dash pattern=on 1pt off 1.5pt] (\t, -0.5) -- (\t, 8.5);
      \draw[FigLine, line width=0.3pt, dash pattern=on 1pt off 1.5pt] (-0.5, \t) -- (8.5, \t);
    }
    \foreach \j in {0,...,8}
      \foreach \k in {0,...,8}
        \node[zeropt] at (\j, \k) {};
    \foreach \j/\k/\y in {4/0/1, 8/1/1, 0/2/2, 2/2/2, 7/3/2, 1/4/1, 5/5/1,
                          7/5/2, 0/6/1, 5/7/1, 8/7/1}
      \node[val] at (\j, \k) {$\y$};
    \node[val, hl] at (8, 1) {$1$};
    \node[hlnote, anchor=west] at (8.5, 1) {$f(2,2,0,1)=1$};

    \foreach \lab [count=\j from 0] in {00,01,02,10,11,12,20,21,22} {
      \node[note] at (\j, -0.9) {$\lab$};
      \node[note, anchor=east] at (-0.55, \j) {$\lab$};
    }
    \node[note] at (4, -1.8) {$x_1x_2$};
    \node[note, anchor=south east] at (-0.5, 8.55) {$x_3x_4$};

    \node[goal, anchor=north west, text width=7.0cm] at (-1.25, -2.6)
      {Choose the nonzero values so that the resulting function has degree at most~$6$.};

    \node[given] at (9.8, 8.25) {};
    \node[anchor=base west, inner sep=0pt] at (10.55, 7.95) {Given support $S$};
    \node[text=FigBlue, font=\footnotesize\bfseries, anchor=base] at (9.8, 6.0) {$1,2$};
    \node[anchor=base west, inner sep=0pt] at (10.55, 6.0) {Computed nonzero};
    \node[anchor=base west, inner sep=0pt] at (10.55, 5.0) {values in $\FF_3$};
    \node[zeropt] at (9.8, 3.8) {};
    \node[anchor=base west, inner sep=0pt] at (10.55, 3.55) {Forced zeros};
    \node[anchor=base west, inner sep=0pt] at (10.55, 2.55) {outside $S$};
  \end{scope}

\end{tikzpicture}
\caption{The selected set $S$ is the input; the displayed nonzero values illustrate one valid
output.  In this example over $\FF_3$ with $v=4$ variables, $S$ has $m=11$ points, and $\matB$
evaluates the $n=5$ monomials $1,X_1,\ldots,X_4$ of degree at most one at these points.  The dual
formulation (left) asks for $\yv$ with no zero entry and $\matB\yv=\mathbf 0$; each entry of $\yv$
is drawn above its column.  The primal formulation (right) asks for a polynomial $f$ of degree at
most $6$ (the degree $v(q-1)-r-1$ of the dual Reed--Muller code) that is nonzero exactly on $S$.  Extending $\yv$ by zero outside $S$, or restricting $f$
to $S$, converts one kind of solution into the other.}
\label{fig:rm-support-size}
\end{figure}

%% file: figures/histogram-csp.tex
\begin{figure}[t]
\centering
\begin{tikzpicture}[line cap=round, line join=round, font=\footnotesize,
    symbox/.style={draw=FigLine, fill=white, line width=0.45pt, inner sep=0pt,
      minimum width=0.4cm, minimum height=0.4cm, text=FigInk},
    hibox/.style={symbox, draw=FigInk, line width=0.9pt},
    gridpt/.style={circle, draw=FigLine, fill=white, line width=0.35pt,
      inner sep=0pt, minimum size=2.8pt},
    wordpt/.style={circle, fill=FigBlue, inner sep=0pt, minimum size=5.5pt},
    hiring/.style={circle, draw=FigInk, line width=0.7pt, inner sep=0pt, minimum size=9.5pt},
    permitted/.style={rectangle, rounded corners=1pt, fill=FigMuted!40, inner sep=0pt,
      minimum width=0.24cm, minimum height=0.2cm},
    countpt/.style={circle, fill=FigBlue, inner sep=0pt, minimum size=4.8pt},
    relabel/.style={-{Latex[length=1.3mm]}, draw=FigMuted!55, line width=0.45pt},
    relabel hi/.style={-{Latex[length=1.8mm]}, draw=FigInk, line width=0.9pt},
    brace/.style={decorate, decoration={brace, amplitude=3pt}, draw=FigMuted, line width=0.5pt},
    note/.style={text=FigMuted, font=\scriptsize},
    rowlabel/.style={text=FigInk, anchor=east},
  ]
  \def\dx{0.5}    
  \def\dy{0.3}    
  \def\yrow{4.25} 
  \def\xrow{5.55} 
  \def\cz{6.9}    
  \def\sep{6.45}  

  \node[note, anchor=east] at (0.2, \xrow+0.53) {coordinate $i$};
  \foreach \i in {1,...,12}
    \node[note] at ({\i*\dx}, \xrow+0.53) {$\i$};

  \node[rowlabel] at (0.2, \xrow) {codeword $\xv\in C$};
  \foreach \v [count=\i] in {8,12,5,0,10,9,10,0,5,12,8,6}
    \node[symbox] (x\i) at ({\i*\dx}, \xrow) {$\v$};

  \node[rowlabel] at (0.2, \yrow) {relabeled word $\yv$};
  \foreach \v [count=\i] in {11,8,12,7,12,11,9,12,8,10,11,12}
    \node[symbox] (y\i) at ({\i*\dx}, \yrow) {$\v$};

  \foreach \i in {1,2,4,5,6,7,8,10,11,12}
    \draw[relabel] (x\i) -- (y\i);
  \foreach \i in {3,9} {
    \node[hibox] at (x\i) {$5$};
    \draw[relabel hi] (x\i) -- (y\i);
  }
  \node[hibox] at (y3) {$12$};
  \node[hibox] at (y9) {$8$};
  \node[text=FigInk, fill=white, inner sep=1pt, anchor=west]
    at ({3*\dx+0.06}, {(\xrow+\yrow)/2}) {$\pi_3$};
  \node[text=FigInk, fill=white, inner sep=1pt, anchor=west]
    at ({9*\dx+0.06}, {(\xrow+\yrow)/2}) {$\pi_9$};
  \node[anchor=west, align=left, text=FigInk] at (6.55, {(\xrow+\yrow)/2})
    {$y_i=\pi_i(x_i)$\\[-1pt]{\scriptsize\color{FigMuted}each coordinate $i$ has its own}\\[-3pt]{\scriptsize\color{FigMuted}permutation $\pi_i$ of $\FF_{13}$}};

  \foreach \xa/\xb in {0.25/6.25, \cz-0.25/\cz+1.45} {
    \foreach \a in {0,...,6}
      \path[fill=FigSoft, rounded corners=1.2pt] (\xa, {\a*\dy-0.12}) rectangle (\xb, {\a*\dy+0.12});
    \foreach \a in {7,...,10}
      \path[fill=FigAqua!14, rounded corners=1.2pt] (\xa, {\a*\dy-0.12}) rectangle (\xb, {\a*\dy+0.12});
    \foreach \a in {11,12}
      \path[fill=FigOrange!14, rounded corners=1.2pt] (\xa, {\a*\dy-0.12}) rectangle (\xb, {\a*\dy+0.12});
  }
  \draw[FigMuted, line width=0.6pt] (\sep, -0.2) -- (\sep, {12*\dy+0.2});

  \foreach \i in {1,...,12}
    \foreach \a in {0,...,12}
      \node[gridpt] at ({\i*\dx}, {\a*\dy}) {};
  \foreach \v [count=\i] in {11,8,12,7,12,11,9,12,8,10,11,12}
    \node[wordpt] at ({\i*\dx}, {\v*\dy}) {};
  \node[hiring] at ({3*\dx}, {12*\dy}) {};
  \node[hiring] at ({9*\dx}, {8*\dy}) {};

  \foreach \a in {0,...,12}
    \node[note, anchor=east] at (0.2, {\a*\dy}) {$\a$};
  \node[note, rotate=90, anchor=south] at (-0.3, {6*\dy}) {symbol $a$};

  \foreach \a in {0,...,6}
    \node[permitted] at (\cz, {\a*\dy}) {};
  \foreach \a in {7,...,10}
    \foreach \k in {1,2}
      \node[permitted] at (\cz+0.3*\k, {\a*\dy}) {};
  \foreach \a in {11,12}
    \foreach \k in {3,4}
      \node[permitted] at (\cz+0.3*\k, {\a*\dy}) {};
  \foreach \a/\c in {0/0, 1/0, 2/0, 3/0, 4/0, 5/0, 6/0, 7/1, 8/2, 9/1, 10/1, 11/3, 12/4}
    \node[countpt] at (\cz+0.3*\c, {\a*\dy}) {};
  \foreach \k in {0,...,4}
    \node[note] at (\cz+0.3*\k, -0.36) {$\k$};
  \node[text=FigInk] at (\cz+0.6, -0.74) {$c_a(\yv)$};

  \draw[brace] (\cz+1.65, {12*\dy+0.1}) -- (\cz+1.65, {11*\dy-0.1});
  \draw[brace] (\cz+1.65, {10*\dy+0.1}) -- (\cz+1.65, {7*\dy-0.1});
  \draw[brace] (\cz+1.65, {6*\dy+0.1}) -- (\cz+1.65, {-0.1});
  \node[anchor=west, align=left] at (\cz+1.85, {11.5*\dy})
    {$\calB_a=\{3,4\}$\\[-1pt]{\scriptsize\color{FigMuted}2 symbols, approximately one sixth}};
  \node[anchor=west, align=left] at (\cz+1.85, {8.5*\dy})
    {$\calB_a=\{1,2\}$\\[-1pt]{\scriptsize\color{FigMuted}4 symbols, approximately one third}};
  \node[anchor=west, align=left] at (\cz+1.85, {3*\dy})
    {$\calB_a=\{0\}$\\[-1pt]{\scriptsize\color{FigMuted}7 symbols, approximately half}};

  \node[countpt] at (1.4, -1.3) {};
  \node[anchor=west, text=FigInk] at (1.56, -1.3) {observed count $c_a(\yv)$};
  \node[permitted] at (5.3, -1.3) {};
  \node[anchor=west, text=FigInk] at (5.5, -1.3) {permitted counts $\calB_a$};
\end{tikzpicture}
\caption{An example over $\FF_{13}$ with twelve coordinates.  The code constraint requires $\xv\in C$,
where $C$ is a Reed--Solomon code with evaluation points $1,\ldots,12$ and dimension at least $3$; this
holds since $x_i=i^2+i+6$ is quadratic in $i$.  Each coordinate $i$ then relabels its symbol by its own
permutation $\pi_i$ of $\FF_{13}$, giving $y_i=\pi_i(x_i)$ without moving any position; the two
highlighted coordinates have the same symbol $5$ but receive different labels.  The histogram
constraints apply to $\yv$, which need not lie in $C$: the symbols $0,\ldots,6$ must not occur, each of
$7,\ldots,10$ must occur once or twice, and each of $11$ and $12$ three or four times.  In this finite
example, the proportions $1/2$, $1/3$ and $1/6$ are rounded to classes of $7$, $4$ and $2$ symbols.  On
the right, each observed count $c_a(\yv)$ lies in its set $\calB_a$ of permitted counts.}
\label{fig:histogram-csp}
\end{figure}

%% file: figures/rm-phase-diagrams.tex
\begin{figure}[t]
\centering
\begin{tikzpicture}[
    pairdqi swatch/.style={fill=FigBlue!22, draw=FigBlue, line width=0.6pt},
    q2 swatch/.style={pattern={Lines[angle=-45, distance=2.4pt, line width=0.55pt]},
      pattern color=FigAqua!85!black, draw=FigAqua!85!black, line width=0.6pt},
    prior swatch/.style={pattern={Lines[angle=45, distance=2.4pt, line width=0.55pt]},
      pattern color=FigOrange!90!black, draw=FigOrange!90!black, line width=0.6pt},
    trivial swatch/.style={fill=FigSoft, draw=FigMuted, line width=0.6pt, dash pattern=on 2pt off 1.5pt},
    trivial label/.style={font=\footnotesize, text=FigMuted},
  ]
  \begin{groupplot}[
      group style={group size=3 by 1, horizontal sep=0.6cm, y descriptions at=edge left},
      scale only axis, width=4.4cm, height=3.9cm,
      xmin=0, ymin=1, ymax=3, ytick={1,1.5,2,2.5,3},
      enlargelimits=false, clip=true,
      xlabel={$\rho$}, ylabel={$c$},
      axis line style={FigMuted}, tick style={FigMuted, line width=0.4pt},
      major tick length=2pt,
      tick label style={font=\scriptsize, text=FigInk},
      label style={font=\footnotesize, text=FigInk},
      title style={font=\small, text=FigInk},
      xlabel style={yshift=2pt}, ylabel style={rotate=-90, xshift=2pt},
      grid=major, grid style={FigLine!60, line width=0.3pt},
      pairdqi fill/.style={fill=FigBlue!22, draw=none},
      q2 fill/.style={draw=none, pattern={Lines[angle=-45, distance=3pt, line width=0.55pt]},
        pattern color=FigAqua!85!black},
      prior fill/.style={draw=none, pattern={Lines[angle=45, distance=3pt, line width=0.55pt]},
        pattern color=FigOrange!90!black},
      trivial fill/.style={fill=FigSoft, draw=none},
      pairdqi edge/.style={draw=FigBlue, line width=0.9pt},
      q2 edge/.style={draw=FigAqua!85!black, line width=0.9pt},
      prior edge/.style={draw=FigOrange!90!black, line width=0.9pt},
      trivial edge/.style={draw=FigMuted, line width=0.7pt, dash pattern=on 3pt off 2pt},
    ]
    \nextgroupplot[title={$q=3$}, xmax=1, xtick={0,0.25,0.5,0.75,1}]
      \addplot[pairdqi fill] table {figures/data/rm-phase-q3-pairdqi.dat} |- (axis cs:0,10) -- cycle;
      \addplot[prior fill] table {figures/data/rm-phase-q3-quadratic.dat} |- (axis cs:0,10) -- cycle;
      \addplot[q2 fill] table {figures/data/rm-phase-q3-q2.dat} |- (axis cs:0,10) -- cycle;
      \addplot[trivial fill] table {figures/data/rm-phase-q3-saturation.dat} |- (axis cs:0,10) -- cycle;
      \addplot[prior edge] table {figures/data/rm-phase-q3-quadratic.dat};
      \addplot[pairdqi edge] table {figures/data/rm-phase-q3-pairdqi.dat};
      \addplot[q2 edge] table {figures/data/rm-phase-q3-q2.dat};
      \addplot[trivial edge] table {figures/data/rm-phase-q3-saturation.dat};
      \node[trivial label] at (axis cs:0.72,2.4) {trivial};

    \nextgroupplot[title={$q=5$}, xmax=2, xtick={0,0.5,1,1.5,2}]
      \addplot[pairdqi fill] table {figures/data/rm-phase-q5-pairdqi.dat} |- (axis cs:0,10) -- cycle;
      \addplot[prior fill] table {figures/data/rm-phase-q5-quadratic.dat} |- (axis cs:0,10) -- cycle;
      \addplot[q2 fill] table {figures/data/rm-phase-q5-q2.dat} |- (axis cs:0,10) -- cycle;
      \addplot[trivial fill] table {figures/data/rm-phase-q5-saturation.dat} |- (axis cs:0,10) -- cycle;
      \addplot[prior edge] table {figures/data/rm-phase-q5-quadratic.dat};
      \addplot[pairdqi edge] table {figures/data/rm-phase-q5-pairdqi.dat};
      \addplot[q2 edge] table {figures/data/rm-phase-q5-q2.dat};
      \addplot[trivial edge] table {figures/data/rm-phase-q5-saturation.dat};
      \node[trivial label] at (axis cs:1.44,2.4) {trivial};

    \nextgroupplot[title={$q=7$}, xmax=3, xtick={0,0.5,1,1.5,2,2.5,3}]
      \addplot[pairdqi fill] table {figures/data/rm-phase-q7-pairdqi.dat} |- (axis cs:0,10) -- cycle;
      \addplot[prior fill] table {figures/data/rm-phase-q7-quadratic.dat} |- (axis cs:0,10) -- cycle;
      \addplot[q2 fill] table {figures/data/rm-phase-q7-q2.dat} |- (axis cs:0,10) -- cycle;
      \addplot[trivial fill] table {figures/data/rm-phase-q7-saturation.dat} |- (axis cs:0,10) -- cycle;
      \addplot[prior edge] table {figures/data/rm-phase-q7-quadratic.dat};
      \addplot[pairdqi edge] table {figures/data/rm-phase-q7-pairdqi.dat};
      \addplot[q2 edge] table {figures/data/rm-phase-q7-q2.dat};
      \addplot[trivial edge] table {figures/data/rm-phase-q7-saturation.dat};
      \node[trivial label] at (axis cs:2.16,2.4) {trivial};
  \end{groupplot}

  \node[anchor=north, font=\footnotesize, text=FigInk, inner sep=0pt, align=left]
    at ($(group c2r1.south)+(0,-0.95cm)$) {%
    \tikz[baseline=-0.6ex]{\path[pairdqi swatch] (0,-0.1) rectangle (0.32,0.1);}\;
      Quantum: two-fold multiplication (Theorem~\ref{thm:rm-pair-dqi})\\[2pt]
    \tikz[baseline=-0.6ex]{\path[q2 swatch] (0,-0.1) rectangle (0.32,0.1);}\;
      Classical: three-fold multiplication (Theorem~\ref{thm:cubecond})\\[2pt]
    \tikz[baseline=-0.6ex]{\path[prior swatch] (0,-0.1) rectangle (0.32,0.1);}\;
      Prior classical algorithms: quadratic scale, $c>2$~\cite{ivanyos2012nil2,ivanyosSantha2015diagonal,imran2024zerosum,kothari2025sis}\\
    \hphantom{\tikz[baseline=-0.6ex]{\path[prior swatch] (0,-0.1) rectangle (0.32,0.1);}}\;
      (benchmark for comparison: their guarantee is for uniformly random $\matB$)\\[2pt]
    \tikz[baseline=-0.6ex]{\path[trivial swatch] (0,-0.1) rectangle (0.32,0.1);}\;
      trivial: $c\cntexp_q(\rho)>1$};
\end{tikzpicture}
\caption{Asymptotic sufficient regimes for finding a dual Reed--Muller codeword whose support is exactly a random set $S$, with
$n=\cnt_q(v,\lfloor\rho v\rfloor)$ and $m=\min(\lfloor n^c\rfloor,q^v)$.
The axes give the degree per variable $\rho$ and the sampling exponent $c$.
Blue shading marks where our quantum algorithm (two-fold multiplication) is guaranteed to succeed, and green hatching where our
classical algorithm (three-fold multiplication) is, with thresholds from Equation~\eqref{eq:rm-overview-thresholds}.
Both guarantees give polynomial-time algorithms succeeding with high probability on a random $S$.
Strictly above the dashed curve $c\cntexp_q(\rho)=1$, the problem is eventually trivial:
$S=\FF_q^v$ and the all-ones vector is a solution (gray).
Orange hatching marks the nontrivial region with $c>2$.
For uniformly random $\matB$, prior classical
algorithms~\cite{ivanyos2012nil2,ivanyosSantha2015diagonal,imran2024zerosum,kothari2025sis}
find a kernel vector with no zero coordinates when $m\geq Kn^2$ for a constant $K$.
Every fixed $c>2$ eventually satisfies this bound; whether $c=2$ does depends on $K$.
That guarantee does not cover the Reed--Muller matrices considered here, so the orange region is a
benchmark for comparison, not a classical guarantee for these instances.
The candidate regimes for quantum advantage are blue with no hatching: covered by our quantum guarantee,
below $c=2$, and covered by neither our three-fold classical guarantee nor the trivial regime.
The figure makes no classical hardness claim.}
\label{fig:rm-phase-diagrams}
\end{figure}